\documentclass[final]{lmcs}
\pdfoutput=1

\usepackage{lastpage}
\lmcsdoi{18}{4}{6}
\lmcsheading{}{\pageref{LastPage}}{}{}%
{Mar.~22,~2022}{Nov.~17,~2022}{}

\pdfoutput=1

\usepackage[utf8]{inputenc}

\usepackage{amssymb}
\usepackage{stmaryrd}
\usepackage{xspace}
\usepackage{mathtools}
\usepackage{booktabs}
\usepackage{pifont}
\usepackage{nicefrac}
\usepackage{subfig}
\usepackage{amsthm}
\usepackage{thmtools}
\usepackage{array} %
\usepackage{tikz}

\newcommand{\takeout}[1]{\empty}

\usepackage[marginclue,footnote]{fixme}
\FXRegisterAuthor{tw}{atw}{TW}%
\FXRegisterAuthor{sm}{asm}{SM}%
\FXRegisterAuthor{ls}{als}{LS}%

\usetikzlibrary{calc,cd,fit,decorations,matrix,arrows,arrows.meta}

\newenvironment{notheorembrackets}{}{}

\usepackage[inline]{showlabels}
\showlabels{bibitem}

\usepackage{seqsplit}
\usepackage{xstring}
\usepackage{xcolor}
\usepackage{lscape}


\takeout{}%

\usepackage{hyperref}
\hypersetup{hidelinks,final}

\usepackage[inline]{enumitem}
\setlist[enumerate,1]{label={(\arabic*)},font=\normalfont,align=left,leftmargin=0pt,labelindent=0pt,listparindent=\parindent,labelwidth=0pt,itemindent=!,topsep=3pt,parsep=0pt,itemsep=3pt,start=1}
\setlist[enumerate,2]{label={(\alph*)},font=\normalfont,labelindent=*,leftmargin=*,start=1}
\setlist[itemize]{labelindent=*,leftmargin=*,topsep=5pt,itemsep=3pt}
\setlist[description]{labelindent=*,leftmargin=*,itemindent=-1 em}

\newcommand{\mybar}[3]{%
  \mathrlap{\hspace{#2}\overline{\scalebox{#1}[1]{\phantom{\ensuremath{#3}}}}}\ensuremath{#3}
}

\newcommand{\barF}{\mybar{0.6}{2.5pt}{F}}

\makeatletter
\@ifpackageloaded{stix}{%
}{%
  \DeclareFontEncoding{LS2}{}{\noaccents@}
  \DeclareFontSubstitution{LS2}{stix}{m}{n}
  \DeclareSymbolFont{stix@largesymbols}{LS2}{stixex}{m}{n}
  \SetSymbolFont{stix@largesymbols}{bold}{LS2}{stixex}{b}{n}
  \DeclareMathDelimiter{\lBrace}{\mathopen} {stix@largesymbols}{"E8}
                                            {stix@largesymbols}{"0E} 
  \DeclareMathDelimiter{\rBrace}{\mathclose}{stix@largesymbols}{"E9}
                                            {stix@largesymbols}{"0F} 
}
\makeatother
\newcommand{\bag}[2][]{\ensuremath{{#1\lBrace #2#1\rBrace}}}

\newcommand{\myker}{\ensuremath{\mathrm{ker}}} 

\newcommand{\Set}{\ensuremath{\mathsf{Set}}\xspace}
\newcommand{\pred}{\ensuremath{\mathsf{pred}}\xspace}
\newcommand{\id}{\ensuremath{\mathsf{id}}\xspace}
\newcommand{\pr}{\ensuremath{\mathsf{pr}}}
\newcommand{\inl}{\ensuremath{\mathsf{in}_{1}}}
\newcommand{\inr}{\ensuremath{\mathsf{in}_{2}}}
\newcommand{\inj}{\ensuremath{\mathsf{in}}}
\newcommand{\fpair}[1]{\ensuremath{\langle #1 \rangle}}
\newcommand{\set}[2][]{\ensuremath{{#1\{#2#1\}}}}

\newcommand{\BoolMonoid}{\ensuremath{{\mathbb{B}}}\xspace}

\newcommand{\grad}[1]{\monmod{#1}}
\newcommand{\gradI}{\monmod{\raisebox{-1pt}{\ensuremath{\scriptstyle I}}}}
\newcommand{\monmod}[1]{\ensuremath{{\langle{\raisebox{1pt}{\ensuremath{\scriptstyle\mathord{=}#1}}}\rangle}}}

\newcommand{\hearts}{\heartsuit}

\newcommand{\Pow}{\ensuremath{\mathcal{P}}}
\newcommand{\CO}{\ensuremath{\mathcal{O}}}
\newcommand{\Powf}{\ensuremath{{\mathcal{P}_\textsf{f}}}}
\newcommand{\Dist}{\ensuremath{{\mathcal{D}}}}
\newcommand{\Bag}{\ensuremath{\mathcal{B}}}
\newcommand{\unzip}{\ensuremath{\mathsf{unzip}}}
\newcommand{\fib}{\ensuremath{\mathord{\textsf{\upshape fib}}}}
\newcommand{\textqt}[1]{\text{`#1'}}
\newcommand{\N}{\ensuremath{\mathbb{N}}}
\newcommand{\R}{\ensuremath{\mathbb{R}}}
\newcommand{\Z}{\ensuremath{\mathbb{Z}}}
\newcommand{\monoto}{\rightarrowtail}
\newcommand{\monodown}{%
\begin{tikzpicture}[baseline=(B)]
  \node[rotate=-90,inner xsep=0pt,inner ysep=2pt]
     (M) {\ensuremath{\rightarrowtail}};
  \coordinate (B) at ([yshift=3pt]M.east);
\end{tikzpicture}}

\newcommand{\Formulae}{\ensuremath{\mathcal{M}}}

\newcommand{\etal}[1][.]{\text{et al.}\xspace}
\newcommand{\arity}[3][]{\ensuremath{\mathord{\raisebox{1pt}{\ensuremath{#1#2}}\mkern-1.5mu/\mkern-1.5mu{\raisebox{-1pt}{\ensuremath{#1#3}}}}}}
\newcommand{\semantics}[2][]{\ensuremath{#1\llbracket #2 #1\rrbracket}}

\newcommand{\fmod}[1]{\ensuremath{{\ulcorner #1\urcorner}}}
\newcommand{\itemref}[2]{\autoref{#1}{\ref{#2}}}
\newcommand{\notni}{\ensuremath{\not\mkern1mu\ni}}

\newcommand{\refAllEdgeCases}{%
  \caseref{edgeS}--\caseref{edgeBS}\xspace%
}

\newcommand{\scissors}{\ding{34}} %
\newcommand{\blackcircle}{\ensuremath{\bullet}}

\newsavebox{\mypullbackcorner}%
\sbox{\mypullbackcorner}{%
  \begin{tikzpicture}
    \draw[-] (0,0) -- (.5em,.5em) -- (0,1em);
  \end{tikzpicture}%
}

\tikzset{shiftarr/.style={
    rounded corners,%
    to path={--([#1]\tikztostart.center) 
      -- ([#1]\tikztotarget.center) \tikztonodes
      -- (\tikztotarget)}, 
  }}

\definecolor{lmcsDarkGrey}{HTML}{2e3436}
\definecolor{lmcsDarkBlue}{HTML}{4F677E}
\definecolor{lmcsLightBlue}{HTML}{B9CDE0}

\tikzset{
  coalgebra/.style={
    block line/.style={
      draw=black!50,
      line width=1.2pt,
    },
    block/.style={
      block line,
      rounded corners=3pt,
      inner sep=1pt,
      minimum height=6mm,
      minimum width=6mm,
    },
    scissors line/.style={
      draw=black!50,
      text=black!50,
      font=\footnotesize,
      line width=0.8pt,
      shorten <= -4pt,
      shorten >= -4pt,
      dotted,
    },
    description/.style={
      fill=white,
      inner sep=1pt,
    },
    state/.style={
      text depth=0pt,
      outer sep=0pt,
      inner sep=4pt,
      text=lmcsDarkGrey,
    },
    transition/.style={
      -{latex},
      line width=0.8pt,
      draw=lmcsDarkBlue,
      every node/.append style={
        text=lmcsDarkBlue,
      },
      preaction = {draw,-,draw=white,line width=4pt,line cap=round},
    },
    path with edges/.style={
      every edge/.append style={transition}
    },
  },
}

\newcommand{\descto}[3][]{\arrow[phantom]{#2}[#1]{\text{\footnotesize{}\begin{tabular}{c}#3\end{tabular}}}}

\usetikzlibrary{decorations.markings}
\usetikzlibrary{decorations.pathreplacing}

\newcolumntype{L}{>{$}l<{$}} %

\theoremstyle{definition}
\newtheorem{notation}[thm]{Notation}
\newtheorem{defn}[thm]{Definition}
\newtheorem{definition}[thm]{Definition}
\newtheorem{expl}[thm]{Example}
\newtheorem{example}[thm]{Example}
\newtheorem*{rem*}{Remark}
\newtheorem{algorithm}[thm]{Algorithm}
\newtheorem{construction}[thm]{Construction}
\theoremstyle{defC}
\newtheorem{algoC}[thm]{Algorithm}
\theoremstyle{plain}
\newtheorem{theorem}[thm]{Theorem}
\newtheorem{proposition}[thm]{Proposition}
\newtheorem{lemma}[thm]{Lemma}
\newtheorem{corollary}[thm]{Corollary}

\newcommand{\xra}[1]{\xrightarrow{~#1~}}

\title[Generic Modal Witnesses for Behavioural Inequivalence]{Quasilinear-time Computation of Generic Modal Witnesses for Behavioural Inequivalence}
\author[T.~Wißmann]{Thorsten Wißmann\lmcsorcid{0000-0001-8993-6486}}[a]
\address{Radboud University, Nijmegen, The Netherlands}
\thanks{Thorsten Wißmann: Funded by the NWO TOP project 612.001.852}
\thanks{Stefan Milius and Lutz Schröder: Funded by the Deutsche Forschungsgemeinschaft (DFG, German Research Foundation) – project number 259234802
}

\author[S.~Milius]{Stefan Milius\lmcsorcid{0000-0002-2021-1644}}[b]
\address{Friedrich-Alexander-Universit\"{a}t Erlangen-N\"{u}rnberg, Germany}

\author[L.~Schröder]{Lutz Schröder\lmcsorcid{0000-0002-3146-5906}}[b] 

\allowdisplaybreaks%
\begin{document}
\maketitle              %
\begin{abstract}
  We provide a generic algorithm for constructing formulae that
  distinguish behaviourally inequivalent states in systems of various
  transition types such as nondeterministic, probabilistic or
  weighted; genericity over the transition type is achieved by working
  with coalgebras for a set functor in the paradigm of universal
  coalgebra. For every behavioural equivalence class in a given
  system, we construct a formula which holds precisely at the states
  in that class. The algorithm instantiates to deterministic finite
  automata, transition systems, labelled Markov chains, and systems of
  many other types. The ambient logic is a modal logic featuring
  modalities that are generically extracted from the functor; these
  modalities can be systematically translated into custom sets of
  modalities in a postprocessing step. The new algorithm builds on an
  existing coalgebraic partition refinement algorithm. It runs in time
  $\mathcal{O}((m+n) \log n)$ on systems with~$n$ states and~$m$ transitions,
  and the same asymptotic bound applies to the dag size of the
  formulae it constructs.  This improves the bounds on run time and
  formula size compared to previous algorithms even for previously
  known specific instances, viz.~transition systems and Markov chains;
  in particular, the best previous bound for transition systems was
  $\mathcal{O}(m n)$.

\end{abstract}

\section{Introduction}
For finite transition systems, the Hennessy-Milner theorem guarantees
that two states are bisimilar if and only if they satisfy the same
modal formulae. Equivalently, this means that whenever two states are
not bisimilar, then one can find a modal formula that holds at one of
the states but not at the other.  Such a formula, usually called a
\emph{distinguishing formula}~\cite{Cleaveland91}, explains the
difference in the behaviour of the two states.  For example, in the
transition system in \autoref{fig:exampleTS}, the formula
$\Box\Diamond\top$ distinguishes the states~$x$ and~$y$; specifically
it is satisfied at~$x$ but not at~$y$.  This gives rise to the
verification task of actually computing distinguishing
formulae. Cleaveland~\cite{Cleaveland91} presents an algorithm that
computes distinguishing formulae for states in a finite transition
system with~$n$ states and~$m$ transitions in time $\CO(m n)$. The
algorithm builds on the Kanellakis-Smolka partition refinement
algorithm~\cite{KanellakisSmolka83,KanellakisS90}, which computes the
bisimilarity relation on a transition system within the same time
bound.

\begin{figure}[t]%
  \begin{minipage}{.45\textwidth}
    \hspace{1mm}
    \centering
  \begin{tikzpicture}[coalgebra,x=1.5cm,baseline=(x.base)]
    \begin{scope}[every node/.append style={
        state,
        label distance=-.5mm,
        outer sep=3pt,
        inner sep=0pt,
        text depth=0pt,
        font=\normalsize,%
      }
      ]
      \node[label=left:$x$] (x)  at (0, 0) {$\bullet$};
      \node (x1)  at (1, 0) {$\bullet$};
      \node[label=right:$y$] (y)  at (3, 0) {$\bullet$};
      \node[] (z)  at (2, 0) {$\bullet$};
    \end{scope}
    \path[path with edges,overlay]
      (x) edge (x1)
      (x) edge[out=65,in=115,looseness=6.5,overlay] (x)
      (x1) edge[out=65,in=115,looseness=6.5,overlay] (x1)
      (x1) edge (z)
      (y) edge (z)
      (y) edge[out=65,in=115,looseness=6.5,overlay] (y)
      ;
  \end{tikzpicture}
    \caption{Example of a transition system}%
    \label{fig:exampleTS}
  \end{minipage}%
  \hfill
  \begin{minipage}{.45\textwidth}
  \centering
  \begin{tikzpicture}[coalgebra,x=1.5cm,baseline=(x.base)]
    \begin{scope}[every node/.append style={
        state,
        label distance=-.5mm,
        outer sep=3pt,
        inner sep=0pt,
        text depth=0pt,
        font=\normalsize,%
      }
      ]
      \node[label=right:$y$] (y)  at (1, 0) {$\bullet$};
      \node[label=left:$x$] (x)  at (-2, 0) {$\bullet$};
      \node[] (z1)  at (0, 0) {$\bullet$};
      \node[] (z2)  at (-1, 0) {$\bullet$};
    \end{scope}
    \path[path with edges,overlay,every node/.append style={description}]
      (y) edge node {1} (z1)
      (x) edge[bend right=10] node {0.5} (z2)
      (z2) edge[bend right=10] node {1} (z1)
      (x) edge[bend left=20] node {0.5} (z1)
      ;
  \end{tikzpicture}
    \caption{Example of a Markov chain}%
    \label{fig:exampleMarkov}
  \end{minipage}%
\end{figure}

Logical characterizations of bisimulation analogous to the
Hennessy-Milner theorem exist for other system types. For instance,
Desharnais \etal~\cite{DesharnaisEA98,desharnaisEA02} characterize
probabilistic bisimulation on (labelled) Markov chains, in the sense
of Larsen and Skou~\cite{LarsenS91} (for each label, every state has
either no successors or a probability distribution on successors).  In
their logic, a formula $\Diamond_{\ge p}\phi$ holds at states that
have a transition probability of at least $p$ to states satisfying
$\phi$. For example, the state $x$ %
in \autoref{fig:exampleMarkov}
satisfies~$\Diamond_{\ge 0.5}\Diamond_{\ge 1}\top$ but $y$ does
not. Desharnais \etal~provide an algorithm that computes
distinguishing formulae for labelled Markov chains in run time
(roughly) $\CO(n^4)$.

\takeout{} %

In the present work, we construct such counterexamples
generically for a variety of system types. We achieve genericity over
the system type by modelling state-based systems as coalgebras for a
set functor in the framework of universal
coalgebra~\cite{Rutten00}. Examples of coalgebras for a set functor include
transition systems, deterministic automata, or weighted systems
(e.g.~Markov chains). Universal coalgebra provides a generic notion of
behavioural equivalence that instantiates to standard notions for
concrete system types, e.g.~bisimilarity (transition systems), language
equivalence (deterministic automata), or probabilistic bisimilarity (Markov
chains).  Moreover, coalgebras come equipped with a generic notion of
modal logic that is parametric in a choice of modalities whose
semantics is constructed so as to guarantee invariance
w.r.t.~behavioural equivalence; under easily checked conditions, such
a \emph{coalgebraic modal logic} in fact characterizes behavioural
equivalence in the same sense as Hennessy-Milner logic characterizes
bisimilarity~\cite{Pattinson04,Schroder08}. Hence, as soon as suitable
modal operators are found, coalgebraic modal formulae serve as
distinguishing formulae.

In a nutshell, the contribution of the present paper is an algorithm
that computes distinguishing formulae for behaviourally inequivalent
states, and in fact \emph{certificates} that uniquely describe
behavioural equivalence classes in a system, in \emph{quasilinear
  time} and in coalgebraic generality. We build on an existing
efficient coalgebraic partition refinement
algorithm~\cite{concurSpecialIssue}, thus achieving run time
$\CO(m\log n)$ on coalgebras with~$n$ states and~$m$ transitions (in a
suitable encoding). The dag size of formulae is also $\CO(m\log n)$
(for tree size, exponential lower bounds are
known~\cite{FigueiraG10}); even for the basic case of transition
systems, we thus improve the previous best bound
$\CO(m n)$~\cite{Cleaveland91} for both run time and formula size. We
systematically extract the requisite modalities from the functor at
hand, requiring binary and nullary modalities in the general case, and
then give a systematic method to translate these generic modal
operators into more customary ones (such as the standard operators of
Hennessy-Milner logic).

We subsequently identify a notion of \emph{cancellative} functor that
allows for additional optimization. E.g.~functors modelling weighted
systems are cancellative if and only if the weights come from a
cancellative monoid, such as $(\Z,+)$, or $(\R,+)$ as used in
probabilistic systems.  For cancellative functors, much simpler
distinguishing formulae can be constructed: the binary modalities can
be replaced by unary ones, and only conjunction is needed in the
propositional base. On labelled Markov chains, this complements the
result that a logic with only conjunction and different unary
modalities (the modalities~$\Diamond_{\ge p}$ mentioned above)
suffices for the construction of distinguishing formulae (but not
certificates)~\cite{desharnaisEA02} (see
also~\cite{Doberkat09}).\twnote{}

\takeout{}%

\subsection*{Related Work} As mentioned above, Cleaveland's algorithm
for labelled transition systems~\cite{Cleaveland91} is based on
Kanellakis and Smolka's partition refinement
algorithm~\cite{KanellakisS90}, while the coalgebraic partition
refinement algorithm we employ~\cite{concurSpecialIssue} is instead
related to the more efficient Paige-Tarjan
algorithm~\cite{PaigeTarjan87}. We do note that in the current paper
we formally cover only unlabelled transition systems; the labelled
case requires an elaboration of compositionality mechanisms in
coalgebraic logic, which is not in the technical focus of the present
work. Details are discussed in \autoref{multisort}. Hopcroft's
automata minimization algorithm~\cite{Hopcroft71} and its
generalization to variable input
alphabets~\cite{Gries1973,Knuutila2001} have quasi-linear run time; in
these algorithms, a word distinguishing two inequivalent states of
interest can be derived directly from a run of the algorithm. König et
al.~\cite{KoenigEA20} extract formulae from winning strategies in a
bisimulation game in coalgebraic generality, under more stringent
restrictions on the functor than we employ here (specifically, they
assume that the functor is \emph{separable by singletons}, which is
stronger than our requirement that the functor is
\emph{zippable}~\cite[Lemma~14]{KoenigEA20}). Their algorithm runs in
$\CO(n^4)$; it does not support negative transition weights.
Characteristic formulae for behavioural equivalence classes taken
across \emph{all} models require the use of fixpoint
logics~\cite{DorschEA18}.  The mentioned algorithm by Desharnais~et
al.~for distinguishing formulae on labelled Markov
processes~\cite[Fig.~4]{desharnaisEA02} %
is based on Cleaveland's. No complexity analysis is made but the
algorithm has four nested loops, so its run time is roughly
$\CO(n^4)$. Bernardo and Miculan~\cite{BernardoMiculan19} provide a
similar algorithm for a logic with only disjunction. There are further
generalizations along other axes, e.g.~to behavioural
preorders~\cite{CelikkanCleaveland95}. The TwoTowers tool set for the
analysis of stochastic process algebras~\cite{BernardoEA98,Bernardo04}
computes distinguishing formulae for inequivalent processes, using
variants of Cleaveland's algorithm. Some approaches construct
alternative forms of certificates for inequivalence, such as Cranen et
al.'s notion of evidence~\cite{cranen_et_al:LIPIcs:2015:5408} or
methods employed on business process models, based on model
differences and event
structures~\cite{Dijkman08,ArmasCervantesEA13,ArmasCervantesEA14}.

In constructive mathematics, \emph{apartness relations} capture
provable difference of elements, and recently, Geuvers and
Jacobs~\cite{GeuversJacobs21} introduced apartness relations as an
inductive notion for the inequality of states in a coalgebra, or in
general, in a state-based system. In active automata learning,
Vaandrager\ \etal~\cite{VGRW22} base their learning algorithm~$L^{\#}$
on an apartness notion for automata. Whenever two states turn out to
be apart, this is \emph{witnessed} by an input word for which the two
states behave differently, and these witnesses are used in the
subsequent learning process.  Instead of words, we construct modal
formulae as universal witnesses for systems of different type. In this
sense, our results may eventually relate to variants of coalgebraic
active automata learning in which words are similarly replaced with
coalgebraic modal formulae~\cite{BarloccoEA19}.

This paper is an extended and revised version of a conference
publi\-ca\-tion~\cite{WissmannEA21}. It contains full proofs as well
as additional material on simplifications that apply in case the
coalgebra functor is cancellative
(\autoref{sec:cancellative}). Moreover, we include a new, elementary
proof of the known fact that the tree size of certificates can be
exponential~\cite{FigueiraG10} in \autoref{app:tree-size}.

\subsection*{Acknowledgements}%
The authors thank the anonymous referees for their helpful comments.

\section{Preliminaries}%
\label{preliminaries}
We first recall some basic notation. We denote by $0=\emptyset$,
$1=\{0\}$, $2 = \{0,1\}$, and $3 = \{0,1,2\}$ the sets representing
the natural numbers $0$, $1$, $2$ and $3$. For every set~$X$, there is
a unique map $!\colon X\to 1$.  We write~$Y^X$ for the set of
functions $X\to Y$, so e.g.~$X^2\cong X\times X$.\twnote{} In particular,~$2^X$ is the set of $2$-valued
\emph{predicates} on~$X$, which is in bijection with the
\emph{powerset}~$\Pow X$ of~$X$, i.e.~the set of all subsets of~$X$;
in this bijection, a subset $A\in\Pow X$ corresponds to its
\emph{characteristic function} $\chi_A\in 2^X$, given by $\chi_A(x)=1$
if $x\in A$, and $\chi(x)=0$ otherwise.  We freely convert between
predicates and subsets; in particular we apply set operations as well
as the subset and elementhood relations to predicates, with the
evident meaning.  We generally indicate injective maps
by~$\monoto$. Given maps $f\colon Z\to X$, $g\colon Z\to Y$, we
write~$\fpair{f,g}$ for the map $Z\to X\times Y$ given by
$\fpair{f,g}(z) = (f(z), g(z))$.  We denote the disjoint union of
sets~$X$,~$Y$ by $X+Y$, with canonical inclusion maps
\[
  \inl\colon X\monoto X+Y\qquad\text{and}\qquad
  \inr\colon Y\monoto X+Y.
\]
More generally, we write $\coprod_{i\in I} X_i$ for the disjoint union
of an $I$-indexed family of sets $(X_i)_{i\in I}$, and
$\inj_i\colon X_i\monoto \coprod_{i\in I} X_i$ for the $i$-th
inclusion map. For a map $f\colon X\to Y$ (not necessarily
surjective), we denote by $\ker(f) \subseteq X\times X$ the
\emph{kernel} of $f$, i.e.~the equivalence relation
\begin{equation}
  \ker(f) := \{ (x,x') \in X\times X \mid f(x) = f(x')\}.%
  \label{eqKer}
\end{equation}

\begin{notation}[Partitions]%
  \label{eqEqClass}
  Given an equivalence relation~$R$ on~$X$, we write $[x]_R$ for the
  equivalence class $\{ x'\in X\mid (x,x')\in R\}$ of $x\in X$. If~$R$
  is the kernel of a map $f$, we simply write $[x]_{f}$ in lieu of
  $[x]_{\ker(f)}$. The partition corresponding to~$R$ is denoted by
  \[
    X/R = \{[x]_R\mid x\in X\}.
  \]
  Note that $[-]_R\colon X\to X/R$ is a surjective map and that
  $R = \ker([-]_R)$.
\end{notation}
\noindent A \emph{signature} is a set $\Sigma$, whose elements are
called \emph{operation symbols}, equipped with a function
$a\colon \Sigma\to \N$ assigning to each operation symbol its
\emph{arity}. We write $\arity{\sigma}{n}\in \Sigma$ for
$\sigma\in \Sigma$ with $a(\sigma) = n$. We will apply the same
terminology and notation to collections of modal operators.

\subsection{Coalgebra} \emph{Universal coalgebra}~\cite{Rutten00}
provides a generic framework for the modelling and analysis of
state-based systems. Its key abstraction is to parametrize notions and
results over the transition type of systems, encapsulated as an
endofunctor on a given base category. %
Instances cover, for example, deterministic automata, labelled (weighted) transition
systems, and Markov chains.
\begin{defn}
  A \emph{set functor} $F\colon \Set\to \Set$ assigns to every set $X$
  a set $FX$ and to every map $f\colon X\to Y$ a map
  $Ff\colon FX\to FY$ such that identity maps and composition are
  preserved: $F\id_X = \id_{FX}$ and $F(g\cdot f) = Fg\cdot Ff$
  whenever the composite $g\cdot f$ is defined.  An
  \emph{$F$-coalgebra} is a pair $(C,c)$ consisting of a set $C$ (the
  \emph{carrier}) and a map $c\colon C\to FC$ (the
  \emph{structure}). When $F$ is clear from the context, we simply
  speak of a \emph{coalgebra}.
\end{defn}
\noindent In a coalgebra $c\colon C\to FC$, we understand the carrier
set~$C$ as consisting of \emph{states}, and the structure~$c$ as
assigning to each state $x\in C$ a structured collection of successor
states, with the structure of collections determined by~$F$. In this
way, the notion of coalgebra subsumes numerous types of state-based
systems, as illustrated next.

\begin{expl}\label{ex:powerset}
  \begin{enumerate}
  \item The \emph{powerset functor} $\Pow$ sends a set~$X$ to its
    powerset $\Pow X$ and a map $f\colon X\to Y$ to the map
    $\Pow f = f[-]\colon \Pow X\to \Pow Y$ that takes direct images. A
    $\Pow$-coalgebra $c\colon C\to \Pow C$ is precisely a transition
    system: It assigns to every state $x\in C$ a set $c(x) \in \Pow C$
    of \emph{successor} states, inducing a transition relation~$\to$
    given by $x\to y$ iff $y\in c(x)$.
    Similarly, the coalgebras for the finite powerset functor $\Powf$
    (with~$\Powf X$ being the set of finite subsets of~$X$) are
    precisely the finitely branching transition systems.

  \item Coalgebras for the functor $FX=2\times X^A$, where $A$ is a
    fixed input alphabet, are deterministic automata (without an
    explicit initial state). Indeed, a coalgebra structure
    $c = \langle f, t\rangle\colon C\to 2\times C^A$ consists of a
    finality predicate $f\colon C \to 2$ and a transition map $C
    \times A \to C$ in curried form $t\colon C\to C^A$.

  \item\label{exSignature} Every signature $\Sigma$ defines a
    \emph{signature functor} that maps a set~$X$ to the set
    \[\textstyle
      F_\Sigma X = \coprod_{\arity[\scriptstyle]{\sigma}{n} \in \Sigma} X^n,
    \]
    whose elements we may understand as flat $\Sigma$-terms
    $\sigma(x_1,\ldots,x_n)$ with variables from~$X$. The action of
    $F_\Sigma$ on maps $f\colon X\to Y$ is then given by
    \[
      F_\Sigma f\colon F_\Sigma X\to F_\Sigma Y
      \qquad
      (F_\Sigma f) (\sigma(x_1,\ldots,x_n)) = \sigma(f(x_1), \ldots,
      f(x_n)).
    \]
    For simplicity, we write $\sigma$ (instead of $\inj_\sigma$) for
    the coproduct injections, and~$\Sigma$ in lieu of $F_\Sigma$ for
    the signature functor. A $\Sigma$-coalgebra is a kind of tree
    automaton: it consists of a set $C$ of states and a transition map
    $c\colon C \to \coprod_{\sigma/n \in \Sigma} C^n$, which
    essentially assigns to each state an operation symbol $\sigma/n
    \in \Sigma$ and $n$-successor states. Hence, every state in a $\Sigma$-coalgebra describes
    a (possibly infinite) $\Sigma$-tree, that is, a rooted and ordered
    tree whose nodes are labelled by operation symbols from $\Sigma$
    such that a node with $n$ successor nodes is labelled by an
    $n$-ary operation symbol.

  \item\label{exMonVal} For a commutative monoid $(M,+,0)$, the \emph{monoid-valued functor}
  $M^{(-)}$~\cite{GummS01} is defined on a set $X$ by
  \begin{equation}
    M^{(X)} := \{ \mu \colon X\to M\mid \text{$\mu(x) = 0$ for all but finitely many
      $x\in X$}\}.%
    \label{eqFinSupport}
  \end{equation}
  For a map $f\colon X\to Y$, the map
  \(M^{(f)}\colon M^{(X)}\to M^{(Y)} \) is defined by
  \begin{equation*}
    (M^{(f)})(\mu)(y) = \textstyle\sum_{x\in X, f(x) = y} \mu(x).
  \end{equation*}
  A coalgebra $c\colon C\to M^{(C)}$ is a finitely branching weighted
  transition system: for $x, x'\in C$, $c(x)(x')\in M$ is the transition weight
  from $x$ to~$x'$. For the Boolean monoid
  $\BoolMonoid = (2,\vee, 0)$, we recover $\Powf =
  \BoolMonoid^{(-)}$. Coalgebras for $\R^{(-)}$, with~$\R$ understood
  as the additive monoid of the reals, are $\R$-weighted transition
  systems. The functor
  \begin{equation*}\textstyle
    \Dist X = \{\mu\in \R_{\ge 0}^{(X)}\mid \sum_{x\in X}\mu(x) = 1\},
  \end{equation*}
  which assigns to a set $X$ the set of all finite probability
  distributions on $X$ (represented as finitely supported probability
  mass functions), is a subfunctor of~$\R^{(-)}$.

\item Functors can be composed; for instance, given a set~$A$ of
  labels, the composite of $\Pow$ and the functor $A\times (-)$ (whose
  action on sets maps a set~$X$ to the set $A\times X$) is the functor
  $FX=\Pow(A\times X)$, whose coalgebras are $A$-labelled transition
  systems. Coalgebras for $(\Dist(-)+1)^A$ have been termed
  \emph{probabilistic transition systems}~\cite{LarsenS91} or
  \emph{labelled Markov chains}~\cite{desharnaisEA02}, and coalgebras
  for $(\Dist((-)+1)+1)^A$ are \emph{partial labelled Markov
    chains}~\cite{desharnaisEA02}.  Coalgebras for
  $SX = \Powf(A\times \Dist X)$ are variously known as \emph{simple
    Segala systems} or \emph{Markov decision processes}.
  \end{enumerate}
\end{expl}

\takeout{}%

We have a canonical notion of \emph{behaviour} on $F$-coalgebras:

\begin{defn}\label{D:coalg-mor}
  An $F$-coalgebra \emph{morphism} $h\colon (C,c)\to (D,d)$ is a map
  $h\colon C\to D$ such that the square below commutes:
  \[
    \begin{tikzcd}%
      |[alias=C]|
      C
      \arrow{r}{c}
      \arrow[swap]{d}{h}
      & FC
      \arrow{d}{Fh}
      \\
      D
      \arrow{r}{d}
      & FD
    \end{tikzcd}
  \]
  States $x,y$ in
  an $F$-coalgebra $(C,c)$ are \emph{behaviourally equivalent}
  (notation: $x\sim y$) if there exists a coalgebra morphism $h$ such that
  $h(x) = h(y)$.
\end{defn}

\noindent Thus, we effectively define the behaviour of a state as
those of its properties that are preserved by coalgebra
morphisms. The notion of behavioural equivalence subsumes standard
branching-time equivalences:
\begin{expl}\label{exCoalg}
  \begin{enumerate}
  \item\label{coalgts} For $F\in\{\Pow,\Powf\}$, behavioural
    equivalence on $F$-coalgebras, i.e.~on transition systems, is
    \emph{bisimilarity} in the usual sense.

  \item For deterministic automata as coalgebras for $FX = 2 \times
    X^A$, two states are behaviourally equivalent iff they accept the same
    formal language.

  \item\label{coalgSignature} For a signature functor $\Sigma$, two
    states of a $\Sigma$-coalgebra are behaviourally equivalent iff
    they describe the same $\Sigma$-tree.

  \item For labelled transition systems as coalgebras for
    $FX = \Pow(A \times X)$, coalgebraic behavioural equivalence
    precisely captures Milner's strong
    bisimilarity; this was shown by Aczel and Mendler~\cite{AczelMendler89}.\twnote{}

  \item For weighted and probabilistic systems, coalgebraic
    behavioural equivalence instantiates to weighted and probabilistic
    bisimilarity, respectively~\cite[Cor.~4.7]{RuttenDV99},~\cite[Thm.~4.2]{BartelsEA04}.

  \end{enumerate}
\end{expl}
\begin{rem}%
  \label{behEqDifferntCoalg}\label{R:trnkovahull}
  \begin{enumerate}
  \item The notion of behavioural equivalence extends
    straightforwardly to states in different coalgebras, as one can
    canonically define the disjoint union of coalgebras:

    Given a pair of $F$-coalgebra $(C,c)$ and $(D,d)$, we
    have a canonical $F$-coalgebra structure on the
    disjoint union $C+D$ of their carriers:
    \[
      C+D \xra{c+d} FC + FD \xra{[F\inl, F\inr]} F(C+D),
    \]
    where $[-,-]$ denotes case distinction on the disjoint components
    $FC$ and $FD$. It is easy to see that the canonical inclusion maps
    $\inl\colon C\to C+D$ and $\inr\colon D\to C+D$ are $F$-coalgebra
    morphisms. We say that states $x \in C$ and $y \in D$ are
    \emph{behaviourally equivalent} if $\inl(x)\sim \inr(y)$ holds in
    $C+D$. This definition coincides with the standard one, according
    to which $x,y$ are behaviourally equivalent if there exist
    coalgebra morphisms $f\colon(C,c)\to(E,e)$ and
    $g\colon (D,d)\to(E,e)$ such that $f(x)=g(y)$. Moreover, the
    extended definition is consistent with \autoref{D:coalg-mor} in
    the sense that states $x,y$ in the coalgebra $(C,c)$ are
    behaviourally equivalent according to \autoref{D:coalg-mor} iff
    $\inl(x) \sim \inr(y)$ in the canonical coalgebra on $C+C$.

  \item\label{R:trnkovahull:2} As shown by Trnkov\'a~\cite{trnkova71},
    we may assume without loss of generality that a set functor~$F$
    preserves injective maps (see also Barr~\cite[Proof of
    Thm.~3.2]{Barr93})
    that is,~$Ff$ is injective whenever~$f$ is.  In fact, for every
    set functor $F$ there exists a set functor~$\barF$ (called the
    \emph{Trnkov\'a hull} of $F$~\cite{AdamekEA12}) that coincides
    with $F$ on nonempty sets and functions and preserves injections.

    Moreover, $\barF$ preserves all finite intersections (pullbacks of
    pairs of injective maps);  it is even the reflection of $F$ in the
    category of set functors preserving finite intersections~\cite[Cor.~VII.2]{AdamekEA12}.

    Since~$\barF$ only differs from~$F$ on the empty set, both
    functors have the same coalgebras and coalgebra morphisms.  All
    functors in \autoref{ex:powerset} already preserve injective maps.
    \twnote{}
  \end{enumerate}
\end{rem}

\subsection{Coalgebraic Logics}%
\label{sec:coalgLogic}

\takeout{}%
We continue with a brief review of basic concepts of coalgebraic modal
logic~\cite{Pattinson03,Schroder08}. Coalgebraic modal logics are
parametric in a functor~$F$ determining the type of systems underlying
the semantics, and additionally in a choice of modalities interpreted
in terms of \emph{predicate liftings}. For now, we use $F = \Pow$ as a
basic example, deferring further examples to \autoref{domainSpecific}.

\subsection*{Syntax} The syntax of coalgebraic modal logic is
parametrized over the choice of a signature $\Lambda$ of \emph{modal
  operators} (with assigned arities). Then, \emph{formulae} $\phi$ are
generated by the following grammar
\[
    \phi_{1},\ldots,\phi_{n} ::= \top
    ~\vert~ \neg \phi_1
    ~\vert~ \phi_1 \wedge \phi_2
    ~\vert~ \hearts(\phi_1,\ldots,\phi_n)
    \qquad
    (\arity{\hearts}{n} \in \Lambda).
\]
\begin{expl}\label{exPowModalities} For $F = \Pow$, one often takes
  $\Lambda=\{\arity{\Diamond}{1}\}$; the induced syntax is that of
  (single-action) Hennessy-Milner logic. As usual, we write $\Box \phi :\equiv
  \neg \Diamond \neg \phi$.
\end{expl}

\subsection*{Semantics} We interpret formulae as sets of states in
$F$-coalgebras. This interpretation arises by assigning
to each modal operator $\arity{\hearts}{n}\in \Lambda$ an $n$-ary
\emph{predicate
  lifting}~$\semantics{\hearts}$~\cite{Pattinson03,Schroder08}, i.e.~a
family of maps $\semantics{\hearts}_X\colon (2^{X})^n \to 2^{FX}$, one
for every set~$X$, such that the \emph{naturality} condition
\begin{equation}\label{eq:naturality}
  Ff^{-1}\big[\semantics{\hearts}_Y(P_1,\ldots,P_n)\big]
  = \semantics{\hearts}_X(f^{-1}[P_1],\ldots,f^{-1}[P_n])
\end{equation}
holds for all $f\colon X\to Y$ and all $P_1,\ldots,P_n\in 2^X$. Thus,
$\semantics{\heartsuit}_X$ lifts~$n$ given predicates on states to a
predicate on structured collections of states.  Categorically
speaking,~$\semantics{\hearts}$ is a natural transformation
$(2^{(-)})^n\to 2^{F^{\mathsf{op}}}$; that is, the naturality square
\[
  \begin{tikzcd}
    (2^Y)^n
    \arrow{r}{\semantics{\heartsuit}_Y}
    \arrow{d}[swap]{(2^{f})^n}
    & 2^{FY}
    \arrow{d}{2^{Ff}}
    \\
    (2^X)^n
    \arrow{r}{\semantics{\heartsuit}_Y}
    & 2^{FX}
  \end{tikzcd}
\]
commutes for $f\colon X\to Y$. Explicitly, $2^{(-)}$ denotes the
contravariant powerset functor, which sends a set~$X$ to the set $2^X$
of $2$-valued predicates on~$X$ (equivalently to the powerset of~$X$),
and a map $f\colon X\to Y$ to the inverse image map
$2^{f}=f^{-1}[-]\colon 2^Y\to 2^X$; writing down the commutativity of
the above square element-wise then yields precisely~\eqref{eq:naturality}. By the Yoneda lemma, one equivalently can
define predicate liftings as
follows:
\begin{lemC}[{\cite[Proposition~43]{Schroder08}}]\label{predLiftYoneda}
  Predicate liftings $(2^X)^n\to 2^{FX}$ of arity $n$ are in
  one-to-one correspondence with subsets of $F(2^n)$. The
  correspondence sends a subset $S\subseteq F(2^n)$ to the predicate
  lifting
  \[
    \lambda_X(P_1,\ldots,P_n) =
    \{ t\in FX \mid F\!\underbrace{\fpair{P_1,\ldots,P_n}}_{X\to 2^n}(t) \in S
    \}
    \hspace{4em}(P_i\colon X \to 2).
  \]
\end{lemC}

Given the above data, the \emph{extension} of a formula $\phi$ in
an~$F$-coalgebra $(C,c)$ is a predicate $\semantics{\phi}_{(C,c)}$, or
just~$\semantics{\phi}$, on~$C$, recursively defined by
  \begin{align*}
    &
    \semantics{\top}_{(C,c)} = C,
    \qquad
    \semantics{\phi\wedge\psi}_{(C,c)} = \semantics{\phi}_{(C,c)} \cap \semantics{\psi}_{(C,c)},
    \qquad
    \semantics{\neg\phi}_{(C,c)} = C\setminus\semantics{\phi}_{(C,c)},
    \\
    &\semantics{\hearts(\phi_1,\ldots,\phi_n)}_{(C,c)}
    = c^{-1}\big[\semantics{\hearts}_C\big(\semantics{\phi_1}_{(C,c)}, \ldots, \semantics{\phi_n}_{(C,c)}\big)\big]
    \qquad\text{for ($\arity{\hearts}{n}\in \Lambda$).}
  \end{align*}
  (Recall that we implicitly convert between predicates and subsets.)
  We say that a state~$x\in C$ \emph{satisfies}~$\phi$ if
  $x\in\semantics{\phi}$. Notice how the clause for modalities says
  that~$x$ satisfies $\hearts(\phi_1,\ldots,\phi_n)$ iff $c(x)$
  satisfies the predicate obtained by lifting the predicates
  $\semantics{\phi_1}, \ldots, \semantics{\phi_n}$ on~$C$ to a
  predicate on $FC$ according to~$\semantics{\hearts}$.

  \begin{expl}\label{exPowPredLift} Over $F=\Pow$, we
    interpret~$\Diamond$ by the predicate lifting
    \[
    \semantics{\Diamond}_X\colon 2^X \to 2^{\Pow X},\quad
    P\mapsto~ \{ K \subseteq X \mid \exists x \in K\colon x\in P\}
                = \{K\subseteq X \mid K\cap P\neq \emptyset\}.%
                \label{eqDiamondLifting}
    \]
  The arising notion of satisfaction over $\Pow$-coalgebras $(C,c)$ is
  precisely the standard one:
  \[
    x\in \semantics{\Diamond \phi}_{(C,c)}
    \qquad\text{iff}\qquad
    \text{$y\in \semantics{\phi}_{(C,c)}$ for some transition $x\to
      y$}.
  \]
\end{expl}

The naturality condition~\eqref{eq:naturality} of predicate liftings
guarantees invariance of the logic under coalgebra morphisms, and
hence under behavioural equivalence:
\begin{proposition}[Adequacy~\cite{Pattinson03,Schroder08}]%
  \label{bisimInvariant}
  Behaviourally equivalent states satisfy the same formulae: $x\sim y$
  implies that for all formulae~$\phi$, we have $x\in \semantics{\phi}$
  iff $y\in \semantics{\phi}$.
\end{proposition}
\noindent In our running example $F=\Pow$, this instantiates to the
well-known fact that modal formulae are bisimulation-invariant, that
is, bisimilar states in transition systems satisfy the same formulae
of Hennessy-Milner logic.

\section{Constructing Distinguishing Formulae}%
\label{sec:main}

A proof method certifying behavioural equivalence of states $x,y$ in a
coalgebra is immediate by definition: One simply needs to exhibit a
coalgebra morphism $h$ such that $h(x) = h(y)$. In fact, for many
system types, it suffices to relate~$x$ and~$y$ by a coalgebraic
\emph{bisimulation} in a suitable sense
(e.g.~\cite{AczelMendler89,Rutten00,GorinSchroeder13,MartiVenema15}), generalizing
the Park-Milner bisimulation principle~\cite{Milner89,Park81}. It is
less obvious how to certify behavioural \emph{inequivalence}
$x\not\sim y$, showing that such a morphism~$h$ does \emph{not}
exist. By \autoref{bisimInvariant}, one option is to exhibit a
(coalgebraic) modal formula~$\phi$ that is satisfied by~$x$ but not
by~$y$. In the case of (image-finite) transition systems, such a
formula is guaranteed to exist by the Hennessy-Milner theorem, which
moreover is known to generalize to
coalgebras~\cite{Pattinson04,Schroder08}. More generally, we consider
separation of \emph{sets} of states by formulae, following
Cleaveland~\cite[Def.~2.4]{Cleaveland91}:
\begin{defn}\label{defDistinguish}
  Let $(C,c)$ be an $F$-coalgebra. A formula $\phi$
  \emph{distinguishes} a set $X\subseteq C$ from a set $Y\subseteq C$
  if $X \subseteq \semantics{\phi}$ and
  $Y \cap \semantics{\phi} = \emptyset$. In case $X = \{x\}$ and
  $Y = \{y\}$, we just say that \emph{$\phi$ distinguishes $x$ from
    $y$}. We say that $\phi$ is a \emph{certificate} of~$X$ if~$\phi$
  distinguishes~$X$ from $C\setminus X$, that is if
  $\semantics{\phi}=X$.
\end{defn}
\noindent Note that~$\phi$ distinguishes~$X$ from~$Y$ iff $\neg\phi$
distinguishes~$Y$ from~$X$. Certificates have also been referred to as
\emph{descriptions}~\cite{FigueiraG10}.  If $\phi$ is a certificate of
a behavioural equivalence class~$[x]_\sim$, then by definition, $\phi$
distinguishes~$x$ from $y$ whenever $x \not\sim y$. To obtain
distinguishing formulae for behaviourally inequivalent states in a
coalgebra, it therefore suffices to construct certificates for all
behavioural equivalence classes, which indeed is what our algorithm
does. Of course, every certificate must be at least as large as a
smallest distinguishing formula. However, already on transition
systems, distinguishing formulae and certificates have the same
asymptotic worst-case size (cf.~\autoref{worstcase}).

A natural approach to computing certificates for behavioural
equivalence classes is to extend algorithms that compute these
equivalence classes. In particular, \emph{partition refinement}
algorithms compute a sequence $C/R_0, C/R_1,\ldots$ of consecutively
finer partitions (i.e.~$R_{i+1}\subseteq R_i$ for every $i \geq 0$) on
the state space, where every \emph{block} $B\in C/R_i$ is a union of
behavioural equivalence classes. Since $C$ is finite, this sequence
stabilizes, and the final partition is precisely
$C/\mathord{\sim}$. Indeed, Cleaveland's algorithm for computing
certificates on (labelled) transition systems~\cite{Cleaveland91}
correspondingly extends Kanellakis and Smolka's partition refinement
algorithm~\cite{KanellakisSmolka83,KanellakisS90}, which runs in
$\CO(m n)$ on systems with $n=|C|$ states and $m$ transitions. Our
generic algorithm will be based on a more efficient partition
refinement algorithm.

\subsection{Paige-Tarjan with Certificates}%
\label{paigeTarjan}
Before we turn to constructing certificates in coalgebraic generality,
we informally recall and extend the Paige-Tarjan
algorithm~\cite{PaigeTarjan87}, which computes the partition modulo
bisimilarity of a given transition system with $n$ states and~$m$
transitions in time $\CO((m+n) \log n)$. We fix a given finite
transition system, viewed as a $\Pow$-coalgebra $c\colon C\to \Pow C$.

The algorithm computes two sequences $(C/P_i)_{i\in \N}$ and
$(C/Q_i)_{i\in \N}$ of partitions of~$C$ (with~$Q_i,P_i$ equivalence
relations), where only the most recent partition is held in memory
and~$i$ indexes the iterations of the main loop.  Throughout the
execution, $C/P_i$ is finer than $C/Q_i$ (that is, $P_i\subseteq Q_i$
for every~$i \geq 0$), and the algorithm terminates when $P_i= Q_i$.
Intuitively,~$P_i$ is `one transition ahead' of $Q_i$: if $Q_i$
distinguishes states~$x$ and~$y$, then~$P_i$ is based on
distinguishing transitions to~$x$ from transitions to~$y$.

Initially, $C/Q_0 := \{C\}$ consists of only one block and $C/P_0$ of
two blocks: the live states and the deadlocks (i.e.~states with no
outgoing transitions). If $P_i\subsetneqq Q_i$, then there is a block
$B\in C/Q_i$ that is the union of at least two blocks in~$C/P_i$.
In such a situation, the algorithm chooses $S \subseteq B$ in $C/P_i$
to have at most half the size of $B$ and then splits the block $B$
into $S$ and $B\setminus S$ in the partition $C/Q_i$:
\[
  C/Q_{i+1} = (C/Q_i \setminus \{B\}) ~\cup~ \{S, B\setminus S\}.
\]
This is correct because every state in $S$ is already known to be
behaviourally inequivalent to every state in $B\setminus S$. By the
definition of bisimilarity, this implies that every block $T\in C/P_i$
with some transition to~$B$ may contain behaviourally inequivalent
states as illustrated in \autoref{fig:parttree}; that is,~$T$ may need
to be split into smaller blocks, as follows:

\newcommand{\caseref}[1]{\ref{#1}}
\begin{enumerate}[({C}1)]
\makeatletter
\renewcommand{\labelenumi}{\theenumi}
\renewcommand{\theenumi}{\text{\bfseries\color{black}(C\arabic{enumi})}}
\renewcommand{\p@enumi}{}
\makeatother
\item\label{edgeS} states in~$T$ with successors in $S$ but not in $B\setminus
  S$ (e.g.~$x_1$ in \autoref{fig:parttree}),
\item\label{edgeBoth} states in~$T$ with successors in $S$ and $B\setminus S$ (e.g.~$x_2$), and
\item\label{edgeBS} states in~$T$ with successors $B\setminus S$ but not in $S$ (e.g.~$x_3$).
\end{enumerate}

\begin{figure}[t] %
  \centering
  \begin{tikzpicture}[
    partitionBlock/.style={
            shape=rectangle,
            rounded corners=2.5mm,
            minimum height=5mm,
            minimum width=5mm,
            draw=black!30,
            line width=1pt,
            inner sep = 2mm,
    }
]
  \begin{scope}[
    coalgebra,
    every node/.append style={
      state,
      inner sep=0pt,
      outer sep=2pt,
      minimum width=1pt,
      minimum height=4pt,
      anchor=center,
    }
    ]
    \begin{scope}[yshift=1.5cm]
      \node (x1) at (-1.2,0) {$x_1$};
      \node (x2) at (0,0) {$x_2$};
      \node (x3) at (1.2,0) {$x_3$};
    \end{scope}
    \node (y1) at (-1.6,0) {$y_1$};
    \node (y2) at (-0.5,0) {$y_2$};
    \node (y3) at (0.5,0) {$y_3$};
    \node (y4) at (1.2,0) {$y_4$};
  \end{scope}
  \node[anchor=west,xshift=2mm] (y5) at (y4.east) {$\ldots$};
  \draw (y1) edge[draw=none] node {$\ldots$} (y2);
  \begin{scope}[
    every node/.append style={
      partitionBlock,
      inner sep=2pt,
    },
    every label/.append style={
      font=\footnotesize,
      anchor=south,
      outer sep=1pt,
      inner sep=1pt,
      minimum height=2mm,
      shape=rectangle,
    },
    ]
    \node[fit=(x1) (x2) (x3)] (x123){};
    \node[fit=(y1) (y5)] (C) {};
    \node[every label,overlay] at ([xshift=-1mm]x123.north east) {\normalsize $T$}; %
    \node[every label] at ([xshift=-1mm]C.north east) {\normalsize $B$};
    \begin{scope}[every node/.append style={minimum width=10mm}, %
      ]
      \newcommand{\nextblockdistance}{8mm} %
      \foreach \name/\nodesource/\anchor/\direction in
      {Peast/x123/east/,
        Qeast/C/east/,
        Pwest/x123/west/-,
        Qwest/C/west/-} {
        \begin{scope}
          \clip ([yshift=1mm]\nodesource.north \anchor)
          rectangle ([yshift=-1mm,xshift=\direction\nextblockdistance]
                     \nodesource.south \anchor);
          \node at ([xshift=\direction\nextblockdistance]\nodesource.\anchor) (\name) {};
        \end{scope}
      }
    \end{scope}
  \end{scope}
  \begin{scope}%
  \draw[thick, decoration={brace,mirror},decorate]
  ([yshift=-2mm]y1.south west) -- node[anchor=north,yshift=-2pt]{\footnotesize $S$} ([yshift=-2mm]y2.south east) ;
  \draw[thick, decoration={brace,mirror},decorate]
  ([yshift=-2mm]y3.south west) -- node[anchor=north,yshift=-2pt]{\footnotesize $B\setminus S$} ([yshift=-2mm]y3.south west -| y5.south east) ;
  \foreach \name in {Peast,Pwest,Qeast,Qwest} {
    \node[draw=none,minimum width=1pt,minimum height=1pt,text height=1pt,inner xsep=1pt]
       at (\name.center) {$\ldots$};
  }
  \end{scope}
  \begin{scope}[text depth=2pt]
  \node[outer sep=5mm,anchor=east] at (Pwest.center -| Qwest.center) (P) {$C/P:$};
  \node[outer sep=5mm,anchor=east] at (Qwest.center) (Q) {$C/Q:$};
  \coordinate (mapXCoordinate) at ([xshift=10mm]Qeast.center);
  \node[outer sep=1pt,anchor=center] at (Peast.center -| mapXCoordinate) (X) {$C$};
  \node[outer sep=1pt,anchor=center] at (Qeast.center -| mapXCoordinate) (PY) {$\Pow C$};
  \end{scope}
  \draw[commutative diagrams/.cd, every arrow, every label] (X) edge node {$c$} (PY);
  \begin{scope}[
    bend angle=10,
    space/.style={
      draw=white,
      line width=4pt,
    },
    edge/.style={
      -{latex},
      shorten <= 1pt,
      line width=0.8pt,
      draw=lmcsDarkBlue,
      preaction={draw,-,line width=1mm,white},
      every node/.append style={
        fill=white,
        shape=rectangle,
        inner sep=1pt,
        anchor=base,
        pos=0.55,
      },
    },
    ]
    \draw[edge,bend right] (x1) to (y1);
    \draw[edge,bend left] (x1) to (y2);
    \draw[edge,bend right] (x2) to (y2);
    \draw[edge,bend left] (x2) to (y3);
    \draw[edge,bend left] (x2) to (y4);
    \draw[edge,bend right] (x3) to (y3);
    \draw[edge] (x3) to (y4);
    \draw[edge,bend left] (x3) to (y5);
  \end{scope}
  \begin{scope}[%
    linestyle/.style={
      dashed,
      draw=black!50,
    },
    scissors/.style={
      text=black!50,
      font=\tiny,
      inner sep=0pt,
      overlay,
    },
    gapstyle/.style={
      draw=white,
      line width=1mm,
    }
    ]
    \foreach \leftnode/\rightnode/\northlen/\southlen in
    {x1/x2/4mm/4mm,x2/x3/4mm/4mm,y2/y3/4mm/4mm} {
      \coordinate (northend) at ($ (\leftnode) !.5! (\rightnode) + (0,\northlen)$);
      \coordinate (southend) at ($ (\leftnode) !.5! (\rightnode) - (0,\southlen)$);
      \node[anchor=east,rotate=-90,scissors] at (northend) {\scissors};
      \draw[gapstyle] (northend) -- (southend);
      \draw[linestyle] (northend) -- (southend);
    }
  \end{scope}
\end{tikzpicture}
  \caption{The refinement step as illustrated in~\cite[Figure 6]{concurSpecialIssue}.}%
  \label{fig:parttree}
\end{figure}

\noindent The partition $C/P_{i+1}$ arises from $C/P_i$ by splitting
all such predecessor blocks~$T$ of $B$ accordingly. The algorithm
terminates as soon as $P_{i+1} = Q_{i+1}$ holds. %
It is straightforward to construct certificates for the blocks arising
during the execution:
\begin{itemize}
\item
  The certificate for the only block $C \in C/Q_0$ is $\top$, and the
  blocks for live states and deadlocks in $C/P_0$ have certificates
  $\Diamond\top$ and~$\neg\Diamond\top$, respectively.

\item %
  In the refinement step, suppose that $\delta, \beta$ are
  certificates of $S\in C/P_i$ and $B\in C/Q_i$, respectively, where
  $S\subsetneqq B$. For every predecessor block $T$ of $B$, the three
  blocks obtained by splitting $T$ are distinguished
  (in the sense of \autoref{defDistinguish}) as follows:
  \begin{equation}
    \text{
    \caseref{edgeS}\quad $\neg\Diamond(\beta\wedge\neg\delta)$,
    \qquad\caseref{edgeBoth}\quad $\Diamond(\delta) \wedge \Diamond(\beta\wedge\neg \delta)$,
    \qquad\caseref{edgeBS}\quad $\neg\Diamond \delta$.
    }%
    \label{edgeCasesFormula}
  \end{equation}
  Of course these formulae only distinguish the states in $T$ from each other
  (e.g.~there may be states in other blocks with transitions to both
  $S$ and~$B$). Hence, given a certificate~$\phi$ of $T$, one obtains
  certificates of the three resulting blocks in $C/P_{i+1}$ via
  conjunction:
  \[
    \text{
    \caseref{edgeS}\quad $\phi\wedge \neg\Diamond(\beta\wedge\neg\delta)$,
    \qquad\caseref{edgeBoth}\quad $\phi \wedge\Diamond(\delta) \wedge \Diamond(\beta\wedge\neg \delta)$,
    \qquad\caseref{edgeBS}\quad $\phi\wedge\neg\Diamond \delta$.
    }
  \]
\end{itemize}
Upon termination, every bisimilarity class $[x]_\sim$ in the
transition system is annotated with a certificate. A key step in the
generic development will be to come up with a coalgebraic
generalization of the formulae for \refAllEdgeCases.

\subsection{Generic Partition Refinement}%
\label{sec:genPartRef}

The Paige-Tarjan algorithm has been adapted to other system types,
e.g.~weighted systems~\cite{ValmariF10}, and it has recently been
generalized to co\-al\-gebras~\cite{DorschEA17,concurSpecialIssue}. A
crucial step in this generalization is to rephrase the case
distinction \refAllEdgeCases in terms of the
functor $\Pow$: Given a predecessor block $T$ in $C/P_i$ for
$S\subsetneqq B\in C/Q_i$, we define the map $\chi_S^B\colon C \to 3$
by
\begin{equation}
\chi_S^B(x) =
\begin{cases}
    2 & \text{if $x \in S$}, \\
    1 & \text{if $x \in B \setminus S$}, \\
    0 & \text{if $x \in C \setminus B$},
  \end{cases}
  \qquad\text{for sets $S\subseteq B\subseteq C$}.%
  \label{eqChi3}
\end{equation}
We then consider the composite
\[
  C\xra{c} \Pow C\xra{\Pow \chi_S^B} \Pow 3.
\]
The three cases \refAllEdgeCases distinguish between the equivalence
classes \([x]_{\Pow \chi_S^B\cdot c}\) for $x\in T$; that is, every
case is a possible value of $t:=\Pow\chi_S^B(c(x)) \in \Pow 3$:
\[
  \text{\caseref{edgeS}~$2\in t \notni 1$,}
  \qquad\qquad
  \text{\caseref{edgeBoth}~$2\in t \ni 1$, and}
  \qquad\qquad
  \text{\caseref{edgeBS}~$2\notin t \ni 1$.}
\]
Since $T$ is a predecessor block of $B$, the `fourth case'
$2\not\in t \not\mkern1mu\ni 1$ is not possible. There is a transition
from $x$ to some state outside $B$ iff $0\in t$. However, because
of the previous refinement steps performed by the algorithm, either
every or no state of $T$ has an edge to~$C\setminus B$ (a
property called \emph{stability}~\cite{PaigeTarjan87}), hence no
distinction on $0\in t$ is necessary.

It is now easy to generalize from transition systems to coalgebras by
simply replacing the functor $\Pow$ with $F$ in the refinement step. We recall the algorithm:

\begin{notheorembrackets}
\begin{algoC}[{\cite[Alg.~4.9, (5.1)]{concurSpecialIssue}}]\label{coalgPT}
  Given a coalgebra $c\colon C\to FC$, put
  \[
  C/Q_0 := \{C\} \qquad\text{and}\qquad
  P_0 := \myker\big(C\xra{c}{FC}\xra{F!} F1\big).
  \]
  Starting at iteration $i=0$, repeat the following while $P_i\neq Q_i$:
  \begin{enumerate}[label={({A}1)},topsep=5pt]
\makeatletter
\renewcommand{\labelenumi}{\theenumi}
\renewcommand{\theenumi}{\text{\bfseries\color{black}(A\arabic{enumi})}}
\renewcommand{\p@enumi}{}
\makeatother
    \item\label{step1} Pick $S\in C/P_i$ and $B\in C/Q_i$ such that $S\subsetneqq B$ and $2\cdot
      |S| \le |B|$

      \vspace*{8pt} %
    \item\label{defQi1} $C/Q_{i+1} := (C/Q_i \setminus \{B\}) \cup
      \{S, B\setminus S\}$
    \item\label{defPi1} $P_{i+1} := P_i \cap \myker \big(
      C \xra{c} FX \xra{F\chi_S^B} F3
      \big)$
  \end{enumerate}
\end{algoC}
\end{notheorembrackets}
\noindent
This algorithm formalizes the intuitive steps from
\autoref{paigeTarjan}. Again, two sequences of partitions~$P_i$,~$Q_i$
are constructed, and $P_i = Q_i$
upon termination. Initially, $Q_0$ identifies all states, and $P_0$
distinguishes states by only their output behaviour.
\begin{expl}
  \begin{enumerate}
  \item For $F=\Pow$ and~$x\in C$, the value $\Pow!(c(x)) \in \Pow 1$
    is~$\emptyset$ if~$x$ is a deadlock, and~$\{1\}$ if $x$ is a live
    state.
  \item For $FX=2\times X^A$, the value
    $F!(c(x))\in F1= 2\times 1^A \cong 2$ indicates whether $x$ is a
    final or non-final state.
  \end{enumerate}
\end{expl}
\noindent In the main loop, blocks $S\in C/P_i$ and $B\in C/Q_i$
witnessing $P_i\subsetneqq Q_i$ are picked, and~$B$ is split into $S$
and $B\setminus S$, like in the Paige-Tarjan algorithm. Note that step~\ref{defQi1} is equivalent to directly defining the equivalence
relation $Q_{i+1}$ as
\[
  Q_{i+1} := Q_i \cap \ker{\chi_S^B}.
\]
A similar intersection of equivalence relations is performed in step~\ref{defPi1}. The intersection splits every block $T\in
C/P_i$ into smaller blocks such that $x,x'\in T$ end up in the same block iff
$F\chi_S^B(c(x)) = F\chi_S^B(c(x'))$, i.e.~$T$ is replaced with
$\{[x]_{F\chi_S^B(c(x))}\mid x\in T\}$.
Again, this corresponds to the distinction of the three
cases~\refAllEdgeCases.
\begin{expl}
  For $FX=2\times X^A$, there are $|F3| = 2\cdot 3^{|A|}$ cases to be
  distinguished, and so every $T\in C/P_i$ is split into at most that
  many blocks in $C/P_{i+1}$.
\end{expl}

The following property of $F$ is needed for correctness~\cite[Ex.~5.11]{concurSpecialIssue}.
\begin{notheorembrackets}
\begin{defiC}[{\cite{concurSpecialIssue}}]\label{defZippable}
  A functor $F$ is \emph{zippable} if the following maps are injective:
  \begin{equation*}
    \fpair{F(A+!), F(!+B)}\colon~ F(A+B) \longrightarrow F(A+1) \times
    F(1+B)
    \qquad
    \text{for all sets $A,B$}.
  \end{equation*}
\end{defiC}
\end{notheorembrackets}
\noindent Intuitively, $t\in F(A+B)$ is a structured collection of
elements from $A$ and $B$. If $F$ is zippable, then~$t$ is uniquely
determined by the two structured collections in $F(A+1)$ and $F(1+B)$
obtained by identifying all $B$- and all $A$-elements, respectively,
with $0\in 1$.
\begin{expl}
  The functor $FX=X\times X$ is zippable:
  $t=(\inl(a),\inr(b)) \in (A+B)^2$ is uniquely determined by
  $(\inl(a),\inr(0)) \in (A+1)^2$ and $(\inl(0),\inr(b)) \in (1+B)^2$,
  and similarly for the three other cases of~$t$.
\end{expl}

\noindent
In fact, all signature functors as well as $\Pow$ and all
monoid-valued functors (see \itemref{ex:powerset}{exMonVal}) are zippable. Moreover, the class of zippable
functors is closed under products, coproducts, and subfunctors but not
under composition, e.g.~$\Pow\Pow$ is not
zippable~\cite{concurSpecialIssue}.

The intuitive reason why $\Pow\Pow$ is not zippable is that $\Pow\Pow$ has two
\textqt{unordered levels}, and so we can not uniquely reconstruct $t\in
\Pow\Pow(A+B)$ given only the restrictions to $\Pow\Pow(A+1)$ and
$\Pow\Pow(1+B)$. In the following example, we take distinct $a_1,a_2\in A$ and
$b_1,b_2\in B$, and omit some of the coproduct injections for the sake of brevity:
\[
  \begin{tikzcd}[row sep=4mm]
  \\
  \Pow\Pow(A+B)
  \arrow{r}[yshift=2mm]{\fpair{\Pow\Pow(A+!), \Pow\Pow(!+B)}}
  & \Pow\Pow(A+1) \times \Pow\Pow(1+B)
  & \Pow\Pow(A+B)
  \arrow{l}[swap,yshift=2mm]{\fpair{\Pow\Pow(A+!), \Pow\Pow(!+B)}}
  \\
  \set[\big]{\set{a_1,b_1}, \set{a_2,b_2}}
  \arrow[mapsto]{r}
  \arrow[draw=none]{u}[anchor=center,sloped]{\in}
  &
  \begin{array}{r}
  \big(
  \set[\big]{\set{a_1,\inr(0)}, \set{a_2,\inr(0)}}, \\
  \set[\big]{\set{\inl(0),b_1}, \set{\inl(0), b_2}}
  \big)
  \end{array}
  \arrow[draw=none]{u}[anchor=center,sloped]{\in}
  &
  \set[\big]{\set{a_1,b_2}, \set{a_2,b_1}}
  \arrow[mapsto]{l}
  \arrow[draw=none]{u}[anchor=center,sloped]{\in}
  \end{tikzcd}
\]
Both the left-hand and the right-hand set of sets yield the same terms when
restricting to $A+1$ and $1+B$ separately, showing that the map in
\autoref{defZippable} is not injective for $\Pow\Pow$. This example extends to
a coalgebra for which partition refinement based on characteristic maps
$\chi_S^B$ would compute wrong results~\cite[Ex.~5.11]{concurSpecialIssue}.

\begin{rem}\label{multisort} To apply the algorithm to coalgebras for
  composites $FG$ of zippable functors, e.g.~$\Pow(A\times (-))$,
  there is a reduction~\cite[Section~8]{concurSpecialIssue} that
  embeds every $FG$-coalgebra into a coalgebra for the zippable
  functor $(F+G)(X) := FX + GX$. This reduction preserves and reflects
  behavioural equivalence, but introduces an intermediate state for
  every transition. The reduction factors through an encoding of
  composite functors via multisorted
  coalgebra~\cite{SchroderPattinson11}, in which, e.g., a composite
  functor $FG$ would be represented in a setting with two sorts $1,2$
  as a pair of functors~$\hat F$,~$\hat G$, one going from sort~$1$ to
  sort~$2$ and one going the other way around. The multisorted
  framework comes with a corresponding multisorted coalgebraic modal
  logic. There are conversion functors between multisorted coalgebras
  (e.g.\ for a multisorted functor made up of~$\hat F$ and~$\hat G$)
  and single-sorted coalgebras (e.g.\ for $FG$) which, in the end,
  guarantee compositionality (w.r.t.\ functor composition, including
  composition with multi-argument functors such as binary sum) of most
  semantic and algorithmic properties of coalgebraic modal logics,
  including the Hennessy-Milner property
  (see~\cite{SchroderPattinson11} for details).

  In principle, these results imply in particular that the algorithms
  and complexity results developed in the present paper are
  compositional w.r.t.\ functor composition. Establishing this
  formally will require transferring the framework of multisorted
  coalgebraic modal logic along the above-mentioned translation from
  multisorted coalgebras to single-sorted coalgebras for sums of
  functors. To keep the paper focused, we refrain from carrying this
  out in the present paper. We do note that this implies that we do
  not, at the moment, cover labelled transition systems, i.e.\ coalgebras
  for the composite functor $\Pow\circ(A\times(-))$, in full
  formality.
\end{rem}

\begin{notheorembrackets}
\begin{thmC}[{\cite[Thm. 4.20, 5.20]{concurSpecialIssue}}]
  On a finite coalgebra $(C,c)$ for a zippable functor,
  \autoref{coalgPT} terminates after $i\le |C|$ loop iterations, and
  the resulting partition identifies precisely the behaviourally
  equivalent states ($P_i = \mathord{\sim}$).
\end{thmC}
\end{notheorembrackets}
In the correctness proof, the zippability is used to show that it is sufficient
to incrementally refine the partition using the characteristic map $\chi_S^B$
under the functor $F$ in step~\ref{defPi1}. When constructing certificates in
the following, this refinement turns into a logical conjunction and the
characteristic map under the functor turns into a modal operator.

\subsection{Generic Modal Operators}\label{genericModalOp} The extended
Paige-Tarjan algorithm (\autoref{paigeTarjan}) constructs a
distinguishing formula according to the three cases
\refAllEdgeCases. In the coalgebraic \autoref{coalgPT}, these cases
correspond to elements of~$F3$, which determine in which block an
element of a predecessor block~$T$ ends up. Indeed, the elements
of~$F3$ will also serve as generic modalities in characteristic
formulae for blocks of states, essentially by the equivalence
between~\mbox{$n$-ary} predicate liftings and (in this case,
singleton) subsets of $F(2^n)$ (\autoref{predLiftYoneda}); such
singletons are also known as \emph{tests}~\cite{Klin05}.
\begin{defn}\label{defF3Mod}
  The signature of \emph{$F3$-modalities} for a functor $F$ is
  \[
    \Lambda = \{ \arity{\fmod{t}}{2} \mid t\in F3 \};
  \]
  that is, we write $\fmod{t}$ for the syntactic representation of a binary
  modality for every $t\in F3$. The interpretation of $\fmod{t}$ for $F3$ is
  given by
  \[
    \semantics{\fmod{t}}\colon
    (2^X)^2 \to 2^{FX},
    \qquad
    \semantics{\fmod{t}}(S,B) = \{t'\in FX\mid
    F\chi_{S\cap B}^B(t') = t
    \}.
  \]
\end{defn}
\begin{lemma}
  The above maps $\semantics{\fmod{t}}\colon (2^X)^2 \to 2^{FX}$ form
  a binary predicate lifting.
\end{lemma}
\begin{proof}
  There is a canonical quotient $q\colon 2^2\to 3$
  given by
  \[
    q(1,1) = 2
    \qquad
    q(0,1) = 1
    \qquad
    q(0,0) = 0
    \qquad
    q(1,0) = 0.
  \]
  The map $q$ satisfies
  \[
    \chi_{S\cap B}^B = \big(
    X\xra{\fpair{\chi_S,\chi_B}}
    2\times 2\cong 2^2
    \xra{~q~} 3
    \big).
  \]
  For a fixed $t\in F3$, define a predicate lifting via the subset
  \[
    (Fq)^{-1}[\{t\}] \quad\subseteq\quad F(2^2)
  \]
  By \autoref{predLiftYoneda}, the corresponding predicate lifting is given by:
  \begin{align*}
    \semantics{\fmod{t}}(S,B) &= \{
    t'\in FX
    \mid
    F\fpair{\chi_S,\chi_B}(t') \in (Fq)^{-1}[\{t\}]
    \}
    \\&
    = \{
    t'\in FX
    \mid
    Fq(F\fpair{\chi_S,\chi_B}(t')) \in \{t\}
    \}
    \\&
    = \{
    t'\in FX
    \mid
    F(q\cdot \fpair{\chi_S,\chi_B})(t') = t
    \}
    \\&
    = \{
    t'\in FX
    \mid
    F\chi_{S\cap B}^B(t') = t
    \}
    \tag*{\qedhere}
  \end{align*}
\end{proof}
\noindent The intended use of $\fmod{t}$ is as follows: Suppose a
block~$B$ is split into subblocks $S\subseteq B$ and $B\setminus S$,
with certificates $\delta$ and $\beta$ for~$S$ and~$B$, respectively;
that is, $\semantics{\delta} = S$ and $\semantics{\beta}= B$. As in
\autoref{fig:parttree}, we then split every predecessor block $T$ of
$B$ into smaller parts, each of which is uniquely characterized by the
formula $\fmod{t}(\delta,\beta)$ for some $t\in F3$.

\begin{example}\label{examplePowF3Mod}
  For $F=\Pow$, the formula $\fmod{\set{0,2}}(\delta,\beta)$ is equivalent to
  \[
    \overbrace{\Diamond \neg \beta}^{\text{\textqt{0}}}
    \wedge
    \neg \overbrace{\Diamond (\beta\wedge\neg\delta)}^{\text{\textqt{1}}}
    \wedge
    \overbrace{\Diamond(\delta\wedge\beta)}^{\text{\textqt{2}}}.
  \]
\end{example}

\begin{lemma}\label{lemF3CoalgSem} Given an $F$-coalgebra $(C,c)$, a state
  $x \in C$, and formulae $\delta$ and $\beta$ such that
  $\semantics{\delta}\subseteq \semantics{\beta}\subseteq C$, we
  have
  \[
    x \in \semantics{\fmod{t}(\delta,\beta)}
    \quad\Longleftrightarrow\quad
    F\chi_{\semantics{\delta}}^{\semantics{\beta}}(c(x)) = t.
  \]
\end{lemma}

\begin{proof}
  This follows directly from \autoref{defF3Mod} applied to $S:=\semantics{\phi_S}$ and
  $B:=\semantics{\phi_B}$, using that $S\cap B = S$:
  \begin{align*}
    \semantics{\fmod{t}(\phi_S,\phi_B)}
    &=
      c^{-1}[\semantics{\fmod{t}}_C(\semantics{\phi_S},\semantics{\phi_B})]
    \\
    &= c^{-1}[\semantics{\fmod{t}}_C(S,B)]
    \\
    & = c^{-1}[\{t'\in FC \mid
      F\chi_{S\cap B}^B(t') = t\}]
    \\
    &= \{x \in C \mid F\chi_{S}^B(c(x)) = t\}.
      \tag*{\qedhere}
  \end{align*}
\end{proof}

In the initial partition $C/P_0$ on a transition system $(C,c)$, we
used the formulae $\Diamond\top$ and~$\neg\Diamond\top$ to distinguish
live states and deadlocks. In general, we can similarly describe the
initial partition using modalities induced by elements of~$F1$:

\begin{notation}\label{notationF1Mod} Define the injective map
  $j_1\colon 1\monoto 3$ by $j_1(0) = 2$. Then the injection
  $Fj_1\colon F1\monoto F3$ provides a way to interpret elements
  $t\in F1$ as nullary modalities $\fmod{t}$:
  \[
    \fmod{t} := \fmod{Fj_1(t)}(\top,\top)
    \qquad\text{for $t\in F1$.}
  \]
  (Alternatively, we could introduce $\fmod{t}$ directly as a nullary
  modality.)
\end{notation}
\begin{lemma}\label{lemF1CoalgSem}
  Given a coalgebra $c\colon C\to FC$, a state $x\in C$, and $t\in F1$, we have
  \[
    x\in \semantics{\fmod{t}} ~\Longleftrightarrow~ F!(c(x)) = t
  \]
\end{lemma}
\begin{proof}
  Note that for $\chi_C^C\colon C\to 3$, we have
  $\chi_C^C = (C \xra{!} 1 \xra{j_1} 3)$ where $j_1(0) = 2$.
  \begin{align*}
    \semantics{\fmod{t}}
    &= \semantics{\fmod{Fj_1(t)}(\top,\top)}
    &\text{(\autoref{notationF1Mod})}
    \\
    & = \{x \in C \mid F\chi_{C}^{C}(c(x)) = Fj_1(t)\}
    &\text{(\autoref{lemF3CoalgSem}, $\semantics{\top}=C$)}
    \\
    & = \{x \in C \mid Fj_1(F!(c(x))) = Fj_1(t)\}
    &\text{($\chi_C^C =j_1 \cdot !$)} 
    \\
    & = \{x \in C \mid F!(c(x)) = t\}
    &\text{($Fj_1$ injective)}
  \end{align*}
  In the last step we use our running assumption that, w.l.o.g.,~$F$
  preserves injective maps
  (\autoref{R:trnkovahull}\ref{R:trnkovahull:2}).
\end{proof}

\subsection{Algorithmic Construction of Certificates}%
\label{certConstruct}
The $F3$-modalities introduced above (\autoref{defF3Mod}) induce an
instance of coalgebraic modal logic (\autoref{sec:coalgLogic}). We
refer to coalgebraic modal formulae employing the $F3$-modalities as
\emph{$F3$-modal formulae}, and write~$\Formulae$ for the set of
$F3$-modal formulae. As in the extended Paige-Tarjan algorithm
(\autoref{paigeTarjan}), we annotate every block arising during the
execution of \autoref{coalgPT} with a certificate in the shape of an
$F3$-modal formula.  Annotating blocks with formulae means that we
construct maps
\[
  \beta_i\colon C/Q_i \to \Formulae
  \qquad\text{and}\qquad
  \delta_i\colon C/P_i \to \Formulae
  \qquad
  \text{for $i \in \N$}.
\]
As in \autoref{coalgPT}, $i$ indexes the loop iterations. For blocks
$B, S$ in the respective par\-ti\-tion, we denote by $\beta_i(B)$ and
$\delta_i(S)$ the corresponding certificates. We shall prove in
\autoref{algoCertsCorrect} further below that the following
invariants hold, which immediately imply correctness:
\begin{equation}\label{eqCertCorrect}
  \forall B\in X/Q_i\colon
  \semantics{\beta_i(B)} = B
  \qquad\text{and}\qquad
  \forall S\in X/P_i\colon
  \semantics{\delta_i(S)}
  = S,
  \qquad\text{for every $i$}.
\end{equation}
We construct $\beta_i(B)$ and $\delta_i(S)$ iteratively, using certificates for the blocks
$S\subsetneqq B$ at every iteration:

\begin{algorithm}\label{algoCerts}
  We extend \autoref{coalgPT} as follows. We add initializations
  \[
    \beta_0(\{C\}) := \top \qquad\text{and}\qquad \delta_0([x]_{P_0})
    := \fmod{F!(c(x))}\quad\text{for every $[x]_{P_0} \in C/P_0$.}
  \]
  In the $i$-th iteration, we add the following assignments to steps~\ref{defQi1} and~\ref{defPi1}, respectively:\medskip
  \begin{enumerate}[({A$\!'$}1)]
\makeatletter
\renewcommand{\labelenumi}{\theenumi}
\renewcommand{\theenumi}{\text{\bfseries\color{black}(A$\!$'\arabic{enumi})}}
\renewcommand{\p@enumi}{}
\makeatother
  \refstepcounter{enumi}
  \item\label{defBetai1} $\mathrlap{\beta_{i+1}(D)}\phantom{\delta_{i+1}([x]_{P_{i+1}})} = \begin{cases}
      \delta_{i}(S) & \text{if }D = S \\
      \beta_{i}(B)\wedge \neg \delta_{i}(S) & \text{if }D = B\setminus S \\
      \beta_{i}(D) & \text{if }D \in C/Q_i \\
      \end{cases}$\medskip
    \item\label{defDeltai1} $\delta_{i+1}([x]_{P_{i+1}}) = \begin{cases}
        \delta_{i}([x]_{P_{i}}) &\text{if }[x]_{P_{i+1}} = [x]_{P_i} \\
        \delta_{i}([x]_{P_{i}}) \wedge \fmod{F\chi_S^B(c(x))}(\delta_i(S),\beta_i(B))
        &\text{otherwise.}\\
        \end{cases}$\medskip
  \end{enumerate}
  Upon termination, return $\delta_i$.
\end{algorithm}
\noindent Like in \autoref{paigeTarjan}, the only block of $C/Q_0$ has
$\beta_0(\{C\}) = \top$ as a certificate. The partition $C/P_0$
distinguishes by the `output' $F!(c(x))\in F1$ (e.g.~final
vs.~non-final states of an automaton), and by \autoref{lemF1CoalgSem},
the certificate of~$[x]_{P_0}$ specifies precisely this output; in
particular, $\delta_0([x]_{P_0})$ is well-defined.

In the $i$-th iteration of the main loop, we have certificates
$\delta_{i}(S)$ and $\beta_i(B)$ for $S\subsetneqq B$ in
step~\ref{step1} satisfying~\eqref{eqCertCorrect} available from the
previous iterations.  In~\ref{defBetai1}, the Boolean connectives
describe how $B$ is split into $S$ and $B\setminus S$. In~\ref{defDeltai1}, new certificates are constructed for every
predecessor block $T\in C/P_i$ that is refined. If $T$ does not
change, then neither does its certificate. Otherwise, the block
$T = [x]_{P_i}$ is split into the blocks $[x]_{F\chi_S^B(c(x))}$ for
$x \in T$ in step~\ref{defPi1}, which is reflected by the $F3$ modality
$\fmod{F\chi_S^B(c(x))}$ as per \autoref{lemF3CoalgSem}.
\begin{theorem}\label{algoCertsCorrect} For every zippable functor
  $F$, \autoref{algoCerts} is correct: the invariants
  in~\eqref{eqCertCorrect} hold.  Thus, upon termination $\delta_i$
  assigns certificates to each block of $C/\mathord{\sim} = C/P_i$.
\end{theorem}
\begin{proof}
  \begin{enumerate}
  \item We first observe that given $x\in C$,
    $S\subseteq B \subseteq C$, and certificates $\phi_S$ and~$\phi_B$
    of $S$ and $B$, respectively, we have:
    \begin{equation}\label{eqF3ModBS}
      \begin{aligned}
        \semantics{\fmod{F\chi_{S}^B(c(x))}(\phi_S,\phi_B)}
        &=%
        \{x' \in C\mid F\chi_{\semantics{\phi_S}}^{\semantics{\phi_B}}(c(x')) =
        F\chi_{S}^B(c(x)) \}
        \\
        &= [x]_{F\chi_S^B(c(x))},
      \end{aligned}
    \end{equation}
    where the first equation uses \autoref{lemF3CoalgSem} and the
    second one holds because $\semantics{\phi_B} = B$ and
    $\semantics{\phi_S} = S$.
  \item We verify~\eqref{eqCertCorrect} by induction on $i$.
  \begin{itemize}
  \item In the base case $i = 0$, we have
    $\semantics{\beta_0(\{C\})} = \semantics{\top} =
    C$ for the only block in $X/Q_0$. Since $P_0 = \ker (F!\cdot
    c)$,
    $\delta_0$ is well-defined, and by \autoref{lemF1CoalgSem} we have
    \[
      \semantics{\delta_0([x]_{P_0})}
      = \semantics{\fmod{F!(c(x))}}
      = \{y\in C \mid F!(c(x)) = F!(c(y))\}
      = [x]_{P_0}.
    \]
  \item The inductive hypothesis states that
    \[
      \semantics{\delta_i(S)} = S
      \qquad\text{and}\qquad
      \semantics{\beta_i(B)} = B.\tag*{\text{(IH)}}
    \]
    We prove that $\beta_{i+1}$ is correct:
    \begin{align*}
      & \semantics{\beta_{i+1}([x]_{Q_{i+1}})}
      \\ &=
      \begin{cases}
        \semantics{\delta_{i}(S)} & \text{if }[x]_{Q_{i+1}} = S
        \text{, hence }S = [x]_{P_i} \\
        \semantics{\beta_{i}(B)} ~\cap ~C\setminus \semantics{\delta_{i}(S)} & \text{if }[x]_{Q_{i+1}} = B\setminus S
        \text{, hence }B = [x]_{Q_i}
 \\
        \semantics{\beta_{i}([x]_{Q_i})} & \text{if }[x]_{Q_{i+1}} \in C/Q_i \\
      \end{cases}
      \\
    &\overset{\mathclap{\text{(IH)}}}{=}
      \begin{cases}
        S & \text{if }[x]_{Q_{i+1}} = S\\
        B ~\cap ~C\setminus S & \text{if }[x]_{Q_{i+1}} = B\setminus S \\
        [x]_{Q_i} & \text{if }[x]_{Q_{i+1}} \in C/Q_i \\
      \end{cases}
      \\ &=
      \begin{cases}
        [x]_{Q_{i+1}} & \text{if }[x]_{Q_{i+1}} = S = [x]_{P_i}\\
        [x]_{Q_{i+1}}& \text{if }[x]_{Q_{i+1}} = B\setminus S
        \qquad\text{(since $B \cap C\setminus S = B \setminus S$)}   \\
        [x]_{Q_{i+1}} & \text{if }[x]_{Q_{i+1}} \in C/Q_i
        \qquad\text{(since $[x]_{Q_i}$ is not split)}\\
      \end{cases}
      \\ &=
      [x]_{Q_{i+1}}.
    \end{align*}
    For $\delta_{i+1}$, we compute as follows:
    \begin{align*}
      & \semantics{\delta_{i+1}([x]_{P_{i+1}})}
      \\
      &=
                                              \qquad
          \begin{cases}
          \semantics{\delta_{i}([x]_{P_{i}})} &\text{if }[x]_{P_{i+1}} = [x]_{P_i} \\
          \semantics{\delta_{i}([x]_{P_{i}})}
          \cap \semantics{\fmod{F\chi_{S}^B(c(x))}(\delta_i(S), \beta_i(B))}
          &\text{otherwise}
          \end{cases}
            \\
        &\overset{\mathclap{\text{(IH) \&~\eqref{eqF3ModBS}}}}{=}
          \qquad
          \begin{cases}
          [x]_{P_{i}} &\text{if }[x]_{P_{i+1}} = [x]_{P_i} \\
          [x]_{P_{i}}
          \cap [x]_{F\chi_S^B(c(x))}
          &\text{otherwise}
          \end{cases}
            \\
        &\overset{\mathclap{\text{def. }P_{i+1}}}{=}
          \qquad
          \begin{cases}
          [x]_{P_{i+1}} &\text{if }[x]_{P_{i+1}} = [x]_{P_i} \\
          [x]_{P_{i+1}}
          &\text{otherwise}
          \end{cases}
            \\ &= \quad[x]_{P_{i+1}}
                 \tag*{\qedhere}
    \end{align*}
  \end{itemize}
\end{enumerate}
\end{proof}

\noindent
The assumption of zippability is used in the correctness of the underlying
partition refinement algorithm. But for the construction of certificates, this
assumption translates into the ability to describe certificates as only a
conjunction of $F3$-modalities in~\ref{defDeltai1}.
As a consequence, we obtain a Hennessy-Milner-type property of
our $F3$-modal formulae:
\begin{corollary}%
  \label{hennessyMilner}
  For zippable $F$, states $x,y$ in a finite $F$-coalgebra are
  behaviourally equivalent iff they agree on all $F3$-modal formulae.
\end{corollary}
\begin{construction}\label{extractDistinguish} Given a coalgebra
  $c\colon C \to FC$ and states $x, y \in
  C$, a smaller formula distinguishing a state $x$ from a state
  $y$ can be extracted from the certificates of the behavioural
  equivalence classes of $x$ and $y$ in time
  $\CO(|C|)$: It is the leftmost conjunct that is different in the
  respective certificates of $x$ and
  $y$.  In other words, this is the subformula starting at the modal
  operator introduced in $\delta_i$ for the least
  $i$ with $(x,y)\notin P_i$; hence,
  $x$ satisfies $\fmod{t}(\delta,\beta)$ but
  $y$ satisfies $\fmod{t'}(\delta,\beta)$ for some $t'\neq t$ in $F3$.

  Hence, when only distinguishing formula for states $x$, $y$ is of interest,
  it is sufficient to run the algorithm until $x$ is split from $y$, and the
  distinguishing formula is conjunct added in step~\ref{defDeltai1}. This leads
  to an earlier termination in practice but does not change the run time
  complexity in $\CO$-notation.
\end{construction}
\begin{proposition}%
  \autoref{extractDistinguish} correctly extracts a formula
  distinguishing~$x$ from~$y$.
\end{proposition}
\begin{proof}%
  In order to verify that the first differing conjunct is a
  distinguishing formula, we distinguish cases on the
  least~$i$ such that $(x,y)\notin P_i$:

  If $x$ and $y$ are already split by $P_0$, then the conjunct at
  index $0$ in the respective certificates of $[x]_{\sim}$ and
  $[y]_{\sim}$ differs, and we have $t = F!(c(x))$ and
  $t' = F!(c(y))$. By \autoref{lemF1CoalgSem}, $\fmod{t}$
  distinguishes $x$ from $y$ (and $\fmod{t'}$ distinguishes~$y$ from
  $x$).

  If $x$ and $y$ are split by $P_{i+1}$ (but $(x,y)\in P_i$), then
  \[
    \underbrace{F\chi_S^B(c(x))}_{t~:=} \neq \underbrace{F\chi_S^B(c(y))}_{t'~:=}.
  \]
  Thus, the conjuncts that differ in the respective certificates for
  $[x]_{\sim}$ and $[y]_{\sim}$ are the following ones at index
  $i+1$:
  \[
    \fmod{t}(\delta_i(S),\beta_i(B))
    \qquad
    \text{and}
    \qquad
    \fmod{t'}(\delta_i(S),\beta_i(B)).
  \]
  By \autoref{lemF3CoalgSem}, $\fmod{t}(\delta_i(S),\beta_i(B))$
  distinguishes $x$ from $y$ (and $\fmod{t}(\delta_i(S),\beta_i(B))$
  distinguishes $y$ from $x$).
\end{proof}

\subsection{Complexity Analysis}%
\label{complexityAnalysis}
The operations introduced by \autoref{algoCerts} can be implemented
with only constant run time overhead. To this end, one implements
$\beta$ and $\delta$ as arrays of formulae of length $|C|$ (note that
at any point, there are at most $|C|$-many blocks). In the
refinable-partition data structure~\cite{ValmariLehtinen08}, every
block has an index (a natural number) and there is an array of
length~$|C|$ mapping every state $x\in C$ to the block it is contained
in. Hence, for both partitions $C/P$ and $C/Q$, one can look up a
state's block and a block's certificate in constant time.

It is very likely that the certificates contain a particular
subformula multiple times and that certificates of different blocks
share common subformulae. For example, every certificate of a block
refined in the $i$-th iteration using $S\subsetneqq B$ contains the
subformulae $\delta_i(S)$ and~$\beta_i(B)$.  Therefore, it is
advantageous to represent all certificates constructed as one directed
acyclic graph (dag) with inner nodes labelled by either a modal
operator or conjunction and having precisely two outgoing edges, and
leaf nodes labelled by either~$\top$ or a nullary modal
operator. Moreover, edges have a binary flag indicating whether they
represent negation $\neg$.  Initially, there is only one (leaf) node
representing~$\top$, and the operations of \autoref{algoCerts}
allocate new nodes and update the arrays for~$\beta$ and~$\delta$ to
point to the right nodes.  For example, if the predecessor block
$T\in C/P_i$ is refined in step~\ref{defDeltai1}, yielding a new block
$[x]_{P_{i+1}}$, then a new node labelled $\wedge$ is allocated with
edges to the nodes $\delta_i(T)$ and to another new node labelled
$F\chi_S^B(c(x))$ with edges to the nodes $\delta_i(S)$ and
$\delta_i(B)$.

For purposes of estimating the size of formulae generated by the
algorithm, we use a notion of \emph{transition} in coalgebras,
inspired by the notion of canonical graph~\cite{Gumm2005}.

\begin{definition}\label{coalgebraEdge}
  For states $x,y$ in an $F$-coalgebra $(C,c)$, we say that there is a
  \emph{transition $x\to y$} if $c(x)\in FC$ is not in the image
  $Fi[F(C\setminus \{y\})]~(\subseteq FC)$, where
  $i\colon C\setminus \{y\}\monoto C$ is the inclusion map.
\end{definition}
\begin{theorem}\label{certSize} For a coalgebra with~$n$ states
  and~$m$ transitions, the formula dag constructed by
  \autoref{algoCerts}
  has at most
  \[
    2\cdot m\cdot (\log_2 n + 1) + 2\cdot n
  \]
  nodes, each with outdegree $\le 2$, and a height of at most~${n+1}$. Hence, the dag size is
  in $\CO(m\cdot \log_2 n + n)$.
\end{theorem}
\noindent
Before proving \autoref{certSize}, we need to establish a sequence of
lemmas on the underlying partition refinement algorithm. Recall from
\autoref{R:trnkovahull}\ref{R:trnkovahull:2} that we may assume that~$F$
preserves finite intersections by working with its Trnkov\'a hull
instead.

Let $(C,c)$ be a coalgebra for $F$. We define a binary relation~$\to$
on $\Pow(C)$ by
\[
  T\to S \qquad\Longleftrightarrow\qquad
  \exists x\in T, y\in  S\colon x\to y
\]
for $T,S\subseteq C$.  In other words, we write $T\to S$ if there is a
transition from some state of $T$ to some state of $S$. Also we define
the set $\pred(S)$ of predecessor states of a set~$S$ as
\[
  \pred(S) = \big\{x\in C\mid \{x\}\to S\big\}.
\]
\begin{lemma}\label{noEdgeNoChiS} For every $F$-coalgebra $(C,c)$,
  $x\in C$, and $S\subseteq B\subseteq C$ with $S$ finite, we have
  \[
    \{x\}\not\to S
    \qquad\Longrightarrow\qquad
    F\chi_S^B(c(x)) = F\chi_{\emptyset}^B(c(x)).
  \]
\end{lemma}
\begin{proof}
  For every $y\in S$, we have that $x\not\to y$. Hence, for every
  $y\in S$, there exists $t_y\in F(C\setminus\{y\})$ such that
  \[
    c(x) = Fi(t_y)
    \qquad\text{for }i\colon C\setminus\{y\} \monoto C.
  \]
  The set $C\setminus S$ is the intersection of all sets
  $C\setminus\{y\}$ with $y \in S$:
  \[
    C\setminus S = \bigcap_{y\in S} (C\setminus\{y\}).
  \]
  Since $F$ preserves finite intersections and $S$ is
  finite, we have that
  \[
    F(C\setminus S) = \bigcap_{y\in S} F(C\setminus\{y\}).
  \]
  Since $c(x)\in FC$ is contained in every $F(C\setminus\{y\})$ (as witnessed by $t_y$)
  it is also contained in their intersection. That is, for $m\colon C\setminus
  S\monoto C$ being the inclusion map, there is $t'\in F(C\setminus
  S)$ such that $Fm(t') = c(x)$. Now consider the following diagrams:
  \[
    \begin{tikzcd}[column sep=2mm]
      c(x)
      \descto{r}{\(\in\)}
      & FC
      \arrow{r}{F\chi_S^B}
      &[8mm] F3
      \\
      t'
      \arrow[mapsto]{u}
      \descto{r}{\(\in\)}
      & F(C\setminus S)
      \arrow{u}{Fm}
      \arrow{ur}[swap]{F\chi_\emptyset^B}
    \end{tikzcd}
    \qquad\text{and}\qquad
    \begin{tikzcd}[column sep=2mm]
      FC
      \arrow{r}{F\chi_\emptyset^B}
      &[8mm] F3
      \\
      F(C\setminus S)
      \arrow{u}{Fm}
      \arrow{ur}[swap]{F\chi_\emptyset^B}
    \end{tikzcd}
  \]
  Both triangles commute because $\chi_\emptyset^B = \chi_S^B\cdot m$ and
  $\chi_{\emptyset}^B = \chi_\emptyset^B\cdot m$. Thus, we conclude
  \[
    F\chi_S^B(c(x))
    = F\chi_S^B(Fm(t'))
    = F\chi_\emptyset^B(t')
    = F\chi_\emptyset^B(Fm(t'))
    = F\chi_\emptyset^B(c(x)).\tag*{\qedhere}
  \]
\end{proof}

The classical Paige-Tarjan algorithm maintains the invariant that the
partition $P_i$ is \emph{stable}~\cite{PaigeTarjan87}, meaning that
for all states $(x,x')\in P_i$ in the same block, either both $x, x'$
or none of~$x,x'$ have a transition to a given block in $C/Q_i$. In
terms of characteristic functions~$\chi_\emptyset^B$ and for
general~$F$, this can be rephrased as follows.
\begin{lemma}\label{stableP}
  For all $(x,x') \in P_i$ and $B\in C/Q_i$ in \autoref{coalgPT}, we have
  \[
    F\chi_\emptyset^B(c(x)) = F\chi_\emptyset^B(c(x')).
  \]
\end{lemma}
\begin{proof}
  One can show~\cite[Prop.~4.12]{concurSpecialIssue} that in every
  iteration $i$ we have a map $c_i\colon C/P_i\to F(C/Q_i)$ that
  satisfies $F[-]_{Q_i}\cdot c= c_i\cdot [-]_{P_i}$:
  \[
    \begin{tikzcd}
      C
      \arrow{r}{c}
      \arrow[->>]{d}[swap]{[-]_{P_i}}
      & FC
      \arrow[->>]{d}{F[-]_{Q_i}}
      \\
      C/P_i
      \arrow{r}{c_i}
      & F(C/Q_i)
    \end{tikzcd}
  \]
  where the maps $[-]_{P_i}, [-]_{Q_i}$ send elements of $C$ to their
  equivalence class (\autoref{eqEqClass}). The map $\chi_\emptyset^B\colon
  C\to 3$ for $B\in C/Q_i$ can be decomposed as follows:
  \[
    \begin{tikzcd}
      C
      \arrow[->>]{d}[swap]{[-]_{Q_i}}
      \arrow{r}{\chi_\emptyset^B}
      & 3
      \\
      C/Q_i
      \arrow{ur}[swap]{\chi_\emptyset^{\{B\}}}
    \end{tikzcd}
  \]
  Combining these two diagrams, we obtain
  \begin{equation}\label{eqChiEq}
    F\chi_\emptyset^B\cdot c
    = F\chi_\emptyset^{\{B\}}\cdot F[-]_{Q_i}\cdot c
    = F\chi_\emptyset^{\{B\}}\cdot c_i\cdot [-]_{P_i}.
  \end{equation}
  Since for all $(x,x') \in P_i$, we have $[x]_{P_i} =
  [x']_{P_{i}}$,
  we conclude that
  \[
    F\chi_\emptyset^B(c(x))
    \overset{\eqref{eqChiEq}}= F\chi_\emptyset^{\{B\}}(c_i([x]_{P_i}))
    = F\chi_\emptyset^{\{B\}}(c_i([x']_{P_i}))
    \overset{\eqref{eqChiEq}}= F\chi_\emptyset^B(c(x')).
    \tag*{\qedhere}
  \]
\end{proof}
\noindent Combining the previous two lemmas, we obtain that in the
refinement step for $S\subsetneqq B$, only predecessor blocks of $S$
need to be adjusted in~\ref{defPi1}, so that only the formulae for
these blocks need to be updated:
\begin{lemma}\label{noEdgeNoSplit}
  For $S\subsetneqq B\in C/Q_i$ in the $i$th iteration of \autoref{coalgPT},
  a block $T\in C/P_i$ with no edge to $S$ is not modified; in
  symbols:

  \[
    T\not\to S \quad\Longrightarrow\quad
    T\in C/P_{i+1}
    \tag*{for all \(T\in C/P_i\).}
  \]
\end{lemma}
\begin{proof}
  Since $T\not\to S$, we have
  $\{x\}\not\to S$ and $\{x'\}\not\to S$ for all $x,x'\in T$. Thus,
  \begin{align*}
    F\chi_S^B(c(x))
    &= F\chi_\emptyset^B(c(x))
    &\text{(\autoref{noEdgeNoChiS}, $\{x\}\not\to S$)}
    \\
    &= F\chi_\emptyset^B(c(x'))
    &\text{(\autoref{stableP}, $(x,x')\in P_i$)}
    \\
    &= F\chi_S^B(c(x'))
    &\text{(\autoref{noEdgeNoChiS}, $\{x'\}\not\to S$)\rlap{.}}
      \tag*{\qedhere}
  \end{align*}
\end{proof}

\noindent The above lemma shows that the number of blocks $T$ that are
split is bounded by the number of predecessor blocks of the state set $S$.
Since most of the resulting smaller blocks $T'$ in the new partition
$C/P_{i+1}$ must have an edge to $S$, the number of these new blocks must also
be bounded essentially in the predecessors of $S$, as we show next:

\begin{lemma}%
  \label{newBlockCount}
  For  $S\subseteq C$ and finite $C$ in the $i$th iteration of
  \autoref{coalgPT},
  \[
    |\{T'\in C/P_{i+1}\mid T'\not\in C/P_i\}| ~\le~ 2\cdot |\pred(S)|.
  \]
\end{lemma}
\begin{proof}
  Let $S\subsetneqq B\in C/Q_i$ be used for splitting in iteration
  $i$.  If $T'\in C/P_{i+1}$ and $T'\not\in C/P_i$, then the block
  $T\in C/P_i$ with $T'\subseteq T$ satisfies $T\not\in C/P_{i+1}$ and
  therefore, by \autoref{noEdgeNoSplit}, has a transition to~$S$. By
  finiteness of $C$,~$T$ is split into finitely many blocks
  $T_1,\ldots,T_k\in C/P_{i+1}$, representing the equivalence classes
  of the kernel of $F\chi_S^B\cdot c\colon C\to F3$.  By
  \autoref{noEdgeNoChiS} we know that if $x\in T$ has no transition
  to~$S$, then $F\chi_S^B(c(x)) = F\chi_\emptyset^B(c(x))$. Moreover,
  all elements of $T\in C/P_i$ are sent to the same value by
  $F\chi_\emptyset^B\cdot c$ (\autoref{stableP}). Hence, there is at
  most one block~$T_j$ with no transition to $S$, and all other blocks
  $T_{j'}$, $j'\neq j$, have transitions to
  $S$. %
  Therefore, the number~$k$ of blocks $T_j$ is at most
  $|T\cap \pred(S)| + 1$.  Summing over all predecessor blocks~$T$
  of~$S$, we obtain
  \begin{align*}
    &|\{T'\in C/P_{i+1}\mid T'\not\in C/P_i\}|
    \\
    \le~&
    |\{T'\in C/P_{i+1}\mid T'\subseteq T\in C/P_i\text{ and }T\to S\}|
    & \text{(\autoref{noEdgeNoSplit})}
    \\
    =~&\sum_{\substack{T\in C/P_i\\ T\to S}} |\{T'\in C/P_{i+1}\mid T'\subseteq T\}|
    \\
    \leq~&\sum_{\substack{T\in C/P_i\\ T\to S}} (|T\cap \pred(S)| + 1)
    & \text{(bound on $k$ above)}
    \\
    \le~& 2\cdot \sum_{\substack{T\in C/P_i\\ T\to S}} |T\cap \pred(S)|
    &\text {($|T\cap \pred(S)| \ge 1$)}
    \\
    \le~& 2\cdot |\pred(S)|
    &\text{($T\in C/P_i$ are disjoint)} \tag*{\qedhere}
  \end{align*}

\end{proof}
We can now show a bound for the total number of blocks that exist at some point
during the execution of the partition refinement algorithm:
\begin{lemma}\label{totalBlockCount}
  Given an input coalgebra $(C,c)$ with $n=|C|$ states and $m$
  transitions, the following holds throughout the execution of
  \autoref{coalgPT}:
  \[
    |\{T\subseteq C\mid T\in C/P_i\text{ for some }i\}|
    \le 2\cdot m \cdot \log_2 n + 2\cdot m + n.
  \]
\end{lemma}
\noindent
Note that the proof is similar to arguments given in the complexity analysis
of the Paige-Tarjan algorithm (cf.~\cite[p.~980]{PaigeTarjan87}).
\begin{proof}
  Since $|S|\le \frac{1}{2}\cdot |B|$ holds in step~\ref{step1} of
  \autoref{coalgPT}, one can show that every state $x\in C$ is
  contained in the set~$S$ picked in step~\ref{step1} in at most
  $\log_2(n)+1$ iterations~\cite[Lem.~7.15]{concurSpecialIssue}.
  More formally, let
  $S_i\subsetneq B_i\in C/Q_i$ be the blocks picked in the $i$th iteration
  of \autoref{coalgPT}. Then we have
  \begin{equation}
    |\{S_i\mid x \in S_i\}| \le \log_2 n + 1
    \qquad\text{for all }x\in C.%
    \label{timesInSubblock}
  \end{equation}
  Let the algorithm terminate after $\ell$ iterations,
  returning~$C/P_\ell$. Then, the number of new blocks introduced by
  step~\ref{defPi1} is bounded as follows: \allowdisplaybreaks%
  \begin{align*}
    &\sum_{0\le i< \ell}
      |\{T'\in C/P_{i+1}\mid T' \notin C/P_i\}|
    \\ \le~& \sum_{0\le i< \ell} 2\cdot |\pred(S_i)|
    &\text{(\autoref{newBlockCount})}
    \\ \le~& 2\cdot \sum_{0\le i< \ell} \,\sum_{x\in S_i}|\pred(\{x\})|
    \\ =~& 2\cdot \sum_{x\in C}\,\sum_{0\le i< \ell\:\mid\: x\in S_i} |\pred(\{x\})|
    \\ =~& 2\cdot \sum_{x\in C}\, |\pred(\{x\})|\cdot \sum_{0\le i< \ell\:\mid\:  x\in S_i} 1
    \\ \le~& 2\cdot \sum_{x\in C}\, |\pred(\{x\})|\cdot (\log_2n + 1)
    & \text{(by~\eqref{timesInSubblock})}
    \\ =~& 2\cdot m \cdot (\log_2n + 1)
           = 2\cdot m\cdot \log_2 n + 2\cdot m
  \end{align*}
  The only blocks we have not counted so far are the blocks of $C/P_0$. Since
  $|C/P_0|\le n$, we have at most $2\cdot m\cdot \log_2 n + 2\cdot m+ n$
  different blocks in $(C/P_i)_{0\le i <\ell}$.
  \qedhere
\end{proof}

Every block in the final or intermediate partitions $C/P_i$ corresponds to a certificate
that is constructed at some point. Hence, the bound on the blocks is also a
bound for the dag size of formulae created by \autoref{algoCerts}.
\begin{proof}[Proof of \autoref{certSize}]
  Regarding the height of the dag, it is immediate that~$\delta_i$ and
  $\beta_i$ have height at most $i+1$. Since
  $|C/Q_{i}| < |C/Q_{i+1}| \le |C|=n$ for all $i$, there are at most
  $n$ iterations, with the final partition being
  $C/P_{n+1} = C/Q_{n+1}$.

  In \autoref{algoCerts} we create a new modal operator in the dag
  whenever \autoref{coalgPT} creates a new block in $C/P_i$ (either by
  initialization or in step~\ref{defDeltai1}). By
  \autoref{totalBlockCount}, the number of modalities in the dag is
  thus bounded by
  \[
    2\cdot m \cdot \log_2 n + 2\cdot m + n.
  \]
  In every iteration of the main loop, $\beta$ is extended by two new
  formulae, one for $S$ and one for $B\setminus S$. The formula
  $\beta_{i+1}(S)$ does not increase the size of the dag, because no
  new node needs to be allocated. For $\beta_{i+1}(B\setminus S)$, we
  need to allocate one new node for the conjunction, so there are at
  most $n$ new such nodes allocated throughout the execution of the
  whole algorithm\lsnote{}.  Thus, the
  total number of nodes in the dag is bounded by
  \[
    2\cdot m \cdot \log_2 n + 2\cdot m + 2\cdot n
  \]
  and each such node has an outdegree of at most 2 by construction.
\end{proof}
\begin{theorem}\label{runTimePreserved} \autoref{algoCerts} adds only
  constant run time overhead per step of \autoref{coalgPT}, and thus
  has the same asymptotic run time as \autoref{coalgPT}.
\end{theorem}
\noindent Like for the run time analysis of the underlying
\autoref{coalgPT}, we assume that the memory model supports random
access, i.e.~array access runs in constant time.
\begin{proof}
  The arrays for~$\beta$ and $\delta$ are re-used in every
  iteration. Hence, the index $i$ can be neglected; it is only used to
  refer to a value before or after the loop iteration. The analysis
  then proceeds as follows:
  \begin{enumerate}
  \item Initialization step:
    \begin{itemize}[leftmargin=!]
    \item The only block $\{C\}$ in $C/Q_0$ has index 0, and so we make $\beta(0)$
      point to the node $\top$, which takes constant time.
    \item For every block $T$ in $C/P_0$, \autoref{coalgPT} has
      computed $F!(c(x)) \in F1$ for some (in fact every) $x\in
      T$. Since $F1$ canonically embeds into $F3$
      (\autoref{notationF1Mod}), we create a new node labelled
      $\fmod{Fj_1(F!(c(x)))}$ with two edges
      to~$\top$.\lsnote{}  For every $T\in C/P_0$, this runs in constant
      time, which we allocate to the step where \autoref{coalgPT}
      creates~$T$.
    \end{itemize}

  \item In the refinement step, we can look up the certificates
    $\delta_i(S)$ and~$\beta_i(B)$ in constant time using the indices
    of the blocks $S$ and $B$\lsnote{}. Whenever the
    original algorithm creates a new block, we also immediately
    construct the certificate of this new block by creating at most
    two new nodes in the dag (with at most four outgoing
    edges). However, if a block does not change (that is,
    $[x]_{Q_i} = [x]_{Q_{i+1}}$ or $[x]_{P_i} = [x]_{P_{i+1}}$,
    resp.), then the corresponding certificate is not changed in
    step~\ref{defBetai1} or step~\ref{defDeltai1}, respectively.

    \smallskip
    \noindent
    In the loop body we update the certificates as follows:
    \begin{enumerate}[label=(I),font=\normalfont,align=left,leftmargin=0pt,labelindent=0pt,listparindent=\parindent,labelwidth=0pt,itemindent=!,topsep=3pt,parsep=0pt,itemsep=3pt,start=1]
    \item[\ref{defBetai1}] The new block $S\in C/Q_{i+1}$ just points
      to the certificate $\delta_i(S)$ constructed earlier. For the
      new block $(B\setminus S) \in C/Q_{i+1}$, we allocate a new node
      $\wedge$, with one edge to $\beta_i(B)$ and one negated edge to
      $\delta_i(S)$.

    \item[\ref{defDeltai1}] Not all resulting blocks have a transition
      to $S$. There may be (at most) one new block $T'\in C/P_{i+1}$,
      $T'\subseteq T$ with no transition to $S$ (see the proof of
      \autoref{newBlockCount}). In the refinable partition structure,
      such a block will inherit the index from $T$ (i.e.~the index of
      $T$ in $C/P_i$ equals the index of $T'$ in
      $C/P_{i+1}$). Moreover, every $x\in T'$ satisfies
      $F\chi_S^B(c(x)) = F\chi_\emptyset^B(c(x))$ (by
      \autoref{noEdgeNoChiS}), %
      and $F\chi_\emptyset^B(c(x)) = F\chi_\emptyset^B(c(y))$ for
      every $y\in T$ (by
      \autoref{stableP}). %

    Now, one first saves the node of the certificate $\delta_i(T)$ in some
    variable $\delta'$, say. Then the array $\delta$ is updated at index $T$ by the formula
    \[
      \fmod{F\chi_\emptyset^B(c(y))}(\delta_i(S),\beta_i(B))
      \qquad\text{for an arbitrary $y\in T$.}
    \]
    Consequently, a block $T'$ inheriting the index of $T$
    automatically has the correct certificate.

    The allocation of nodes for this formula is completely analogous
    to the one for an ordinary block $[x]_{P_{i+1}} \subsetneqq T$ having edges to $S$: One
    allocates a new node labelled $\wedge$ with edges to the saved node $\delta'$
    (the original value of $\delta_i(T)$) and to another newly allocated node labelled
    $\fmod{F\chi_S^B(c(x))}$ with edges to the nodes~$\delta_i(S)$ and $\delta_i(B)$.
    \qedhere
    \end{enumerate}
\end{enumerate}
\end{proof}

\noindent
In order to keep the formula size smaller, one can implement the
following optimization. Intuitively, note that for $S\in X/P_i$ and
$B\in X/Q_i$ such that $S\subseteq B\subseteq C$, every conjunct of
$\beta_i(B)$ is also a conjunct of $\delta_i(S)$.  In
$\beta_i(B)\wedge \neg \delta_i(S)$, one can hence remove all
conjuncts of $\beta_i(B)$ from $\delta_i(S)$, obtaining a formula
$\delta'$, and then equivalently use $\beta_i(B)\wedge \neg \delta'$
in the definition of $\beta_{i+1}(D)$.

\begin{proposition}\label{cancelConjunct} In step~\ref{defBetai1},
  $\beta_{i+1}(D)$ can be simplified to be no larger than
  $\delta_i(S)$ without increasing the overall run time.%
\end{proposition}
\begin{proof}
  Mark every modal operator node $\fmod{t}(\delta,\beta)$ in the
  formula dag with a boolean flag expressing whether
  \begin{center}
    $\fmod{t}(\delta,\beta)$ is a conjunct of some $\beta_i$-formula.
  \end{center}
  Thus, every new modal operator in~\ref{defDeltai1} is marked `false'
  initially. When the block~$B$ in $C/Q_i$ is split into~$S$ and
  $B\setminus S$ in step~\ref{defBetai1}, the formula for block
  $B\setminus S$ is a conjunction of $\beta_i(B)$ and the negation of
  all conjuncts of~$\delta_i(S)$ marked `false'. Afterwards these
  conjuncts are all marked `true', because they are inherited by
  $\beta_i(S)$. The conjuncts marked `false' always form a prefix of
  all conjuncts of a formula in $\delta_i$. It therefore suffices to
  greedily take conjuncts from the root of a formula dag while they
  are marked `false'.

  As a consequence, step~\ref{defDeltai1} no longer runs in constant
  time but instead takes as many steps as there are conjuncts marked
  `false' in~$\delta_i(S)$. However, over the whole execution of the
  algorithm this eventually amortizes because every newly allocated
  modal operator is initially marked `false' and later marked `true'
  precisely once. Thus, in the analysis conducted in the proof of
  \autoref{runTimePreserved}, the additional run time can be neglected
  asymptotically.
\end{proof}

\noindent For a tighter run time analysis of the underlying partition refinement
algorithm, one additionally requires that~$F$ is equipped with a \emph{refinement
  interface}~\cite[Def.~6.4]{concurSpecialIssue}, which is based on a
given encoding of $F$-coalgebras in terms of \emph{edges} between
states (encodings serve only as data structures and have no direct
semantic meaning, in particular do not entail a semantic reduction to
relational structures). This notion of edge yields the same numbers (in
$\CO$-notation) as \autoref{coalgebraEdge} for
all functors considered.
All zippable functors we consider here have
refinement interfaces~\cite{concurSpecialIssue,wdms21}. In presence
of a refinement interface, step~\ref{defPi1} can be implemented
efficiently, with resulting overall run time
$\CO((m+n) \cdot \log n\cdot p(c))$ where $n=|C|$, $m$ is the number
of edges in the encoding of the input coalgebra $(C,c)$, and the
\emph{run-time factor} $p(c)$ is associated with the refinement interface. In
most instances, e.g.~for~$\Pow$, $\R^{(-)}$, one has $p(c) = 1$; in particular,
the generic algorithm has the same run time as the Paige-Tarjan
algorithm. Usually, it is less of a challenge to find some refinement interface
for a functor $F$ but more to find one with low run time, i.e.~low $p(c)$. For
example, for general monoid-valued functors $FX = M^{(X)}$ the refinement
interface uses binary search trees as additional data structures resulting in
an additional logarithmic factor $p(c) = \log m$~\cite{wdms21}.

\begin{rem}
  The claimed run time relies on close attention to a number of
  implementation details. This includes use of an efficient data
  structure for the
  partition~$C/P_i$~\cite{Knuutila2001,ValmariLehtinen08}; the other
  partition $C/Q_i$ is only represented implicitly in terms of a queue
  of blocks $S\subsetneqq B$ witnessing $P_i\subsetneqq Q_i$,
  requiring additional care when splitting blocks in the queue~\cite[Fig.~3]{ValmariF10}. Moreover, grouping the elements of a
  block by $F3$ involves the consideration of a \emph{possible
    majority candidate}~\cite{ValmariF10}.
\end{rem}

\begin{theorem}\label{thmLogRunTime} For a zippable set functor with
  a refinement interface with factor~$p(c)$ and an input coalgebra
  with $n$ states and $m$ transitions, \autoref{algoCerts} runs in
  time
  \[
    \CO((m+n)\cdot \log n\cdot p(c)).
  \]
\end{theorem}
\noindent
Indeed, the time bound holds for the underlying \autoref{coalgPT}, and
is inherited by \autoref{algoCerts} due to \autoref{runTimePreserved}.

If the functor $F$ satisfies additional assumptions, we can simplify
the certificates even further, as discussed next.

\section{Cancellative Functors}%
\label{sec:cancellative}
Our use of binary modalities relates to the fact that, as observed
already by Paige and Tarjan, when splitting a block according to an
existing partition of a block~$B$ into $S\subseteq B$ and~$B\setminus S$,
it is not in general sufficient to look only at the successors
in~$S$. However, this does suffice for some transition types.
E.g.~Hopcroft's algorithm for deterministic automata~\cite{Hopcroft71}
and Valmari and Franceschinis' algorithm for weighted systems
(e.g.~Markov chains)~\cite{ValmariF10} both split only with respect
to~$S$. In the following, we exhibit a criterion on the level of
functors that captures that splitting w.r.t.~only $S$ is sufficient:
\begin{samepage}
\begin{defn}\label{funCancellative}
  A functor $F$ is \emph{cancellative} if the map
  \[
    \fpair{F\chi_{\set{1,2}},F\chi_{\set{2}}}\colon F3\to F2\times F2
  \]
  is injective.
\end{defn}
\end{samepage}
\noindent To understand the role of the above map, recall the function
$\chi_S^B\colon C\to 3$ from~\eqref{eqChi3} and note that
\begin{equation}\label{eq:chi}
  \chi_{\set{1,2}}\cdot\chi_S^B=\chi_B
  \qquad\text{and}\qquad
  \chi_{\set{2}}\cdot\chi_S^B=\chi_S,
\end{equation}
so the composite
$\fpair{F\chi_{\set{1,2}},F\chi_{\set{2}}}\cdot F\chi_S^B$ yields
information about the accumulated transition weights into~$B$ and~$S$
but not about the one into $B\setminus S$. The injectivity condition
means that for cancellative functors, this information suffices in the
splitting step for $S\subseteq B\subseteq C$.  The term
\emph{cancellative} stems from the respective property on monoids;
recall that a monoid~$M$ is \emph{cancellative} if $s + b_1 = s+b_2$
implies $b_1 = b_2$ for all $s,b_1,b_2\in M$.

\subsection{Properties of Cancellative Functors}

Before presenting the optimized algorithm, we gather properties of cancellative
functors and compare them to zippability and related notions, starting with the
property that gave cancellative functors their name.

\begin{proposition}%
  \label{monoidValCancellative}
  The monoid-valued functor $M^{(-)}$ for a commutative monoid $M$ is
  cancellative if and only if~$M$ is a cancellative monoid.
\end{proposition}
\begin{proof}%
  First note that for $FX=M^{(X)}$, the maps $F\chi_{\set{1,2}}$,
  $F\chi_{\set{2}}$ used in \autoref{funCancellative} are given by
  \[
    \begin{array}{r@{\ }l@{\qquad}l}
      M^{(\chi_{\set{1,2}})}\colon & M^{(3)}\to M^{(2)}, & t \mapsto
      (t(0), t(1)+t(2)), \\[5pt]
      M^{(\chi_{\set{2}})}\colon & M^{(3)}\to M^{(2)}, & t \mapsto (t(0)+t(1), t(2)),
    \end{array}
  \]
  where we write $s\in M^{(2)}$ as the pair $(s(0), s(1))$.

  For ``$\Leftarrow$'', let $s,t\in M^{(3)}$ such that
  \[
    \fpair{M^{(\chi_{\set{1,2}})},M^{(\chi_{\set{2}})}}(s)
    = \fpair{M^{(\chi_{\set{1,2}})},M^{(\chi_{\set{2}})}}(t),
  \]
  which is written point-wise as follows:
  \begin{align*}
    (s(0), s(1)+s(2))
    &= (t(0), t(1)+t(2)) \\
    (s(0)+s(1), s(2))
    &= (t(0)+t(1), t(2)).
  \end{align*}
  We thus have $s(0) = t(0)$, $s(2) = t(2)$, and
  \[
    s(1) + s(2) = t(1) + t(2) = t(1) + s(2).
  \]
  Since $M$ is cancellative, it follows that $s(1) = t(1)$, so
  $s=t$. Thus, the map
  $\fpair{M^{(\chi_{\set{1,2}})},\allowbreak M^{(\chi_{\set{2}})}}$ is injective.

  For ``$\Rightarrow$'', let $a,b,c\in M$ such that $c+a =
  c+b$. Define $s,t\in M^{(3)}$ by
  \[
    s(0) = s(2) = c,\quad s(1) = a
    \qquad\text{and}\qquad
    t(0) = t(2) = c,\quad t(1) = b.
  \]
  Thus,
  \[
  \begin{aligned}
    M^{(\chi_{\set{1,2}})}(s) &= (s(0), s(1) + s(2)) \\
    &= (c,a+c) \\
    &= (c,b+c)
    \\
    &= (t(0), t(1) + t(2))
    \\
    &= M^{(\chi_{\set{1,2}})}(t),
  \end{aligned}
  \qquad\qquad
  \begin{aligned}
    M^{(\chi_{\set{2}})}(s)
    & = (s(0) + s(1), s(2))\\
    &= (c+a,c)\\
    &= (c+b,c) \\
    &= (t(0) + t(1), t(2)) \\
    &= M^{(\chi_{\set{2}})}(t).
  \end{aligned}
  \]
  Since $\fpair{M^{(\chi_{\set{1,2}})},M^{(\chi_{\set{2}})}}$ is
  injective, it follows that $s =t $. Thus, we have
  $a=s(1) = t(1) = b$, so $M$ is cancellative.
\end{proof}
\noindent The property of cancellativity nicely extends a list of
correspondences between properties of the monoid-valued functor
$M^{(-)}\colon \Set\to\Set$ on the one hand and the underlying
commutative monoid $M$ on the other hand, see
\autoref{tabMonoidVsFunctor} and work by Gumm and
Schröder~\cite{GummS01,SchroederPhd01} for more details.
\begin{example}
  The functor $\R^{(-)}$ is cancellative, but $\Powf$, being naturally
  isomorphic to $M^{(-)}$ for the (non-cancellative) Boolean
  monoid~$M=2$, is not.
\end{example}
\noindent
All signature functors are cancellative:

\begin{table}
  \begin{tabular}{@{}lllr@{}}
    \toprule
    The functor $M^{(-)}\colon \Set\to\Set$ $\ldots$
    & $\Leftrightarrow$ & The monoid $M$ $\ldots$ &
    \\
    \midrule
    is cancellative
    & $\Leftrightarrow$ & is cancellative & (\autoref{monoidValCancellative})
    \\
    \midrule
    preserves inverse images
    & $\Leftrightarrow$ & is positive & \cite[4.74]{SchroederPhd01} \&~\cite{GummS01} \\ 
    preserves weak kernel pairs
    & $\Leftrightarrow$ & is refinable & \cite[4.74]{SchroederPhd01} \&~\cite{GummS01}\\ 
    preserves weak pullbacks
    & $\Leftrightarrow$ & is positive and refinable & \cite[4.35]{SchroederPhd01}\\ 
    \bottomrule
  \end{tabular}
  \caption{Correspondence between properties of the functor and the
    monoid}%
  \label{tabMonoidVsFunctor}
\end{table}

\begin{proposition}\label{cancellativeClosure}
  The class of cancellative functors
  contains the identity functor and all constant functors, and
  is closed under subfunctors, products, and coproducts.
\end{proposition}

\begin{proof}
  \begin{enumerate}
  \item The identity functor is cancellative because the map $\langle
    \chi_{\set{1,2}},\chi_{\set{2}}\rangle$ is clearly injective.
  \item For the constant functor $C_X$ with value $X$,  $C_X(\chi_S)$ is the
    identity map on $X$ for every set~$S$. Therefore $C_X$ is cancellative.

  \item Let $\alpha\colon F\monoto G$ be a natural transformation with injective
    components and let $G$ be cancellative. Combining the naturality squares of
    $\alpha$ for $\chi_{\set{1,2}}$ and $\chi_{\set{2}}$, we obtain the
    commutative square
    \[
      \begin{tikzcd}
        F3
        \arrow{r}{\fpair{F\chi_{\set{1,2}}, F{\chi_{\set{2}}}}}
        \arrow[>->]{d}[swap]{\alpha_{3}}
        &[15mm]
        F2\times F2
        \arrow{d}{\alpha_{2}\times \alpha_{2}}
        \\
        G3
        \arrow[>->]{r}{\fpair{G\chi_{\set{1,2}}, G{\chi_{\set{2}}}}}
        & G2\times G2,
      \end{tikzcd}
    \]
    in which the composite $F3\to G2\times G2$ is injective by
    hypothesis. Hence, $\fpair{F\chi_{\set{1,2}},F\chi_{\set{2}}}$ is
    injective as well, showing that the subfunctor $F$ is
    cancellative.

  \item Let $(F_i)_{i\in I}$ be a family of cancellative
    functors, and suppose that we have elements $s,t\in
    (\prod_{i\in I}F_i)(3) = \prod_{i \in I}F_i3$ with
    \[
      \big(\prod_{i\in I}F_i\chi_{\set{1,2}}\big)(s)
      = \big(\prod_{i\in I}F_i\chi_{\set{1,2}}\big)(t)
      \quad\text{and}\quad
      \big(\prod_{i\in I}F_i\chi_{\set{2}}\big)(s)
      = \big(\prod_{i\in I}F_i\chi_{\set{2}}\big)(t).
    \]
    Write $\pr_i$ for the $i$th projection function from the product.
    For every $i \in I$ we have:
    \[
      F_i\chi_{\set{1,2}}(\pr_i(s))
      = F_i\chi_{\set{1,2}}(\pr_i(t))
      \qquad\text{and}\qquad
      F_i\chi_{\set{2}}(\pr_i(s))
      = F_i\chi_{\set{2}}(\pr_i(t)).
    \]
    Since every $F_i$ is cancellative, we have $\pr_i(s)=\pr_i(t)$ for every
    $i\in I$. This implies~$s=t$ since the product projections $(\pr_i)_{i\in I}$
    are jointly injective.
    \smnote{}

  \item Again, let $(F_i)_{i\in I}$ be a family of cancellative
    functors. Suppose that we have elements $s,t\in
    (\coprod_{i\in I}F_i)(3) = \coprod_{i \in I}F_i3$ satisfying%
    \smnote{}
    \[%
      \big(\coprod_{i\in I}F_i\chi_{\set{1,2}} \big)(s)
      =  \big(\coprod_{i\in I}F_i\chi_{\set{1,2}} \big)(t)
      \quad\text{and}\quad
       \big(\coprod_{i\in I}F_i\chi_{\set{2}} \big)(s)
      =  \big(\coprod_{i\in I}F_i\chi_{\set{2}} \big)(t).
    \]
    This implies that there exists an $i\in I$ and $s',t'
    \in F_i3$ with $s = \inj_i (s')$, $t=\inj_i(t')$, and
    \[
      F_i\chi_{\set{1,2}}(s')
      = F_i\chi_{\set{1,2}}(t')
      \qquad\text{and}\qquad
      F_i\chi_{\set{2}}(s')
      = F_i\chi_{\set{2}}(t').
    \]
    Since $F_i$ is cancellative, we have $s'=t'$, which implies $s=t$.
    \qedhere
  \end{enumerate}
\end{proof}

\noindent
A consequence of closure under subfunctors is that, for example,
$\Dist$ is cancellative, being a subfunctor of $\R^{(-)}$, but $\Pow$
is not, as we have already seen that its subfunctor~$\Powf$ fails to
be cancellative.
\begin{proposition}\label{cancellativeVsZippable}
  Cancellative functors are neither closed under quotients nor under composition.
  Zippability and cancellativity are independent properties.
\end{proposition}
\begin{table}
  \subfloat[][Operations]{
  \begin{tabular}{@{}lll@{}}
    \toprule
    Operation & cancellative & non-cancellative \\
    \midrule
    Quotient & $X  \mapsto \coprod_{n \in \N}X^n$ & $\Powf$ \\[1mm]
    Composition & $\Bag=\N^{(-)}$ & $\Bag\Bag$ \\
    \bottomrule
  \end{tabular}
  }
  \hfill
  \subfloat[][Independence of zippability]{
    \begin{tabular}{@{}lll@{}}
      \toprule
      & cancellative & non-cancel. \\
      \midrule
      zippable &
                $X\mapsto X$
                 \phantom{\rlap{$\coprod_{\N}X^n$}} %
                    & $\Powf$
      \\[1mm]
      non-zippable & see~\eqref{size4sets}
                     \phantom{\rlap{$\N^{(-)}$}} %
                    & $\Powf\Powf$
      \\
      \bottomrule
    \end{tabular}
  }
  \caption{Counter examples regarding cancellative functors.}%
  \label{cancelCounterEx}
\end{table}

\begin{proof}
  \autoref{cancelCounterEx} shows an overview of all counterexamples used in
  the present proof.
  \begin{enumerate}
\item Cancellative functors are not closed under quotients: e.g.~the
  non-cancellative functor~$\Powf$ is a quotient of the signature
  functor $X  \mapsto \coprod_{n \in \N}X^n$ (which is cancellative by \autoref{cancellativeClosure}).

\item Cancellative functors are not closed under composition. For the additive
  monoid $(\N,+,0)$ of natural numbers, the monoid-valued functor $\Bag =
  \N^{(-)}$ sends $X$ to the set of finite multisets on~$X$ (`bags'). Since $\N$
  is cancellative, $\Bag$ is a cancellative functor. However,
  $\Bag\Bag$ is not (below we write $\bag{\cdots}$ to denote
  multisets, so~$\bag{0,1} = \bag{1,0}$ but $\bag{1}\neq \bag{1,1}$):
  \begin{align*}
    &\fpair{\Bag\Bag\chi_{\set{1,2}},\Bag\Bag\chi_{\set{2}}}\big(\bag[\big]{\bag{0,1},\bag{1,2}}\big)
    \\
    &= \big(\bag[\big]{\bag{0,1},\bag{1,1}}, \bag[\big]{\bag{0,0},\bag{0,1}}\big) \\
    &= \big(\bag[\big]{\bag{0,1},\bag{1,1}}, \bag[\big]{\bag{0,1},\bag{0,0}}\big) \\
    &=\fpair{\Bag\Bag\chi_{\set{1,2}},\Bag\Bag\chi_{\set{2}}}\big(\bag[\big]{\bag{0,2}, \bag{1,1}}\big).
  \end{align*}
  Thus, the map
  $\fpair{\Bag\Bag\chi_{\set{1,2}},\Bag\Bag\chi_{\set{2}}}$ is not injective.

\item The identity functor $X\mapsto X$ is both
  zippable~\cite{concurSpecialIssue} and cancellative~(\autoref{cancellativeClosure}).
\item The monoid-valued functor $\Powf = \BoolMonoid^{(-)}$ is
  zippable~\cite{concurSpecialIssue}, but not
  cancellative~(\autoref{monoidValCancellative}), because $\BoolMonoid$ is a
  non-cancellative monoid.

  \item
  The functor $\Pow\Pow$ is neither
  zippable~\cite[Ex.~5.10]{concurSpecialIssue} nor cancellative because
  \begin{align*}
    \fpair{\Pow\Pow\chi_{\set{1,2}},\Pow\Pow\chi_{\set{2}}} (\big\{\set{0}, \set{2}\big\})
    &= (\set[\big]{\set{0}, \set{1}}, \set[\big]{\set{0}, \set{1}})
    \\ &= \fpair{\Pow\Pow\chi_{\set{1,2}},\Pow\Pow\chi_{\set{2}}} (\big\{\set{0},
    \set{1}, \set{2}\big\}).
  \end{align*}

  \item
  Every functor $F$ satisfying $|F(2+ 2)| > 1$ and $|F3| = 1$ is
  cancellative but not zippable:
  \begin{itemize}[leftmargin=!]
  \item Indeed, every map with domain $1$ is injective, in particular
    the map
    \[
      \fpair{F\chi_{\set{1,2}},F\chi_{\set{2}}}\colon 1\cong F3 \longrightarrow F2\times F2,
    \]
    whence $F$ is cancellative.

  \item If $|F(2+2)|>1$ and $|F3|=1$ we have that the map
    \[
      \fpair{2+\mathord{!},\mathord{!}+2}
      \colon \underbrace{F(2+2)}_{|-|\mathrlap{\, > 1}} \to
      \underbrace{F(2+1)}_{\cong F3\cong 1} \times
      \underbrace{F(1+2)}_{\cong F3\cong 1}\cong 1
    \]
    is not injective, whence $F$ is not zippable.
  \end{itemize}
  A concrete example of such a functor is given by\twnote{}
  \begin{equation}
    FX = \{ S \subseteq X : |S| = 0 \text{ or }|S| = 4\}%
    \label{size4sets}
  \end{equation}
  which sends a map $f\colon X\to Y$ to the map $Ff\colon FX\to FY$ defined by
  \[
    Ff(S) = \begin{cases}
      f[S] & \text{if }|f[S]| = 4 \\
      \emptyset & \text{otherwise}.
    \end{cases}
    \qedhere
  \]
\end{enumerate}
\end{proof}
\noindent
In related work, König~\etal~\cite{KoenigEA20} construct
distinguishing formulae in coalgebraic generality. Their assumption is
a generalized version of zippability, where the binary coproduct is
replaced with an $m$-ary coproduct for $m\in \N$:

\begin{defiC}[{\cite{KoenigEA20}}]
A functor $F$
is $m$-zippable
if the canonical map
\[
  \unzip_m\colon~~ F(A_1+A_2+\ldots + A_m) 
  ~~\longrightarrow~~ F(A_1 + 1)
  \times F(A_2 + 1)
  \times\cdots
  \times F(A_m + 1)
\]
is injective. Explicitly, $\unzip_m$ is given by
\[\textstyle
  \fpair{F[\Delta_{i,j}]_{j\in \bar m}}_{i\in \bar m}\colon~~
  F\big(\coprod_{j = 1}^{m} A_j\big)
  \longrightarrow
  \prod_{i = 1}^{m} F(A_i + 1)
\]
where $\bar m$ is the set $\bar m = \set{1,\ldots,m}$ and the map $\Delta_{i,j}$ is defined by
\[
  \Delta_{i,j}\colon A_j\to A_i + 1
  \qquad
  \Delta_{i,j} := \begin{cases}
    A_j \xra{\inl} A_i + 1
    &\text{if }i = j \\
    A_j \xra{!} 1 \xra{\inr} A_i + 1
    & \text{if }i \neq j.
  \end{cases}
\]
\end{defiC}

Ordinary zippability is then equivalent to $2$-zippability.
\begin{proposition}
  Every zippable and cancellative set functor is $m$-zippable for
  every $m$.%
\end{proposition}
\begin{proof}
  First, we show that for a zippable and cancellative set functor $F$,
  the map%
  \smnote{}
  \[
    g_{A,B} ~~:=~~
    F(A+1+B)
    \xra{\fpair{F(A+!), F(!+B)}}
    F(A+1)\times F(1+B)
  \]
  is injective for all sets $A,B$. Indeed, we have
  the following chain of injective maps, where the index at the $1$ is only
  notation to distinguish coproduct components more easily:
  \allowdisplaybreaks%
  \begin{align*}
    & F(A+(1_M +B))
    \\
    & \quad\monodown \fpair{F(A+!), F(!+(1_M+B))} \tag{$F$ is zippable}
    \\
    & F(A+1) \times F(1_A+1_M+B)
    \\
    & \quad\monodown \id\times \fpair{F(!+B),F(1_A+1_M+!)}\tag{$F$ is zippable}
    \\
    & F(A+1)\times F(1+B) \times F(1_A+1_M+1_B)
    \\
    & \quad\monodown \id\times \id \times \fpair{F\chi_{1_M+1_B}, F\chi_{1_B}} \tag{$F$ is cancellative, $1_A+1_M+1_B \cong \{0,1,2\}$} 
    \\
    & F(A+1)\times F(1+B) \times F2 \times F2
  \end{align*}
  Call this composite~$f$.  It factorizes through $g_{A,B}$, because
  it matches with $g_{A,B}$ on the components $F(A+1)$ and $F(1+B)$,
  and for the other components, one has the map
  \[
    h ~:=~
    F(A+1)\times F(1+B)
    \xra{F\chi_{1}\times F\chi_B}
    F2 \times F2
  \]
  with $f = \fpair{\id_{F(A+1)\times F(1+B)}, h}\cdot g_{A,B}$. Since $f$ is
  injective, $g_{A,B}$ must be injective, too.

  Also note that a rewriting \autoref{funCancellative} along the
  isomorphisms $1+1+1\cong 3$ and $1+1 \cong 2$, we obtain that a
  functor~$F$ is cancellative iff the map
  \[
    \fpair{F(1+!), F(!+1)}\colon
    F(1+1+1)\longrightarrow F(1+1)\times F(1+1)
  \]
  (where $!\colon 1+1 \to 1$) is injective.

  We now proceed with the proof of the desired implication by
  induction on $m$.  In the base cases $m=0$ and $m=1$, there is
  nothing to show because every functor is $0$- and $1$-zippable, and
  for $m=2$, the implication is trivial (zippability coincides with
  $2$-zippability by definition). In the inductive step, given that
  $F$ is $2$-zippable, $m$-zippable ($m\ge 2$), and cancellative, we
  show that $F$ is $(m+1)$-zippable.

  We have the following chain of injective maps, where we again annotate some of
  the singleton sets $1$ with indices to indicate from which coproduct
  components they come:
  \begin{align*}
    & F(A_1+\ldots+A_{m-1}+(A_m+A_{m+1})) 
    \\
    & \qquad\monodown \unzip_m
      \tag{$F$ is $m$-zippable}
    \\
    &
    \big(\textstyle\prod_{i=1}^{m-1} F(A_i + 1) \big)
    \times F(A_m + A_{m+1} + 1_{1..(m-1)})
    \\
    \cong~
    &
    \big(\textstyle\prod_{i=1}^{m-1} F(A_i + 1)\big)
    \times F(A_m + 1_{1..(m-1)} + A_{m+1})
    \\
    &\quad \monodown {\id\times g_{A_m,A_{m+1}}}
      \tag{the above injective helper map $g$}
      \\
    &
    \big(\textstyle\prod_{i=1}^{m-1} F(A_i + 1) \big)
    \times F(A_m + 1)
    \times F(1+A_{m+1})
    \\
    \cong~
    &
    \big(\textstyle\prod_{i=1}^{m-1} F(A_i + 1) \big)
    \times F(A_m + 1)
    \times F(A_{m+1} + 1)
  \end{align*}
  This composite thus is injective as well, and coincides with
  $\unzip_{m+1}$, showing that $F$ is $(m+1)$-zippable.
\end{proof}

The converse, however, does not hold because the finite powerset functor $\Powf$
is $m$-zippable for all $m$~\cite[Ex.~10, Lem.~14]{KoenigEA20}, but not
cancellative as we have seen.

\subsection{Optimized Partition Refinement and Certificates}
The optimization present in the algorithms for
Markov chains~\cite{ValmariF10} and automata~\cite{Hopcroft71} can now
be adapted to coalgebras for cancellative functors, where it suffices
to split only according to transitions into $S$, ignoring
transitions into $B\setminus S$. More formally, this means that we
replace the three-valued $\chi_S^B\colon C \to 3$ with
$\chi_S\colon C\to 2$ in the refinement step~\ref{defPi1}: %

\begin{proposition}\label{optimizedPcancel}
  Let $F$ be a cancellative set functor. For $S\in C/P_i$ in the $i$-th
  iteration of \autoref{coalgPT}, we have
  \(
  P_{i+1} = P_i \cap \ker (C\xra{c}{FC}\xra{F\chi_S} F2).
  \)
\end{proposition}

\begin{proof}
  From the definition~\eqref{eqKer} of the kernel, we immediately
  obtain the following properties for all maps $f,g\colon Y\to Z$,
  $h \colon X\to Y$:
  \begin{align}
    f\text{ injective} ~~&\Longrightarrow~~ \ker(f\cdot h) = \ker(h)%
                         \label{kerInjective}
                         \\
    \ker(f) = \ker(g) ~~&\Longrightarrow~~ \ker(f\cdot h) = \ker(g\cdot h)%
                         \label{kerPreCompose}
                          \\
    \ker(\fpair{f,g}) &= \ker(f)\cap\ker(g).%
                         \label{kerIntersect}
  \end{align}
  For every coalgebra $c\colon C\to FC$ and $S\subseteq B\subseteq C$
  we have by~\eqref{eq:chi} that
  \[
    \fpair{F\chi_B,F\chi_S} = \fpair{F\chi_{\set{1,2}}, F\chi_{\set{2}}}\cdot F\chi_S^B.
  \]
  Since $F$ is cancellative, $\fpair{F\chi_{\set{1,2}}, F\chi_{\set{2}}}$ is
  injective, and we thus obtain
  \begin{equation}
    \ker(\fpair{F\chi_B,F\chi_S})
    = \ker(\fpair{F\chi_{\set{1,2}}, F\chi_{\set{2}}}\cdot F\chi_S^B)
    \overset{\text{\eqref{kerInjective}}}{=} \ker(F\chi_S^B).
  \end{equation}
  By~\eqref{kerPreCompose}, this implies that
  \begin{equation}
    \ker(\fpair{F\chi_B,F\chi_S}\cdot c)
    = \ker(F\chi_S^B\cdot c).%
    \label{cancelBinaryChi}
  \end{equation}

  Let $B\in C/Q_i$ be the block that is split into $S$ and $B\setminus S$ in
  iteration $i$. Since $P_i$ is finer than $Q_i$ and $B\in C/Q_i$, we have
  $P_i\subseteq Q_i\subseteq \ker(F\chi_B\cdot c)$; thus:
  \begin{equation}
    P_i = P_i \cap \ker(C\xra{c}{FC} \xra{F\chi_B} F2).%
    \label{PiKerB}
  \end{equation}
  Now we verify the desired property:\smnote{}
  \begin{align*}
    P_{i+1} &= ~P_i\cap \ker(C\xra{c}{FC}\xra{F\chi_S^B} F2)
    & \text{(by~\ref{defPi1})}
    \\
    &\overset{\mathclap{\text{}}}{=}
    ~P_i\cap \ker(\fpair{F\chi_B,F\chi_S}\cdot c)
    & \text{(by~\eqref{cancelBinaryChi})}
    \\
    &= P_i\cap \ker(\fpair{F\chi_B\cdot c,F\chi_S\cdot c})
    \\ &
    = P_i\cap \ker(F\chi_B\cdot c)\cap\ker(F\chi_S\cdot c)
    &\text{(by~\eqref{kerIntersect})}
      \\ &
    = P_i\cap\ker(F\chi_S\cdot c)
    &\text{(by~\eqref{PiKerB})} & \qedhere
  \end{align*}
\end{proof}
\noindent Note that this result is independent of certificate
construction and already improves the underlying partition refinement
algorithm.
\begin{expl}
  Suppose that $F$ is a a signature functor $\Sigma$ or a
  monoid-valued functor~$M^{(-)}$ for a cancellative monoid~$M$. Given
  an input coalgebra for $F$, the refinement step~\ref{defPi1} of
  \autoref{coalgPT} can be optimized to compute $P_{i+1}$ according to
  \autoref{optimizedPcancel}.
\end{expl}
\noindent Observe that, in the optimized step~\ref{defPi1}, $B$ is no
longer mentioned. It is therefore unsurprising that we do not need a
certificate for it when constructing certificates for the blocks of
$P_{i+1}$. Instead, we can reflect the map
$F\chi_S\cdot c\colon C\to F2$ in the coalgebraic modal formula and
take (unary) modal operators from $F2$.  Just like $F1$ in
\autoref{notationF1Mod}, the set~$F2$ canonically embeds into $F3$.

\begin{notation}\label{notationF2Mod} Define the injective map
  $j_2\colon 2\monoto 3$ by $j_2(0) = 1$ and $j_2(1)=2$. The
  injection $Fj_2\colon F2\monoto F3$ provides a way to interpret
  elements $t\in F2$ as unary modalities~$\fmod{t}$:
  \[
    \fmod{t}(\delta) := \fmod{Fj_2(t)}(\delta,\top).
  \]
\end{notation}

\begin{rem}
  There are several different ways to define $\fmod{t}(\delta)$ for $t\in
  F2$, depending on the definition of the inclusion $j_2$.
  \begin{center}
    \def\arraystretch{1.2}
    \begin{tabular}{l@{\qquad}l@{\qquad}l}
      \toprule
      $j_2\colon 2\monoto 3$ & $j_2\cdot \chi_S$ for $S\subseteq C$
      & Definition for $t\in F2$ \\
      \midrule
      $0\mapsto 0, 1\mapsto 1$
      & $j_2\cdot \chi_S = \chi_{\emptyset}^S$
      & $\fmod{t}(\delta) := \fmod{Fj_2(t)}(\bot,\delta)$
        \\
      $0\mapsto 0, 1\mapsto 2$
      & $j_2\cdot \chi_S = \chi_{S}^S$
      & $\fmod{t}(\delta) := \fmod{Fj_2(t)}(\delta,\delta)$
        \\
      $0\mapsto 1, 1\mapsto 2$
      & $j_2\cdot \chi_S = \chi_{S}^C$
      & $\fmod{t}(\delta) := \fmod{Fj_2(t)}(\delta,\top)$
        \\
      \bottomrule
    \end{tabular}
  \end{center}
  All these variants make the following \autoref{lemF2CoalgSem} true
  because in each case, the data $j_2,\phi,\psi$ used in the
  definition of $\fmod{t}(\delta)$ as $\fmod{Fj_2(t)}(\phi,\psi)$ satisfy
  \begin{equation*}
    j_2\cdot \chi_{\semantics{\delta}} = \chi_{\semantics{\phi}}^{\semantics{\psi}}.%
    \label{eqF2ModProp}
  \end{equation*}
\end{rem}
\noindent In analogy to \autoref{lemF3CoalgSem}, we can show:
\begin{lemma}\label{lemF2CoalgSem} Given a cancellative functor~$F$,
  an $F$-coalgebra $(C,c)$, $t\in F2$, a formula~$\delta$, and
  a state $x \in C$, we have
  \[
    x\in \semantics{\fmod{t}(\delta)} ~\Longleftrightarrow~ F\chi_{\semantics{\delta}}(c(x)) = t.
  \]
\end{lemma}
\begin{proof}
  By the definition of~$j_2$, we have $j_2\cdot \chi_{S} = \chi_S^C$
  for all $S\subseteq C$. Thus,
  \begin{align*}
    \semantics{\fmod{t}(\delta)}
    &= \semantics{\fmod{Fj_2(t)}(\delta,\top)}
    &\text{(\autoref{notationF2Mod})}
    \\
    & = \{x \in C \mid F\chi_{\semantics{\delta}}^{C}(c(x)) = Fj_2(t)\}
    &\text{(\autoref{lemF3CoalgSem}, $\semantics{\top} = C$)}
    \\
    & = \{x \in C \mid Fj_2(F\chi_{\semantics{\delta}}(c(x))) = Fj_2(t)\}
    &\text{($\chi_{\semantics{\delta}}^C = j_2\cdot \chi_{\semantics{\delta}}$)}
    \\
    & = \{x \in C \mid F\chi_{\semantics{\delta}}(c(x)) = t\}
    &\text{($Fj_2$ injective)}
  \end{align*}
  In the last step, we use that $F$ preserves injective maps (\autoref{R:trnkovahull}\ref{R:trnkovahull:2}).
\end{proof}

In \autoref{algoCerts}, the family $\beta$ is only used in the
definition of $\delta_{i+1}$ to characterize the larger block $B$ that
has been split into the smaller blocks $S\subseteq B$ and
$B\setminus S$. For a cancellative functor, we can replace
\[
    \fmod{F\chi_S^B(c(x))}(\delta_i(S),\beta_i(B))
    \qquad\text{with}\qquad
    \fmod{F\chi_S(c(x))}(\delta_i(S))
\]
in the definition of $\delta_{i+1}$. Hence, we can omit $\beta_i$ from
\autoref{algoCerts} altogether, obtaining the following algorithm,
which is again based on coalgebraic partition refinement
(\autoref{coalgPT}).

\begin{algorithm}\label{algoCertsCancel}
  We extend \autoref{coalgPT} as follows. Initially, define
  \[
    \delta_0([x]_{P_0}) = \fmod{F!(c(x))}.
  \]
  In the $i$-th iteration, extend step~\ref{defPi1} by the additional assignment
  \begin{enumerate}[label=({A'}1),topsep=5pt]
    \setcounter{enumi}{2}
    \makeatletter
    \renewcommand{\labelenumi}{\theenumi}
    \renewcommand{\theenumi}{\text{\bfseries\color{black}(A\!'\arabic{enumi})}}
    \renewcommand{\p@enumi}{}
    \makeatother

  \item\label{defDeltai1Optimized} $
      \delta_{i+1}([x]_{P_{i+1}}) = \begin{cases}
           \delta_{i}([x]_{P_{i}}) &\text{if }[x]_{P_{i+1}} = [x]_{P_i} \\
          \delta_{i}([x]_{P_{i}}) \wedge \fmod{F\chi_S(c(x))}(\delta_i(S))
          &\text{otherwise.}\\
          \end{cases}
        $
  \end{enumerate}
\end{algorithm}

\begin{theorem}\label{algoPatchCorrect}
  For cancellative functors, \autoref{algoCertsCancel} is correct;
  that is, we have:
  \[
    \forall S\in X/P_i\colon
    \semantics{\delta_i(S)}
    = S
    \qquad\text{for all $i\in \N$}.
  \]
\end{theorem}
\begin{rem}
  \noindent Note that the optimized \autoref{algoCertsCancel} can also
  be implemented treating the unary modal operators in~$F2$ as first
  class citizens, in lieu of embedding them into~$F3$ as we did in
  \autoref{notationF2Mod}.  The only difference between the two
  implementation approaches w.r.t.\ the size of the formula dag is one
  edge per modality, namely the edge to the node~$\top$ from the node
  $\fmod{Fj_2(F\chi_S(c(x)))}(\delta_i(\delta_i),\top)$ that arises
  when step~\ref{defDeltai1Optimized} is expanded according to
  \autoref{notationF2Mod}.
\end{rem}

\begin{proof}[Proof of \autoref{algoPatchCorrect}]
  Induction over~$i$, the index of loop iterations.

  The definition of $\delta_0$ is identical to the definition in
  \autoref{algoCerts}, whence
  \[
    \semantics{\delta_0(S)} = S
    \qquad
    \text{for all }S\in C/P_0,
  \]
  proved completely analogously as in the proof of \autoref{algoCertsCorrect}.

  In the $i$-th iteration with chosen block $S\in C/P_i$, we
  distinguish cases on whether a block $[x]_{P_{i+1}} \in C/P_{i+1}$
  remains the same or is split into other blocks:
  \begin{itemize}%
  \item If $[x]_{P_{i+1}} = [x]_{P_{i}}$, then we have
    \[
      \semantics{\delta_{i+1}([x]_{P_{i}})}
      \overset{\text{\ref{defDeltai1Optimized}}}{=}
      \semantics{\delta_{i}([x]_{P_{i}})}
      \overset{\text{I.H.}}{=}
      [x]_{P_i} = [x]_{P_{i+1}}.
    \]
  \item If $[x]_{P_{i+1}} \neq [x]_{P_{i}}$, we compute as follows:
    \begin{align*}
      \semantics{\delta_{i+1}([x]_{P_{i+1}})}
      &=
      \semantics{\delta_{i}([x]_{P_{i}})\wedge
        \fmod{F\chi_S(c(x))}(\delta_i(S))}
      &\text{\ref{defDeltai1Optimized}}
        \\
      &=
      \semantics{\delta_{i}([x]_{P_{i}})}\cap\semantics{ \fmod{F\chi_S(c(x))}(\delta_i(S))}
      \\
      &=
      [x]_{P_{i}} \cap\semantics{ \fmod{F\chi_S(c(x))}(\delta_i(S))}
      &
      \text{(I.H.)}
      \\
      &=
      [x]_{P_{i}} \cap
      \{x'\in C\mid F\chi_{\semantics{\delta_i(S)}}(c(x'))
      = F\chi_S(c(x)) \}
      &\text{(\autoref{lemF2CoalgSem})}
      \\ &=
      [x]_{P_i} \cap
      \{x'\in C\mid F\chi_{S}(c(x'))
      = F\chi_S(c(x)) \}
      &\text{(I.H.)}
      \\ &=
       [x]_{P_i} \cap
       \{x'\in C\mid (x,x') \in \ker(F\chi_{S}\cdot c) \}
       &\text{(def.~$\ker$)}
      \\ &=
       [x]_{P_i} \cap
       [x]_{F\chi_{S}\cdot c}
       &\text{(def.~$[x]_{R}$)}
      \\ &= [x]_{P_{i+1}}. %
    \end{align*}
    The last step is follows from
    $P_{i+1} = P_i\cap \ker(F\chi_S\cdot c)$ (see \autoref{optimizedPcancel}).
    \qedhere
  \end{itemize}
\end{proof}
\noindent
The formulae resulting from the optimized construction involve only
$\wedge$, $\top$, and modalities from the set~$F2$ (or $F3$ with the
second parameter fixed to $\top$), which we term
\emph{$F2$-modalities}. Thus, Hennessy-Milner Theorem
(\autoref{hennessyMilner}) can be sharpened for cancellative functors
as follows.
\begin{corollary}
  For a zippable and cancellative set functor~$F$, states in a finite
  $F$-coalgebra are behaviourally equivalent iff they agree on modal
  formulae built using $\top$, $\wedge$, and unary $F2$-modalities.
\end{corollary}
\noindent The certificates thus computed are reduced to roughly half
the size compared to \autoref{algoCerts}; the asymptotic run time and
formula size (\autoref{complexityAnalysis}) remain unchanged.

\section{Domain-Specific Certificates}%
\label{domainSpecific}
On a given specific system type, one is typically interested in
certificates and distinguishing formulae expressed via modalities
whose use is established in the respective domain, e.g.~$\Box$
and~$\Diamond$ for transition systems. We next describe how the
generic $F3$ modalities can be rewritten to domain-specific ones in a
postprocessing step. The domain-specific modalities will not always
be equivalent to $F3$-modalities, but still yield
certificates.  \twnote{}

\begin{defn}\label{remBoolCombination}
  The \emph{Boolean closure} $\bar\Lambda$ of a modal signature
  $\Lambda$ has as $n$-ary modalities
  propositional combinations of atoms of the form
  $\hearts(i_1,\dots,i_k)$, for $\arity{\hearts}{k}\in\Lambda$, where
  $i_1,\dots,i_k$ are propositional combinations of elements of
  $\{1,\ldots,n\}$. Such a modality~$\arity{\lambda}{n}$ is
  interpreted by predicate liftings
  $\semantics{\lambda}_X\colon (2^X)^n \to FX$ defined inductively in
  the obvious way.
\end{defn}
\noindent For example, the Boolean closure of
$\Lambda = \{\arity{\Diamond}{1}\}$ contains the unary modality
$\Box = \neg\Diamond\neg 1$.

\begin{defn}\label{domainCert}
  Given a modal signature $\Lambda$ for a functor $F$, a
  \emph{domain-specific interpretation} consists of functions
  $\tau\colon F1 \to \bar\Lambda$ and
  $\lambda\colon F3 \to \bar\Lambda$ assigning to each $o \in F1$ a
  nullary modality~$\tau_o$ and to each $t\in F3$ a binary
  modality~$\lambda_t$ such that the predicate liftings
  $\semantics{\tau_o}_X\in 2^{FX}$ and
  $\semantics{\lambda_t}_X\colon (2^X)^2 \to 2^{FX}$ satisfy
  \[
    \semantics{\tau_o}_1 = \{o\}
    \quad\text{(in $2^{F1}$)}
    \quad\text{ and }\quad
    [t]_{F\chi_{\{1,2\}}} \cap \semantics{\lambda_t}_3(\{2\},\{1\})
    =
    \{t\}
    \quad\text{(in $2^{F3}$)}.
  \]
  (Recall that $\chi_{\{1,2\}}\colon 3\to 2$ is the characteristic
  function of $\{1,2\}\subseteq 3$, and
  $[t]_{F\chi_{\{1,2\}}}\subseteq F3$ denotes the equivalence class of
  $t$ w.r.t.~$F\chi_{\{1,2\}}\colon F3 \to F2$.)
\end{defn}
\noindent Thus,~$\tau_o$ holds precisely at states with output
behaviour $o\in F1$.  Intuitively, $\lambda_t(\delta,\rho)$ describes
the refinement step of a predecessor block $T$ when splitting
$B:=\semantics{\delta}\cup\semantics{\rho}$ into
$S:=\semantics{\delta}$ and $B\setminus S:=\semantics{\rho}$
(\autoref{fig:parttree}), which translates into the arguments $\{2\}$
and $\{1\}$ of~$\semantics{\lambda_t}_3$. In the refinement step, we
know from previous iterations that all elements have the same
behaviour w.r.t.~$B$. This is reflected in the intersection with
$[t]_{F\chi_{\set{1,2}}}$. The given condition on~$\lambda_t$ thus
guarantees that $\lambda_t$ characterizes $t\in F3$ uniquely, but only
within the equivalence class representing a predecessor block.  Thus,
$\lambda_t$ can be much smaller than equivalents of $\fmod{t}$
(cf.~\autoref{examplePowF3Mod}):

\begin{expl}\label{exDomainSpecInt}
  We provide examples for set functors of interest; the verification
  that these are indeed domain-specific interpretations follows in
  \autoref{L:correct} further below.
  \begin{enumerate}
  \item\label{F3Pow} For $F=\Pow$, we have a domain-specific
    interpretation over the modal signature
    $\Lambda = \set{ \Diamond/1}$.  For $\emptyset, \{0\}\in\Pow 1$,
    take $\tau_\emptyset = \neg\Diamond\top$ and
    $\tau_{\set 0} = \Diamond\top$. For $t \in \Pow 3$, we put
  \[
    \begin{array}{r@{\,}l@{\quad}l@{\qquad\qquad}r@{\,}l@{\quad}l}
      \lambda_t(\delta,\rho)
      &=
      \neg\Diamond \rho &\text{if }2\in t\notni 1
    &
    \lambda_t(\delta,\rho) &=
        \Diamond \delta \wedge \Diamond\rho &\text{if }2\in t \ni 1
                          \\
    \lambda_t(\delta,\rho) &=
        \neg\Diamond \delta &\text{if }2\notin t\ni 1
                          &
    \lambda_t(\delta,\rho) &=
        \top &\text{if }2\not\in t \notni 1.
    \end{array}
  \]
  The certificates obtained via this translation are precisely the
  ones generated in the example using the Paige-Tarjan
  algorithm, cf.~\eqref{edgeCasesFormula}, with $\rho$ in lieu of
  $\beta\wedge\neg\delta$. %
\item\label{dsiSignature} For a signature (functor) $\Sigma$, take
  $\Lambda=\{\arity{\sigma}{0}\mid \arity{\sigma}{n}\in \Sigma\}\cup
  \{\arity{\gradI}{1}\mid I\in\Powf (\N)\}$. We
  interpret~$\Lambda$ by the predicate liftings
    \begin{align*}
      \semantics{\sigma}_X & = \{\sigma(x_1, \ldots, x_n) \mid x_1,
      \ldots, x_n \in X\} \subseteq \Sigma X, \\
      \semantics{\gradI}(S) &= \{\sigma(x_1,\ldots,x_n) \in \Sigma X \mid
                              \forall i \in \N\colon i\in I \leftrightarrow (1\le i\le n~\wedge~ x_i \in S)  \}.
    \end{align*}
    Intuitively,~$\sigma$ states that the next operation symbol
    is~$\sigma$, and $\gradI\,\phi$ states that the $i$th successor
    satisfies $\phi$ iff $i\in I$.  We then have a domain-specific
    interpretation $(\tau,\lambda)$ given by
    \[
      \begin{array}{r@{\,}l@{\qquad}l}
        \tau_o &= \sigma &
        \text{for $o = \sigma(0,\ldots,0) \in \Sigma 1$, and}\\
        \lambda_t(\delta,\rho) &= \grad{I}\delta
        & \text{for $t = \sigma(x_1,\ldots,x_n) \in \Sigma 3$ and $I =
          \{i\in \{1,\ldots,n\}\mid x_i = 2\}$}.
      \end{array}
    \]
  \item\label{dsiMonoid} For a monoid-valued functor~$M^{(-)}$, take
    $\Lambda = \{ \arity{\monmod{m}}{1} \mid m \in M \}$, interpreted
    by the predicate liftings
    $\semantics{{\monmod{m}}}_X\colon 2^X\to 2^{M^{(X)}}$ given by
    \[
      \semantics{{\monmod{m}}}_X(S)= \{ \mu \in M^{(X)} \mid
      m =
      \textstyle\sum_{x \in S}\mu(x)\}.
    \]
    A formula $\monmod{m}\,\delta$ thus states that the accumulated weight
    of the successors satisfying~$\delta$ is exactly~$m$. A
    domain-specific interpretation $(\tau,\lambda)$ is then given by
    \[
      \begin{array}{r@{\,}l@{\qquad}l}
        \tau_o &= \monmod{o(0)}\top & \text{for $o \in M^{(1)}$, and} \\
        \lambda_t(\delta,\rho)
        &= \grad{t(2)}\,\delta \wedge  \grad{t(1)}\,\rho
        &\text{for $t\in M^{(3)}$.}
      \end{array}
    \]
    In case~$M$ is cancellative, we can also simply put
    $\lambda_t(\delta,\rho) = \grad{t(2)}\,\delta$.

  \item\label{dsiMarkov} For labelled Markov chains,
    i.e.~$FX=(\Dist X +1)^A$, let
    $\Lambda = \{\arity{\fpair{a}_p}{1}\mid a\in A, p\in [0,1]\}$,
    where $\fpair{a}_p\phi$ denotes that on input $a$, the next state
    will satisfy~$\phi$ with probability at least~$p$, as in cited work by
    Desharnais \etal~\cite{desharnaisEA02}. This gives rise to the
    interpretation:
    \[
      \tau_o =
      \bigwedge_{\substack{a\in A\\ o(a) \in \Dist 1}} \fpair{a}_{1}\top
      \wedge
      \bigwedge_{\substack{a\in A\\ o(a) \in 1}} \neg \fpair{a}_{1}\top,
      \qquad\qquad
      \lambda_t(\delta,\rho)
      = \bigwedge_{\substack{a\in A\\t(a)\in \Dist 3}}(\fpair{a}_{t(a)(2)}\,\delta\wedge \fpair{a}_{t(a)(1)}\,\rho).
    \]
  \end{enumerate}
\end{expl}
\noindent To ease the verification of the requisite properties of
interpretations, we will now show that for cancellative~$F$,
domain-specific interpretations can be derived from a simpler kind of
interpretation, one where the set $F3$ in \autoref{exDomainSpecInt} is
replaced with $F2$.

\begin{defn}%
  \label{domainCertSimple}
  Given a modal signature~$\Lambda$ for a functor~$F$, a \emph{simple
    domain-specific interpretation} consists of functions
  $\tau\colon F1 \to \bar\Lambda$ and
  $\kappa\colon F2 \to \bar\Lambda$ assigning a nullary
  modality~$\tau_o$ to each $o \in F1$ and a unary modality $\kappa_s$
  to each $s \in F2$ such that the predicate liftings
  $\semantics{\tau_o}_X\in 2^{FX}$ and
  $\semantics{\kappa_s}\colon 2^X \to 2^{FX}$ satisfy
  \[
    \semantics{\tau_o}_1 = \{o\}
    \quad\text{(in $2^{F1}$)}
    \qquad\text{and}\qquad
    [s]_{F!} \cap \semantics{\kappa_s}_2(\{1\}) = \{s\}
    \qquad\text{(in $2^{F2}$).}
  \]
\end{defn}

\begin{proposition}\label{domainCertCancellative}
  Let $\Lambda$ be a modal signature for a cancellative functor~$F$,
  and $(\tau,\kappa)$ a simple domain-specific interpretation. Define
  $\lambda\colon F3 \to \bar\Lambda$ by
  $\lambda_{t}(\delta,\rho) =
  \kappa_{F\chi_{\set{2}}(t)}(\delta)$. Then $(\tau, \lambda)$ is a
  domain-specific interpretation.
\end{proposition}

\begin{proof}
  Given $t \in F3$, we put $s = F\chi_{\set{2}}(t) \in F2$.  We have
  to show that
  \[
    [t]_{s} \cap \semantics{\lambda_t}_3(\{2\},\{1\}) = \{t\}
    \qquad\text{in $2^{F3}$}.
  \]
  By the naturality of the predicate lifting $\semantics{\kappa_s}$,
  the following square commutes (recall that $2^{(-)}$ is
  contravariant):
  \begin{equation}\label{diagNat}
    \begin{tikzcd}
      2^2
      \arrow{r}{\semantics{\kappa_s}_2}
      \arrow{d}[swap]{2^{\chi_{\set{2}}}}
      & 2^{F2}
      \arrow{d}{2^{F\chi_{\set{2}}}}
      \\
      2^3
      \arrow{r}{\semantics{\kappa_s}_3}
      & 2^{F3}
    \end{tikzcd}
  \end{equation}
  We thus have
  \begin{align*}
    \semantics{\lambda_t}_3(\set{2},\set{1})
    &=
    \semantics{\kappa_s}_3(\set{2})
    & \text{(def.~$\lambda_t$)}
    \\
    &= \semantics{\kappa_s}_3(\chi_{\set{2}}^{-1}[\set{1}])
    &\text{(def.~$\chi_{\set{2}}$)}
    \\
    &= \semantics{\kappa_s}_3(2^{\chi_{\set{2}}}(\set{1}))
    &\text{(def.~$2^{(-)}$)}
    \\
    &= 2^{F\chi_{\set{2}}}(\semantics{\kappa_s}_2(\set{1}))
    &\text{(by~\eqref{diagNat})}
    \\
    &= \set[\big]{t'\in F3\mid F\chi_{\set{2}}(t')\in
    \semantics{\kappa_s}_2(\set{1})}
    &\text{(def.~$2^{(-)}$)\rlap{.}}
  \end{align*}
  The square
  \[
    \begin{tikzcd}
      3 \ar{r}{\chi_{\{1,2\}}}
      \ar{d}[swap]{\chi_{\{2\}}}
      &
      2
      \ar{d}{!}
      \\
      2
      \ar{r}{!}
      &
      1
    \end{tikzcd}
  \]
  trivially commutes. Hence, given $t' \in [t]_{F\chi_{\{1,2\}}}$,
  that is, $F\chi_{\{1,2\}}(t') = F\chi_{\{1,2\}}(t)$, postcomposing
  this equation with $F!$ and using the above commutative square under 
  $F$ we see that
  \[
    F! \cdot F\chi_{\{2\}}(t') = F! \cdot
    F_{\chi_{\{2\}}}(t),
  \]
  which yields
  \begin{equation}\label{eq:F-bang}
    F\chi_{\{2\}}(t') \in [F\chi_{\{2\}}(t)]_{F!}.
  \end{equation}
  Using this, we have for every $t'\in F3$ the following chain of equivalences%
  \begin{align*}
    &t' \in [t]_{F\chi_{\{1,2\}}} \cap \semantics{\lambda_t}_3(\{2\},\{1\})
      \\
    \Leftrightarrow~&  t'\in [t]_{F\chi_{\{1,2\}}}\text{ and }t'\in\semantics{\lambda_t}_3(\{2\},\{1\})
    \\
    \Leftrightarrow~&  t'\in [t]_{F\chi_{\{1,2\}}}\text{ and }F\chi_{\set{2}}(t')\in \semantics{\kappa_s}_2(\set{1})
    \tag{by the above calculation}
    \\
    \Leftrightarrow~&  t'\in [t]_{F\chi_{\{1,2\}}}\text{ and } F\chi_{\set{2}}(t')\in [F\chi_{\set{2}}(t)]_{F!}\cap\semantics{\kappa_s}_2(\set{1})
    \tag{by~\eqref{eq:F-bang} since $t' \in [t]_{F\chi_{\{1,2\}}}$} 
    \\
    \Leftrightarrow~&  t'\in [t]_{F\chi_{\{1,2\}}}\text{ and } F\chi_{\set{2}}(t')\in [s]_{F!}\cap\semantics{\kappa_s}_2(\set{1})
        \tag{def.~$s$}
    \\
    \Leftrightarrow~&  t'\in [t]_{F\chi_{\{1,2\}}}\text{ and } F\chi_{\set{2}}(t')\in \{s\}
        \tag{assumption on $\kappa_s$}
    \\
    \Leftrightarrow~& t'\in [t]_{F\chi_{\{1,2\}}}\text{ and } F\chi_{\set{2}}(t')\in \{F\chi_{\set{2}}(t)\}
        \tag{def.~$s$}
    \\
    \Leftrightarrow~& F\chi_{\{1,2\}}(t') = F\chi_{\{1,2\}}(t)\text{ and } F\chi_{\set{2}}(t') = F\chi_{\set{2}}(t)
    \\
    \Leftrightarrow~& \fpair{F\chi_{\{1,2\}},F\chi_{\{2\}}}(t') = \fpair{F\chi_{\{1,2\}},F\chi_{\{2\}}}(t)
    \tag{def.~$\fpair{-,-}$}
                    \\
    \Leftrightarrow~& t' = t
                      \tag{$F$ cancellative}
  \end{align*}
  For the last step, recall that
  $\fpair{F\chi_{\{1,2\}},F\chi_{\{2\}}}$ is injective because $F$ is
  cancellative.
\end{proof}

\begin{lemma}\label{L:correct}
  \autoref{exDomainSpecInt} correctly defines domain-specific
  interpretations.
\end{lemma}
\begin{proof}
  We verify the items in \autoref{exDomainSpecInt} separately:
  \begin{enumerate}
  \item%
    Recall that for $t\in \Pow 3$, we have defined
  \[
    \lambda_t(\delta,\rho) ~=~
    \begin{cases}
      \neg\Diamond \rho &\text{if }2\in t\notni 1 \\
      \Diamond \delta \wedge \Diamond\rho &\text{if }2\in t \ni 1 \\
      \neg\Diamond \delta &\text{if }2\not\in t\ni 1  \\
      \top &\text{if }2\not\in t \notni 1 \\
    \end{cases}
  \]
  Evaluating $\semantics{\lambda_t}_3(\delta,\rho)$ on $(\{2\},\{1\})$
  for the above cases, we obtain
    \begin{align*}
      &\semantics{\neg\Diamond\rho}_3(\set{2},
        \set{1}) = \{t'\in \Pow 3 \mid 1\notin t' \},
        && \text{if }2\in t \notni 1                      \\
       & \semantics{
        \Diamond\delta\wedge\Diamond\rho}_3(\set{2}, \set{1}) =
      \{t'\in \Pow 3 \mid 2\in t'\text{ and }1\in t' \},
       && \text{if }2\in t \ni 1              \\
      & \semantics{
                   \neg\Diamond
                         \delta}_3(\set{2},
                     \set{1}) = \{t'\in \Pow 3 \mid 2\notin t' \}
      && \text{if }2\notin t \ni 1                 \\
      &  \semantics{\Diamond\delta\wedge\Diamond\rho}_3(\set{2},
                          \set{1}) = \Pow 3.
      &&\text{if }2\notin t \notni 1
    \end{align*}
    Intersecting with $[t]_{\Pow\chi_{\set{1,2}}}$ yields $\set{t}$ as
    desired in all cases:
    \[
      \begin{array}{@{}l|l|l@{\,\cap\,}l@{\,=\,}l@{}} 
        \toprule
        t & \lambda_t(\delta,\rho) & \semantics{\lambda_t}_3(\set{2},\set{1}) & [t]_{\Pow\chi_{\set{1,2}}} & \set{t}
        \\
        \midrule
        \set{2} &\neg\Diamond \rho&\{t'\in \Pow 3 \mid 1\notin t' \} & \set[\big]{\set{2},\set{1},\set{2,1}}& \set[\big]{\set{2}}
        \\
        \set{2,0} &\neg\Diamond \rho&\{t'\in \Pow 3 \mid 1\notin t' \}
        &\set[\big]{\set{2,0},\set{1,0},\set{2,1,0}}
        & \set[\big]{\set{2,0}}\\
        \set{2,1} &\Diamond\delta\wedge\Diamond\rho& \{t'\in \Pow 3 \mid 2\in t'\text{ \& }1\in t' \}
        & \set[\big]{\set{2},\set{1},\set{2,1}} & \set[\big]{\set{2,1}}\\
        \set{2,1,0} &\Diamond\delta\wedge\Diamond\rho& \{t'\in \Pow 3 \mid 2\in t'\text{ \& }1\in t' \}
        &\set[\big]{\set{2,0},\set{1,0},\set{2,1,0}} & \set[\big]{\set{2,1,0}} \\
        \set{1} &\neg\Diamond \delta& \{t'\in \Pow 3 \mid 2\notin t' \}
        & \set[\big]{\set{2},\set{1},\set{2,1}} & \set[\big]{\set{1}}\\
        \set{1,0} &\neg\Diamond \delta& \{t'\in \Pow 3 \mid 2\notin t' \}
        & \set[\big]{\set{2,0},\set{1,0},\set{2,1,0}} & \set[\big]{\set{1,0}} \\
        \set{0} &\top& \Pow 3 & \set[\big]{\set{0}} & \set[\big]{\set{0}} \\
        \emptyset &\top& \Pow 3 & \set[\big]{\emptyset} & \set[\big]{\emptyset} \\
        \bottomrule
      \end{array}
    \]
    Hence, $\semantics{\lambda_t}_3(\set{2},\set{1}) \cap [t]_{\Pow\chi_{\set{1,2}}} = \set{t}$.
\item %
  For a signature functor $\Sigma$, we first define a helper map
  $v\colon \Sigma 2\to \Powf \N$ by
  \[
    v(\sigma(x_1,\ldots,x_n))= \{i\in \N\mid x_i = 1\}.
  \]
  The predicate lifting for the (unary)
  modal operator $\gradI$, for $I\subseteq \N$, is obtained from
  \autoref{predLiftYoneda} by the predicate $f_I\colon \Sigma 2\to 2$
  corresponding to the subset
  \[
    f_I = \{t\in \Sigma 2\mid v(t) = I\}.
  \]
  This gives rise to the predicate lifting
  \begin{align*}
    \semantics{\gradI}_X(P) &=
    \{t\in \Sigma X\mid F\chi_P(t) \in f_I \}
    &\text{(\autoref{predLiftYoneda})}
    \\ &=
    \{t\in \Sigma X\mid v(F\chi_P(t)) = I\}
    &\text{(def.~$f_I$)\rlap{.}}
  \end{align*}
  Similarly, for the nullary modal operator $\sigma$ given by
  $\arity{\sigma}{n}\in \Sigma$, take the predicate
  $\Sigma 1 \to 2$ corresponding to the subset
    \[
      g_\sigma = \{\sigma(0,\ldots,0)\} \subseteq \Sigma 1
    \]
    Since $1 = 2^0$, this gives rise to the $0$-ary predicate lifting
    \begin{align*}
      \semantics{\sigma}_X &=
      \{t\in \Sigma X\mid F!(t) \in g_\sigma \}
      &\text{(\autoref{predLiftYoneda})}
      \\ &=
      \big\{t\in \Sigma X\mid F\chi_P(t) \in \{\sigma(0,\ldots,0)\}\big\}
      &\text{(def.~$g_\sigma$)}
      \\ &= \{\sigma(x_1,\ldots,x_n)\mid x_1,\ldots,x_n \in X\}.
    \end{align*}
     We now put
    \[
      \kappa_s(\delta) := \grad{v(s)}\delta
      \qquad\text{for }s\in \Sigma 2,
    \]
    and we proceed to show that this yields a simple domain-specific
    interpretation (\autoref{domainCertSimple}) and that it induces
    the desired $\lambda_t$ via \autoref{domainCertCancellative}:
    \[
      \lambda_{\sigma(x_1,\ldots,x_n)}(\delta,\rho) = \grad{\{i\in \N\mid x_i = 2\}}\delta
      \qquad\text{for }\sigma(x_1,\ldots,x_n) \in \Sigma 3.
    \]
    There is nothing to show for $\tau_o := \sigma$ since it has the correct
    semantics by the definition of~$\semantics{\sigma}_1$.
    Next note that the map
    $\fpair{\Sigma !,v}\colon \Sigma 2\to \Sigma 1\times \Powf\N$ is
    injective because for every $s\in \Sigma 2$, the operation symbol
    and all its parameters (from $2$) are uniquely determined by
    $\Sigma !(s)$ and $v(s)$. %
    Recalling that
    $[s]_{\Sigma !} = \{s' \in \Sigma 2\mid \Sigma !(s) = \Sigma
    !(s')\}$, we now compute%
    \[
      \renewcommand{\arraystretch}{1.2}
      \begin{array}{@{}r@{\,}l@{}r}
      [s]_{\Sigma!} \cap \semantics{\kappa_s}_2(\set{1})
      &= \{ s'\in \Sigma 2\mid s'\in [s]_{\Sigma!}\text{ and }s'\in \semantics{\kappa_s}_2(\set{1})\}
      \\&= \{ s'\in \Sigma 2\mid \Sigma!(s) = \Sigma!(s')\text{ and }s'\in \semantics{\grad{v(s)}}_2(\set{1})\}
      \\&= \{ s'\in \Sigma 2\mid \Sigma!(s) = \Sigma!(s')\text{ and }
      v(\Sigma\chi_{\set{1}}(s')) = v(s)\}
      & \text{(def.~$\semantics{\grad{v(s)}}_2$)}
      \\&= \{ s'\in \Sigma 2\mid \Sigma!(s) = \Sigma!(s')\text{ and }
      v(s') = v(s)\}
      &
      \text{($\id_2=\chi_{\set{1}}\colon 2\to 2$)}
      \\&= \{ s'\in \Sigma 2\mid \fpair{\Sigma!,v}(s) = \fpair{\Sigma!,v}(s') \}
      &\text{(def.~$\fpair{-,-}$)}
      \\&= \{ s\}
      &\text{($\fpair{\Sigma!,v}$ injective)\rlap{.}}
    \end{array}
    \]
  \item%
    For every $m\in M$, let $f_m\colon
    M^{(2)}\to 2$ be the predicate corresponding to the
    subset
    \[
      \{\mu \in M^{(2)}\mid \mu(1) = m\}.
    \]
    It induces the unary predicate lifting
    $\semantics{\grad{m}}$ by
    \begin{align*}
      \semantics{\grad{m}}_X(P) &= \{\mu \in M^{(X)} \mid M^{(P)}(\mu) \in f_m\}
      &\text{(\autoref{predLiftYoneda})} \\
      &=\{\mu \in M^{(X)} \mid M^{(P)}(\mu)(1) = m\}
      &\text{(def.~$f_m$)\rlap{.}}
    \end{align*}
    To see that we have a domain-specific interpretation
    (\autoref{domainCert}), we note first (using
    $\semantics{\top}_1 = 1 = \{0\}$) that $\tau$ satisfies
    \begin{align*}
      \semantics{\tau_o}_1 = \semantics{\grad{o(0)}\top}_1
      &= \textstyle\{\mu\in M^{(1)}\mid \sum_{x\in \semantics{\top}_1} \mu(x) = o(0)\}
      \\
      &= \{\mu\in M^{(1)}\mid \mu(0) = o(0)\} \\
      &= \{o\}.
    \end{align*}
    For the second component of the domain-specific interpretation, we
    proceed by case distinction:
    \begin{itemize}
    \item If $M$ is non-cancellative, we have
      $\lambda_t(\delta,\rho)=\grad{t(2)}\delta\wedge \grad{t(1)}\rho$
      for $t\in M^{(3)}$. Thus, we obtain the following chain of
      equivalences for every $t'\in M^{(3)}$:
      \begin{align*}
        & t'\in ([t]_{F\chi_{\set{1,2}}}
        \cap \semantics{\lambda_t}_3(\set{2},\set{1}))
        \\ \Leftrightarrow~&
        t'\in [t]_{F\chi_{\set{1,2}}}
        \text{ and } t'\in \semantics{\lambda_t}_3(\set{2},\set{1})
        \\ \Leftrightarrow~&t'\in [t]_{F\chi_{\set{1,2}}}
        \text{ and } t'\in \semantics{(\delta,\rho)\mapsto \grad{t(2)}\delta\wedge \grad{t(1)}\rho}_3(\set{2},\set{1})
        &\text{(def.~$\lambda_t$)}
        \\ \Leftrightarrow~&t'\in [t]_{F\chi_{\set{1,2}}}
        \text{ and } t'\in \semantics{\grad{t(2)}}_3(\set{2}) \cap \semantics{\grad{t(1)}}_3(\set{1})
        \\ \Leftrightarrow~&t'\in [t]_{F\chi_{\set{1,2}}}
        \text{ and } t'\in \semantics{\grad{t(2)}}_3(\set{2})
        \text{ and } t'\in \semantics{\grad{t(1)}}_3(\set{1})
        \\ \Leftrightarrow~&t'\in [t]_{F\chi_{\set{1,2}}}
        \text{ and } t'(2) = t(2)
        \text{ and } t'(1) = t(1)
        &\text{(def.~$\semantics{\grad{m}}$)}
        \\ \Leftrightarrow~&t'(0) = t(0)\text{ and }t'(1)+t'(2) = t(1)+t(2)
        \text{ and }
        \\ & t'(2) = t(2)
        \text{ and } t'(1) = t(1)
        \\ \Leftrightarrow~&t'(0) = t(0)
        \text{ and } t'(2) = t(2)
        \text{ and } t'(1) = t(1)
        \\ \Leftrightarrow~&t' = t
        \\ \Leftrightarrow~&t' \in \{t\}.
      \end{align*}

    \item If $M$ is cancellative, we put \( \kappa_s(\delta) =
      \grad{s(1)}\,\delta \) for $s \in M^{(2)}$, which then induces
      $\lambda_t(\delta,\rho) = \grad{s(2)}\,\delta$ via
      \autoref{domainCertCancellative}. Hence, we verify that $\kappa$
      is part of a simple domain-specific interpretation
      (\autoref{domainCertSimple}). Indeed, for
      every $s'\in M^{(2)}$ we have the following chain of equivalences:
      \begin{align*}
        &s' \in ([s]_{F!} \cap \semantics{\kappa_s}_2(\{1\}))
        \\ \Leftrightarrow~&s' \in [s]_{F!} \text{ and } s'\in \semantics{\kappa_s}_2(\{1\})
        \\ \Leftrightarrow~&F!(s') = F!(s) \text{ and } s'\in \semantics{\grad{s(1)}}_2(\{1\})
        &\text{(def.~$\kappa_s$)}
        \\ \Leftrightarrow~&\textstyle F!(s') = F!(s) \text{ and } \sum_{x\in \{1\}}s'(x) = s(1)
        &\text{(def.~$\grad{s(1)}$)}
        \\ \Leftrightarrow~&F!(s') = F!(s) \text{ and } s'(1) = s(1)
        \\ \Leftrightarrow~&s'(0) + s'(1) = s(0) + s(1) \text{ and } s'(1) = s(1)
        \\ \Leftrightarrow~&s'(0) = s(0) \text{ and } s'(1) = s(1)
        &\text{($M$ cancellative)}
        \\ \Leftrightarrow~&s' = s
        \\ \Leftrightarrow~&s' \in \{s\}.
      \end{align*}
    \end{itemize}

  \item %
    For $FX=(\Dist X +1)^A$, recall that the predicate lifting
    $\semantics{\fpair{a}_p}$, where $a\in A$, $p\in[0,1]$, is given by
    \[
      \semantics{\fpair{a}_p}_X(S) = \{
      t\in FX\mid \text{if }p > 0 \text{, then }t(a) \in \Dist X\text{ and }
      \textstyle\sum_{x\in S}t(a)(x) \ge p
      \}.
    \]
    First note that
    \[
    \semantics{\fpair{a}_{1}\top}_1 = \{o\in F1\mid o(a) \in \Dist 1\}
    \qquad\text{ and }\qquad
    \semantics{\neg\fpair{a}_{1}\top}_1 = \{o\in F1\mid o(a) \in 1\}.
    \]
    Thus, we have:
    \begin{align*}
      \semantics{\tau_o}_1 &=
      \semantics[\big]{
      \bigwedge_{\substack{a\in A\\ o(a) \in \Dist 1}} \fpair{a}_{1}\top
      \wedge
      \bigwedge_{\substack{a\in A\\ o(a) \in 1}} \neg \fpair{a}_{1}\top
      }_1
      \\ & =
      \bigcap_{\substack{a\in A\\ o(a) \in \Dist 1}} \{o'\in F1\mid o'(a)\in \Dist 1\}
      \cap
      \bigcap_{\substack{a\in A\\ o(a) \in 1}} \{o'\in F1\mid o'(a)\in
      1\}
      \\
      & = \{o\}.  %
    \end{align*}
    For $\lambda_t$, $t\in F3 = (\Dist 3+1)^A$, we
    have the following chain of equivalences for every $t'\in F3$
    (note that the crucial step is the arithmetic argument for replacing
    the inequalities with equalities):
    \allowdisplaybreaks%
    \begin{align*}
      &t' \in ([t]_{F\chi_{\set{1,2}}}\cap \semantics{\lambda_t}_3(\set{2},\set{1}))
      \\ \Leftrightarrow~&
      t' \in [t]_{F\chi_{\set{1,2}}}\text{ and } t'\semantics{\lambda_t}_3(\set{2},\set{1})
      \\ \Leftrightarrow~&
      t' \in [t]_{F\chi_{\set{1,2}}} \text{ and }
                           t' \in
      \semantics[\big]{(\delta,\rho)\mapsto
      \bigwedge_{\substack{a\in A\\t(a) \in \Dist 3}}(\fpair{a}_{t(a)(2)}\,\delta\wedge \fpair{a}_{t(a)(1)}\,\rho)
      }_3(\set{2},\set{1})
      \\ \Leftrightarrow~&
      t' \in [t]_{F\chi_{\set{1,2}}} \text{ and }
                           t' \in
                           \bigcap_{\substack{a\in A\\t(a) \in \Dist 3}}
      \semantics{(\delta,\rho)\mapsto
      \fpair{a}_{t(a)(2)}\,\delta\wedge \fpair{a}_{t(a)(1)}\,\rho
      }_3(\set{2},\set{1})
      \\ \Leftrightarrow~&
      t' \in [t]_{F\chi_{\set{1,2}}} \text{ and }
                           t' \in
                           \bigcap_{\substack{a\in A\\t(a) \in \Dist 3}}
      \semantics{\fpair{a}_{t(a)(2)}}_3(\set{2})
      \cap\semantics{\fpair{a}_{t(a)(1)}}_3(\set{1})
      \\ \Leftrightarrow~&
      t' \in [t]_{F\chi_{\set{1,2}}} \text{ and }
                           \forall a\in A, t(a)\in \Dist 3:
                           t' \in
      \semantics{\fpair{a}_{t(a)(2)}}_3(\set{2})
      \cap\semantics{\fpair{a}_{t(a)(1)}}_3(\set{1})
      \\ \Leftrightarrow~&
      t' \in [t]_{F\chi_{\set{1,2}}} \text{ and }
                           \forall a\in A, t(a)\in \Dist 3:
                           t'(a)(2) \ge t(a)(2)
                           \wedge
                           t'(a)(1) \ge t(a)(1)
                           \\ &\text{(def.~$\semantics{\fpair{a}p}$)}
      \\ \Leftrightarrow~&
              \forall a \in A\colon (t'(a)\in 1\leftrightarrow t(a)\in 1)\text{ and if }t(a)\in \Dist 3\text{ then:}
                           \\ &\qquad
                           t'(a)(0) = t(a)(0),~~ t'(a)(1) + t'(a)(2) = t(a)(1)+t(a)(2),
                           \\ &\qquad
                           t'(a)(2) \ge t(a)(2),~~
                           t'(a)(1) \ge t(a)(1)
      \\ \Leftrightarrow~&
              \forall a \in A\colon (t'(a)\in 1\leftrightarrow t(a)\in 1)\text{ and if }t(a)\in \Dist 3\text{ then:}
                           \\ &\phantom{\forall a \in A\colon}
                           t'(a)(0) = t(a)(0),~~
                           t'(a)(1) = t(a)(1),~~
                           t'(a)(2) = t(a)(2)
                           \tag{arithmetic}
      \\ \Leftrightarrow~&
              \forall a \in A\colon (t'(a)\in 1\leftrightarrow t(a)\in 1)\text{ and if }t(a)\in \Dist 3\text{ then }
                           t'(a) = t(a)
      \\ \Leftrightarrow~& t' \in \set{t}.
                           \tag*{\qedhere}
    \end{align*}
  \end{enumerate}
\end{proof}
\noindent
The intuitive meaning of a domain-specific interpretation
(\autoref{domainCert}) is formalized by the following technical result.%
\begin{lemma}%
  \label{domainSpecificSetIndexed}
  Let $(\tau,\lambda)$ be a domain-specific interpretation for $F$.
  For all $t\in FC$ and $S\subseteq B\subseteq C$, we have
  \[
    \big([t]_{F\chi_B} ~\cap~
    \semantics{\lambda_{F\chi_S^B(t)}}_C(S,B\setminus S)\big)
    = [t]_{F\chi_S^B}
    \qquad\text{in $2^{FC}$.}
  \]
\end{lemma}
\begin{proof}
  Put $d= F\chi_S^B(t)$;
  the naturality square of $\semantics{\lambda_d}$ for $\chi_S^B\colon C\to 3$
  is
  \[
    \begin{tikzcd}
      2^3\times 2^3
      \arrow{r}{\semantics{\lambda_d}_3}
      \arrow{d}[swap]{2^{\chi_S^B}\times 2^{\chi_S^B}}
      & 2^{F3}
        \arrow{d}{2^{F\chi_S^B}}
        \\
        2^C\times 2^C
        \arrow{r}{\semantics{\lambda_d}_C}
        & 2^{FC}
      \end{tikzcd}
    \]
    Hence:
    \begin{align*}
      (F\chi_S^B)^{-1}\big[\semantics{\lambda_d}_3(\set{2},\set{1})\big]
      &=
      \semantics{\lambda_d}_C((\chi_S^B)^{-1}[\set{2}],(\chi_S^B)^{-1}[\set{1}])
      \\
      &= \semantics{\lambda_d}_C(B,B\setminus S).
      \tag{$*$}
    \end{align*}
    Now we have the following chain of equivalences for every $t'\in FC$:
    \allowdisplaybreaks%
    \begin{align*}
      &t' \in \big([t]_{F\chi_B} ~\cap~ \semantics{\lambda_{d}}_C(S,B\setminus S)\big)
      \\
  \Leftrightarrow~& t' \in [t]_{F\chi_B}\text{ and }
                    t'\in \semantics{\lambda_{d}}_C(S,B\setminus S)
      \\
  \Leftrightarrow~& t' \in [t]_{F\chi_B}\text{ and }
                    t'\in (F\chi_S^B)^{-1}\big[\semantics{\lambda_{d}}_3(\set{2},\set{1})\big]
                    &\text{(by~($*$))}
      \\
  \Leftrightarrow~& F\chi_S^B(t') \in [F\chi_S^B(t)]_{F\chi_{\set{1,2}}}\text{ and }
                    F\chi_S^B(t')\in \semantics{\lambda_{d}}_3(\set{2},\set{1})
                    &\text{($\chi_{\set{1,2}}\cdot \chi_S^B = \chi_B$)}
      \\
  \Leftrightarrow~& F\chi_S^B(t') \in [F\chi_S^B(t)]_{F\chi_{\set{1,2}}} \cap
                    \semantics{\lambda_{d}}_3(\set{2},\set{1})
      \\
  \Leftrightarrow~& F\chi_S^B(t') \in \{F\chi_S^B(t)\}
                    &\text{(\autoref{domainCert}, $d = F\chi_S^B(t)$)}
      \\
  \Leftrightarrow~& F\chi_S^B(t') = F\chi_S^B(t)
      \\
  \Leftrightarrow~& t' \in [t]_{F\chi_S^B}.
                    \tag*{\qedhere}
    \end{align*}
  \end{proof}

\noindent Given a domain-specific interpretation $(\tau,\lambda)$ for
a modal signature~$\Lambda$ for the set functor~$F$, we can postprocess
certificates~$\phi$ produced by \autoref{algoCerts} by replacing the
modalities $\fmod t$ for $t \in F3$ according to the translation~$T$
recursively defined by the following clauses for modalities and by
commutation with propositional operators:
\twnote{}
\[
  T\big(\fmod{t}(\top,\top)\big) = \tau_{F!(t)}
  \qquad
  T\big(\fmod{t}(\delta,\beta))
  =
  \lambda_t\big(T(\delta), T(\beta) \wedge \neg T(\delta)\big).
\]
Note that one can replace $T(\beta)\wedge\neg T(\delta)$ with $T(\beta)\wedge
\neg T(\delta')$ for the optimized $\delta'$ from
\autoref{cancelConjunct}; the latter conjunction has essentially the same size as $T(\delta)$.

The domain-specific modal signatures inherit a Hennessy-Milner Theorem.
\begin{proposition}\label{domainCertMainThm}
  \smnote{}
  For every certificate $\phi$ of a behavioural equivalence class of
  a given co\-al\-ge\-bra produced by either \autoref{algoCerts}
  or its optimization (\autoref{algoCertsCancel}),~$T(\phi)$ is also a certificate
  for that class.%
\end{proposition}
\begin{proof}
    We prove by induction on the index $i$ of main loop iterations that
    $T(\delta_i([x]_{P_i}))$ and $T(\beta_i([x]_{Q_i}))$ are certificates for $[x]_{P_i}$
    and $[x]_{Q_i}$, respectively. (In the
    cancellative case, $Q_i$ and $\beta_i$ are not defined; so just
    put~$C/Q_i=\{C\}$, $\beta_i(C) = \top$ for convenience.)
    \begin{enumerate}
    \item For $i=0$, we trivially have
      \[
        \semantics{T(\beta_0([x]_{P_i}))} =  \semantics{T(\top)} = \semantics{\top} = C.
      \]
      Furthermore, unravelling \autoref{notationF1Mod}, we have
      \[
        \delta_0([x]_{P_0}) = \fmod{F!(c(x))} = \fmod{Fj_1(F!(c(x)))}(\top,\top).
      \]
      Consequently,
      \[
        T(\delta_0([x]_{P_0})) = \tau_{F!(Fj_1(F!(c(x))))} = \tau_{F!(c(x))}
      \]
      using $!\cdot j_1\cdot \mathord{!} = \mathord{!}\colon C\to 1$. Naturality of
      $\semantics{\tau_o}$, $o\in F1$, implies that
      \[
        \semantics{\tau_o}_X = \{t\in FX\mid F!(t) = o\}.
      \]
      Hence, we obtain the desired equality:
      \[
        \semantics{T(\delta_0([x]_{P_0}))} =
        c^{-1}[\semantics{\tau_{F!(c(x))}}_C] = \{x'\in C\mid F!(c(x')) = F!(c(x))\}=[x]_{P_0}.
      \]

    \item In the inductive step, there is nothing to show for $\beta_{i+1}$
      because it is only a boolean combination of $\beta_{i}$ and $\delta_i$.
      For $\delta_{i+1}$, we distinguish two cases: whether the class
      $[x]_{P_i}$ is refined or not. If $[x]_{P_{i+1}} = [x]_{P_i}$, then
      \[
        \semantics{T(\delta_{i+1}([x]_{P_{i+1}}))}
        = \semantics{T(\delta_{i}([x]_{P_i}))}
        = [x]_{P_i},
      \]
      and we are done. Now suppose that $[x]_{P_{i+1}}\neq [x]_{P_i}$ in the $i$-th
      iteration with chosen $S\subsetneqq B\subseteq C$. By step~\ref{defDeltai1} of \autoref{algoCerts} (or \autoref{algoCertsCancel}, respectively), we have
      \[
        \delta_{i+1}([x]_{P_{i+1}}) = \delta_i([x]_{P_i}) \wedge
        \fmod{t}(\delta_i(S), \beta')
      \]
      where $\beta'$ is $\beta_i(B)$ or $\top$; in any case
      $\semantics{\delta_i(S)}= S\subseteq \semantics{\beta'}$.
      Note that $t$ here is either~$F\chi_S^B(c(x))$ (\autoref{algoCerts})
      or $Fj_2(F\chi_S(c(x)))$ (\autoref{algoCertsCancel}). Put $B'=B$ in
      the first case and $B'=C$ else. Using $\chi_S^{C} =
      j_2\cdot\chi_S$, we see that
      \[
        t = F\chi_S^{B'}(c(x))
        \qquad
        \semantics{\beta'} = B',
        \qquad\text{and}\qquad
        \semantics{T(\beta')} = B',
      \]
      where the last equation follows from the inductive hypothesis.
      Thus, we have
      \[
        \delta_{i+1}([x]_{P_{i+1}}) = \delta_i([x]_{P_i}) \wedge
        \fmod{F\chi_S^{B'}(c(x))}(\delta_i(S), \beta'),
      \]
      and therefore
      \[
        T(\delta_{i+1}([x]_{P_{i+1}})) = T(\delta_i([x]_{P_i})) \wedge
        \lambda_{F\chi_S^{B'}(c(x))}\big(T(\delta_i(S)), T(\beta')
        \wedge \neg T(\delta_i(S))\big).
      \]
      Moreover, we have
      \[
        P_{i+1} = P_i\cap \ker(F\chi_S^{B'}\cdot c),
      \]
      in the first case by step~\ref{defPi1}, and in the second case by
      \autoref{optimizedPcancel}, recalling that $\chi_S = \chi_S^C$.

      We are now prepared for our final computation:
      \allowdisplaybreaks%
      \begin{align*}
        & \semantics{T(\delta_{i+1}([x]_{P_{i+1}})) }
          \\
        =~& \semantics{T(\delta_i([x]_{P_i})) \wedge
        \lambda_{F\chi_S^{B'}(c(x))}(T(\delta_i(S)), T(\beta')\wedge \neg T(\delta_i(S)))}
            \\
        =~& \semantics{T(\delta_i([x]_{P_i}))} \cap
        \semantics{\lambda_{F\chi_S^{B'}(c(x))}(T(\delta_i(S)), T(\beta')\wedge \neg T(\delta_i(S)))}
            \\
        =~& \semantics{T(\delta_i([x]_{P_i}))} \cap
        c^{-1}\big[\semantics{\lambda_{F\chi_S^{B'}(c(x))}}_C(\semantics{T(\delta_i(S))},
        \semantics{T(\beta')}\cap C\setminus \semantics{T(\delta_i(S))})\big]
            \tag{semantics of $\lambda$}
            \\
        =~& [x]_{P_i} \cap
        c^{-1}\big[\semantics{\lambda_{F\chi_S^{B'}(c(x))}}_C(S, B'
        \cap (C\setminus S) )\big]
            &\tag{induction hypothesis}
            \\
        =~& [x]_{P_i} \cap
        c^{-1}\big[\semantics{\lambda_{F\chi_S^{B'}(c(x))}}_C(S, B'\setminus S)\big]
        \tag{$B' \cap (C\setminus S) = B'\setminus S$} 
        \\
        =~& [x]_{P_i} \cap
            [x]_{F\chi_{B'}\cdot c} \cap
        c^{-1}\big[\semantics{\lambda_{F\chi_S^{B'}(c(x))}}_C(S, B'\setminus S)\big]
            \tag{$P_i\subseteq \ker F\chi_{B'}\cdot c$} 
            \\
        =~& [x]_{P_i} \cap
            c^{-1}\big[[c(x)]_{F\chi_{B'}}\big] \cap
        c^{-1}\big[\semantics{\lambda_{F\chi_S^{B'}(c(x))}}_C(S, B'\setminus S)\big]
            \\
        =~& [x]_{P_i} \cap
            c^{-1}\big[[c(x)]_{F\chi_{B'}} \cap
        \semantics{\lambda_{F\chi_S^{B'}(c(x))}}_C(S, B'\setminus S)\big]
            \\
        =~& [x]_{P_i} \cap
            c^{-1}\big[[c(x)]_{F\chi_S^{B'}}\big]
            \tag{domain-specific interpret.~(\autoref{domainSpecificSetIndexed})}
            \\
        =~& [x]_{P_i} \cap
            [x]_{F\chi_S^{B'}\cdot c}
            \\
        =~& [x]_{P_{i+1}}
            \tag{$P_{i+1} = P_i\cap \ker(F\chi_S^{B'}\cdot c$)} 
      \end{align*}
    \end{enumerate}
    Thus, $\semantics{T(\delta_{i+1}([x]_{P_{i+1}})) }$ is a certificate.
  \end{proof}

\begin{example}\label{exCertMarkov}
  For labelled Markov chains ($FX=(\Dist X +1)^A$) and the
  interpretation via the modalities $\fpair{a}_p$
  (\itemref{exDomainSpecInt}{dsiMarkov}), we thus obtain certificates
  (in particular also distinguishing formulae) in run time
  $\CO(|A|\cdot m\cdot \log n)$, with the same bound on formula
  size. Indeed,%
  \smnote{}  the
  \autoref{algoCerts} runs in $\CO(m\cdot \log n)$ producing
  certificates of total size $\CO(m\cdot \log n)$. By
  \autoref{domainCertMainThm}, we can translate these certificates
  into ones using the modalities~$\fpair{a}_p$. The translation blows
  up the size of certificates by the additional factor~$|A|$ appearing
  in the above size estimate because of the big conjunctions in the
  domain-specific interpretation
  (\itemref{exDomainSpecInt}{dsiMarkov}).

  By comparison, the algorithm for distinguishing formulae by
  Deshar\-nais \etal{}~\cite[Fig.~4]{desharnaisEA02} runs roughly in
  time $\CO(|A|\cdot n^4)$, as it nests four loops over all blocks
  seen so far and one additional loop over~$A$.  The distinguishing
  formulae computed by the algorithm live in the negation-free
  fragment of the logic that we use for certificates; Desharnais \etal
  note that this fragment does not suffice for certificates.
\end{example}

\twnote[inline]{}

\section{Worst Case Tree Size of Certificates}%
\label{worstcase}
In the complexity analysis (\autoref{complexityAnalysis}), we have
seen that certificates -- and thus also distinguishing formulae -- 
have dag size $\CO(m\cdot \log n + n)$ on input coalgebras with~$n$
states and~$m$ transitions. However, when formulae are written in the
usual linear way, multiple occurrences of the same subformula lead to
an exponential blow up of the formula size in this sense, which for
emphasis we refer to as the \emph{tree size}.

\subsection{Transition Systems}%
\label{worstCaseTS}
The certificate of a state can be exponentially large and even the
size of formulae separating two particular states of interest is in
the worst case as big as the certificate of one the states. This has
been shown previously by Figueira and Gor{\'{\i}}n~\cite{FigueiraG10}
via winning strategies in bisimulation games, a technique that is also
applied in other works giving lower bounds for formula
size~\cite{FrenchEA13,AdlerImmermanLics01,AdlerImmermanTocl03}.  For
the convenience of the reader, we recall the example and give a direct
argument for the size estimate in the appendix.

\newcommand{\threetowertransitions}[2]{%
    \path[path with edges]
      (x #1) edge (x #2)
      (x #1) edge (y #2)
      (x #1) edge (z #2)
      (y #1) edge (z #2)
      (y #1) edge (y #2)
      (z #1) edge (z #2)
      (z #1) edge (x #2)
      ;
  }
  \newcommand{\threetowerstates}{%
    \begin{scope}[every node/.append style={state,shape=circle,inner sep=2pt}]
      \node[label=above:$z_{i+1}$] (z up)  at (0, 1) {$\blackcircle$};
      \node[label=above:$x_{i+1}$] (x up)  at (1, 1) {$\blackcircle$};
      \node[label=above:$y_{i+1}$] (y up)  at (2, 1) {$\blackcircle$};
      \node[label=below:$z_{i}$] (z down)  at (0, 0) {$\blackcircle$};
      \node[label=below:$x_{i}$] (x down)  at (1, 0) {$\blackcircle$};
      \node[label=below:$y_{i}$] (y down)  at (2, 0) {$\blackcircle$};
    \end{scope}
  }

\begin{expl}\label{threetower}
  We define a $\Powf$-coalgebra $(C,c)$ with state set $C = \bigcup_{i\in \N} L_i$
  made up of \emph{layers} $L_i = \{x_i,y_i,z_i\}$.  The successors of
  states in layer $L_{i+1}$ are in
  layer~$L_i$; specifically,
  \[
    \begin{aligned}
      c(x_0) &= \{y_0\}
      & c(x_{i+1}) &= (0, \{x_i,y_i,z_i\})
      \\
      c(y_0) &= \emptyset
      & c(y_{i+1}) &= (0, \{y_i,z_i\})
      \\
      c(z_0) &= \{x_0\}
      & c(z_{i+1}) &= (0, \{x_i,z_i\}).
    \end{aligned}
    \qquad\qquad
    \begin{minipage}{.3\textwidth}
    \begin{tikzpicture}[coalgebra,x=10mm, y=10mm,every label/.append style={inner ysep=3pt,inner xsep=0pt}]
    \threetowerstates{}
    \threetowertransitions{up}{down}
    \begin{scope}[yshift=-1cm]
    \node[label=below:$z_{0}$] (z 0)  at (0, 0) {$\blackcircle$};
    \node[label=below:$x_{0}$] (x 0)  at (1, 0) {$\blackcircle$};
    \node[label=below:$y_{0}$] (y 0)  at (2, 0) {$\blackcircle$};
    \end{scope}
    \path[path with edges]
      (x 0) edge[bend left=0] (y 0)
      (z 0) edge[bend left=0] (x 0)
      ;
    \end{tikzpicture}
  \end{minipage}
  \]
  No two distinct states of $(C,c)$ are bisimilar. For a lower bound
  on tree size of certificates, one shows that a certificate
  of~$x_{n+1}$ necessarily contains (distinct) certificates of~$x_{n}$
  and~$y_n$, and a certificate of~$y_{n+1}$ contains one
  of~$x_n$. Hence, every certificate of~$x_{n}$ has size at least
  $\fib(n)$, the $n$-th Fibonacci number. Moreover, by the above,
  every formula distinguishing~$x_{n+1}$ from~$y_{n+1}$ contains a
  certificate for~$x_n$, and thus also has size at least~$\fib(n)$.
\end{expl}

\noindent Cleaveland~\cite[p.~368]{Cleaveland91} also mentions that
minimal distinguishing formulae may be exponential in size, however
for a slightly different notion of minimality: a formula~$\phi$
distinguishing $x$ from $y$ is called \emph{minimal} by Cleaveland if
no $\phi$ obtained by replacing a non-trivial subformula of $\phi$
with the formula $\top$ distinguishes $x$ from $y$.  This is weaker
than demanding that the formula size of $\phi$ is as small as
possible. For example, in the transition system
\begin{center}
  \hspace{8mm}
  \begin{tikzpicture}[coalgebra,x=1.5cm,baseline=(x.base)]
    \begin{scope}[every node/.append style={
        state,
        label distance=-.5mm,
        outer sep=3pt,
        inner sep=0pt,
        text depth=0pt,
        font=\normalsize,%
      }
      ]
      \node[label=above:$x$] (x)  at (0, 0) {$\bullet$};
      \node (x1)  at (1, 0) {$\bullet$};
      \begin{scope}[xshift=3cm]
        \node[label=above:$y$]
        (y 0) at (0,0) {$\bullet$};
        \foreach \n in {1,3} {
          \node (y \n) at (\n,0) {$\bullet$};
        }
        \node (y 2) at (2,0) {$\cdots$};
      \end{scope}
    \end{scope}
    \path[path with edges]
      (x) edge (x1)
      (x) edge[out=150,in=200,looseness=6.5,overlay] (x)
      (y 0) edge (y 1)
      (y 1) edge[shorten >=1mm,-] (y 2)
      (y 2) edge[shorten <=1mm] (y 3)
      ;
    \draw [decorate,decoration={brace,amplitude=3pt,raise=-1pt},yshift=0pt]
    (y 1.north east) -- node[font=\normalsize,yshift=6pt] {$n$}
    (y 3.north west);
  \end{tikzpicture}
  \qquad\qquad for $n\in \N$,
\end{center}
the formula $\phi = \Diamond^{n+2}\top$ distinguishes $x$ from $y$
and is minimal in the above sense. \mbox{However, $x$} can in fact be
distinguished from~$y$ in size $\CO(1)$,
by the formula~$\Diamond \neg\Diamond\top$.

To verify the minimality of $\phi = \Diamond^{n+2}\top$,
one considers all possible replacements of subformulae of $\phi$ by $\top$:
\[
  \Diamond\top
  \qquad
  \Diamond\Diamond\top
  \qquad
  \ldots
  \qquad
  \Diamond^n\top
  \qquad
  \Diamond^{n+1}\top
\]
All of these hold at both $x$ and $y$, because $x$ can perform arbitrarily
many transitions, and~$y$ can perform $n+1$ transitions.

\subsection{Weighted Systems}
In contrast to transition systems, lower bounds on the size of distinguishing
formulae have not been established.

As a negative result, even the optimized algorithm for cancellative
functors presented above (\autoref{algoCertsCancel}) constructs
certificates of exponential worst-case tree size, even in cases where
linear-sized certificates exist.
\begin{expl}\label{algoExpFormula}
  Define the $\R^{(-)}$-coalgebra $c$ on $C =\bigcup_{k\in \N}\{w_k,x_k,y_k,z_k\}$ by
  \[
    \begin{array}{r@{\,}l@{\,}l@{\,}l@{\,}l@{\qquad\qquad}r@{}c@{\,}c@{\,}r}
      c(w_{k+1}) = & \{w_k\mapsto 1,& x_k\mapsto 2,& y_k\mapsto 1, &z_k\mapsto 2\}, & c(w_0) = \{&w_0&\mapsto & 1\},\\
      c(x_{k+1}) = & \{w_k\mapsto 1,& x_k\mapsto 2,& y_k\mapsto 2, &z_k\mapsto 1\}, & c(x_0) = \{&x_0&\mapsto &2\},\\
      c(y_{k+1}) = & \{w_k\mapsto 2,& x_k\mapsto 1,& y_k\mapsto 1, &z_k\mapsto 2\}, & c(y_0) = \{&y_0&\mapsto &3\},\\
      c(z_{k+1}) = & \{w_k\mapsto 2,& x_k\mapsto 1,& y_k\mapsto 2, &z_k\mapsto 1\}, & c(z_0) = \{&z_0&\mapsto &4\}.
    \end{array}
  \]
  We say that $L_k= \{w_k,x_k, y_k, z_k\}$ is the $k$-th \emph{layer}
  of this coalgebra. So the states in the $0$-th layer each just have
  a loop with weight $1$, $2$, $3$, and $4$, respectively, and the
  states of the $k+1$-st layer and the $k$-th one are connected by a
  complete bipartite graph with weights as indicated above.

  We now show that \autoref{algoCertsCancel} constructs a certificate
  of size $2^n$ in the $n$-th layer.%
  \smnote{}  Hence,
  for the finite subcoalgebra $L_0\cup \cdots \cup L_n$, the states
  in~$L_n$ receive certificates of size exponential in the size of the
  input coalgebra, that is the number $4 + 4n$ of states plus the
  number $4 + 16n$ of edges.

  To see this, first note that the initial partition
  \[
    P_0=\set[\big]{\set{w_0},\set{x_0}, \set{y_0}, \set{z_0},
      L_1\cup\cdots \cup L_n}
  \]
  distinguishes on the total out-degree (being
  1, 2, 3, 4, or 6). The states in $L_0$ are assigned the following
  certificates:
  \begin{equation}\label{eq:L0-certs}
    w_0 = \grad{1}\top,
    \qquad
    x_0 = \grad{2}\top,
    \qquad
    y_0 = \grad{3}\top,
    \qquad\text{and}\qquad
    z_0 = \grad{4}\top.
  \end{equation}
  Assume that after $i$ iterations of the
  main loop of the algorithm, the states $w_k, x_k, y_k, z_k$ have
  just been found to be behaviourally different and all states of
  $L_{k+1}\cup \cdots \cup L_{n}$ are still identified. Then the
  algorithm has to use some of the blocks $\set{w_k}$, $\set{x_k}$,
  $\set{y_k}$, $\set{z_k}$ as the splitter $S$ for further
  refinement. Assume wlog that the first block used as the splitter is
  $S = \{w_k\}$. Then the block $L_{k+1}\cup\cdots\cup L_{n}$
  will be refined into the blocks
  \[
    \{w_{k+1},x_{k+1}\},\qquad
    \{y_{k+1},z_{k+1}\},\qquad\text{and}\qquad
    L_{k+2}\cup\cdots\cup L_n.
  \]
  Denote by $\delta(\set{w_k})$ the formula that we have at this point
  for $\set{w_k}$.
  The definition of $\delta$ in the algorithm annotates the block
  $\{w_{k+1},x_{k+1}\}$ with $\grad{1}\delta(\set{w_k})$
  and the block $\{y_{k+1},z_{k+1}\}$ with $\grad{2}\delta(\set{w_k})$.

  Splitting by $\{x_k\}$ does not lead to further refinement. However,
  when splitting by $S = \{y_k\}$ (or equivalently $\{z_k\}$), we
  split $\set{w_{k+1}, x_{k+1}}$ into $\set{w_{k+1}}$ and
  $\set{x_{k+1}}$ and likewise $\set{y_{k+1}, z_{k+1}}$ into
  $\set{y_{k+1}}$ and $\set{z_{k+1}}$. Let $\delta(\set{y_k})$ be the
  certificate constructed for $\set{y_k}$. This implies that the
  formulae for $\set{w_{k+1}}$ and $\set{y_{k+1}}$ are both extended
  by the conjunct $\grad{1}\delta(\set{y_k})$; likewise, the formulae
  for $\set{x_{k+1}}$ and $\set{z_{k+1}}$ are extended by a new
  conjunct $\grad{2}\delta\set{y_k}$. Hence, for every $s\in L_{k+1}$
  the tree-size of the constructed formula is at least
  \[
    |\delta(\set{s})| \ge |\delta(\set{w_{k}})| + |\delta(\set{y_{k}})| .
  \]
  By induction we thus see that the certificates for $s \in L_k$ is of
  size at least $2^k$.
\end{expl}
\begin{rem}
  However, note that in \autoref{algoExpFormula},
  linear-sized certificates do exist for all states.

  Indeed, for states in $L_0$ we have the certificates
  from~\eqref{eq:L0-certs}. Using those, we can easily read off the
  following certificates for $w_1$ and $z_1$ from the coalgebra
  structure:
  \begin{align*}
    &\grad{1}\grad{1}\top \wedge \grad{2}\grad{2}\top \wedge
    \grad{1}\grad{3}\top \wedge \grad{2}\grad{4}\top, \ \text{and}
    \\
    &\grad{1}\grad{1}\top \wedge \grad{2}\grad{2}\top \wedge
    \grad{2}\grad{3}\top \wedge \grad{1}\grad{4}\top.
  \end{align*}
  For all other states, we obtain certificates as follows. For every
  $k \geq 0$ we have
  \[
    \phi_{k} := \grad{3}^k(\grad{1}\top \vee \grad{4}\top)
    \qquad\text{satisfying}\qquad
    \semantics{\phi_{k}} = \set{w_k,z_k}.
  \]
  This lets us define certificates for $\set{x_{k+1}}$ and $\set{y_{k+1}}$:
  \[
    \semantics{\grad{2}\phi_k} = \set{x_{k+1}}
    \quad\text{and}\quad
    \semantics{\grad{4}\phi_k} = \set{y_{k+1}}.
  \]
  For the remaining states $w_k$ and $z_k$,
  we note that
  \[
    \semantics{\grad{1}\grad{4}\phi_k} = \set{w_{k+2},y_{k+2}}.
  \]
  Thus, we have certificates
  \[
    \semantics{\phi_{k+2}\wedge \grad{1}\grad{4}\phi_k } = \set{w_{k+2}}
    \quad\text{and}\qquad
    \semantics{\phi_{k+2}\wedge \neg\grad{1}\grad{4}\phi_k } = \set{z_{k+2}}.
  \]
  Since $\phi_k$ involves $k+2$ modal operators, every state in $L_k$ has a certificate
  with at most $2\cdot k + 8$ modal operators.

  Hence, in each of the finite coalgebras $L_0 \cup \cdots \cup L_k$
  from \autoref{algoExpFormula}, every state has a certificate whose
  size is linearly bounded in the size of the coalgebra.
\end{rem}
\begin{oprob}%
  Do states in $\R^{(-)}$-coalgebras generally have certificates of
  subexponential tree size in the number of states? If yes, can small
  certificates be computed efficiently?
\end{oprob}

\noindent We note that for another cancellative functor, the answer is
well-known: On deterministic automata, i.e.~coalgebras for
$FX = 2\times X^A$, the standard minimization algorithm constructs
distinguishing words of linear length.

\section{Conclusions and Further Work}

We have presented a generic algorithm that computes distinguishing
formulae for behaviourally inequivalent states in state-based systems
of various types, cast as coalgebras for a functor capturing the
system type. Our algorithm is based on an efficient coalgebraic
partition refinement algorithm~\cite{concurSpecialIssue}, and like
that algorithm runs in time $\CO((m+n)\cdot \log n \cdot p(c))$, with
a functor-specific factor $p(c)$ that is $1$ in many cases of
interest. Independently of this factor, the distinguishing formulae
constructed by the algorithm have dag size $\CO(m\cdot \log n + n)$;
they live in a dedicated instance of coalgebraic modal
logic~\cite{Pattinson04,Schroder08}, with binary modalities extracted
from the type functor in a systematic way.
We have also introduced the notion of a \emph{cancellative} functor,
and we have shown that for such functors, the construction of formulae
and, more importantly, the logic employed can be simplified, requiring
only unary modalities and conjunction. We have also discussed how
distinguishing formulae can be translated into a more familiar
domain-specific syntax (e.g. Hennessy-Milner logic for transition
systems).

There is a proof-of-concept implementation of the certificate
construction, based on an earlier open source implementation of the
underlying partition refinement algorithm~\cite{coparFM19,wdms21}.

In partition refinement, blocks are successively refined in a top-down
manner, and this is reflected by the use of conjunction in
distinguishing formulae. Alternatively, bisimilarity may be computed
bottom-up, as in a recent partition \emph{aggregation}
algorithm~\cite{BjorklundCleophas2020}. It is an interesting point for
future investigation whether this algorithm can be extended to compute
distinguishing formulae, which would likely be of a rather different
shape than those computed via partition refinement.%
\label{maintextend} %
\bibliographystyle{alphaurl} %
\bibliography{refs}

\providecommand{\noopsort}[1]{}
\begin{thebibliography}{VGRW22}

\bibitem[ABDG14]{ArmasCervantesEA14}
Abel {Armas-Cervantes}, Paolo Baldan, Marlon Dumas, and Luciano
  {Garc{\'i}a-Ba{\~{n}}uelos}.
\newblock Behavioral comparison of process models based on canonically reduced
  event structures.
\newblock In {\em Business Process Management}, pages 267--282. Springer, 2014.

\bibitem[ABLM12]{AdamekEA12}
Ji\v{r}\'i Ad\'amek, Nathan Bowler, Paul~B. Levy, and Stefan Milius.
\newblock Coproducts of monads on {S}et.
\newblock In {\em Proc.~27th Annual Symposium on Logic in Computer Science
  (LICS'12)}, pages 45--54. IEEE Computer Society, 2012.

\bibitem[AGD13]{ArmasCervantesEA13}
Abel {Armas-Cervantes}, Luciano {Garc{\'i}a-Ba{\~{n}}uelos}, and Marlon Dumas.
\newblock Event structures as a foundation for process model differencing, part
  1: Acyclic processes.
\newblock In {\em Web Services and Formal Methods}, pages 69--86. Springer,
  2013.

\bibitem[AI01]{AdlerImmermanLics01}
Micah Adler and Neil Immerman.
\newblock An \emph{n!} lower bound on formula size.
\newblock In {\em LICS 2001}, pages 197--206. {IEEE} Computer Society, 2001.

\bibitem[AI03]{AdlerImmermanTocl03}
Micah Adler and Neil Immerman.
\newblock An \emph{n!} lower bound on formula size.
\newblock {\em {ACM} Trans. Comput. Log.}, 4(3):296--314, 2003.

\bibitem[AM89]{AczelMendler89}
Peter Aczel and Nax Mendler.
\newblock A final coalgebra theorem.
\newblock In {\em Proc.~Category Theory and Computer Science (CTCS)}, volume
  389 of {\em LNCS}, pages 357--365. Springer, 1989.

\bibitem[Bar93]{Barr93}
Michael Barr.
\newblock Terminal coalgebras in well-founded set theory.
\newblock {\em Theoretical Computer Science}, 114(2):299--315, 1993.

\bibitem[BC20]{BjorklundCleophas2020}
Johanna Bj\"{o}rklund and Loek Cleophas.
\newblock Aggregation-based minimization of finite state automata.
\newblock {\em Acta Informatica}, 2020.

\bibitem[BCSS98]{BernardoEA98}
Marco Bernardo, Rance Cleaveland, Steve Sims, and W.~Stewart.
\newblock Two{T}owers: {A} tool integrating functional and performance analysis
  of concurrent systems.
\newblock In {\em Formal Description Techniques and Protocol Specification,
  Testing and Verification, {FORTE} / {PSTV} 1998}, volume 135 of {\em {IFIP}
  Conference Proceedings}, pages 457--467. Kluwer, 1998.

\bibitem[Ber04]{Bernardo04}
Marco Bernardo.
\newblock Two{T}owers 5.1 user manual, 2004.

\bibitem[BKR19]{BarloccoEA19}
Simone Barlocco, Clemens Kupke, and Jurriaan Rot.
\newblock Coalgebra learning via duality.
\newblock In Mikolaj Bojanczyk and Alex Simpson, editors, {\em Foundations of
  Software Science and Computation Structures, {FOSSACS} 2019}, volume 11425 of
  {\em LNCS}, pages 62--79. Springer, 2019.
\newblock \href {https://doi.org/10.1007/978-3-030-17127-8\_4}
  {\path{doi:10.1007/978-3-030-17127-8\_4}}.

\bibitem[BM19]{BernardoMiculan19}
Marco Bernardo and Marino Miculan.
\newblock Constructive logical characterizations of bisimilarity for reactive
  probabilistic systems.
\newblock {\em Theoretical Computer Science}, 764:80 -- 99, 2019.
\newblock Selected papers of ICTCS 2016.

\bibitem[BSdV04]{BartelsEA04}
Falk Bartels, Ana Sokolova, and Erik de~Vink.
\newblock A hierarchy of probabilistic system types.
\newblock {\em Theoret.\ Comput.\ Sci.}, 327:3--22, 2004.

\bibitem[CC95]{CelikkanCleaveland95}
Ufuk Celikkan and Rance Cleaveland.
\newblock Generating diagnostic information for behavioral preorders.
\newblock {\em Distributed Computing}, 9(2):61--75, 1995.
\newblock \href {https://doi.org/10.1007/s004460050010}
  {\path{doi:10.1007/s004460050010}}.

\bibitem[Cle91]{Cleaveland91}
Rance Cleaveland.
\newblock On automatically explaining bisimulation inequivalence.
\newblock In {\em Computer-Aided Verification}, pages 364--372. Springer, 1991.
\newblock \href {https://doi.org/10.1007/BFb0023750}
  {\path{doi:10.1007/BFb0023750}}.

\bibitem[CLW15]{cranen_et_al:LIPIcs:2015:5408}
Sjoerd Cranen, Bas Luttik, and Tim A.~C. Willemse.
\newblock {Evidence for Fixpoint Logic}.
\newblock In {\em 24th EACSL Annual Conference on Computer Science Logic (CSL
  2015)}, volume~41 of {\em LIPIcs}, pages 78--93. Schloss
  Dagstuhl--Leibniz-Zentrum f{\"u}r Informatik, 2015.
\newblock \href {https://doi.org/10.4230/LIPIcs.CSL.2015.78}
  {\path{doi:10.4230/LIPIcs.CSL.2015.78}}.

\bibitem[DEP98]{DesharnaisEA98}
J.~{Desharnais}, A.~{Edalat}, and P.~{Panangaden}.
\newblock A logical characterization of bisimulation for labeled markov
  processes.
\newblock In {\em Proceedings. Thirteenth Annual IEEE Symposium on Logic in
  Computer Science (Cat. No.98CB36226)}, pages 478--487, 1998.
\newblock \href {https://doi.org/10.1109/LICS.1998.705681}
  {\path{doi:10.1109/LICS.1998.705681}}.

\bibitem[DEP02]{desharnaisEA02}
Josée Desharnais, Abbas Edalat, and Prakash Panangaden.
\newblock Bisimulation for labelled markov processes.
\newblock {\em Information and Computation}, 179(2):163--193, 2002.
\newblock \href {https://doi.org/10.1006/inco.2001.2962}
  {\path{doi:10.1006/inco.2001.2962}}.

\bibitem[Dij08]{Dijkman08}
Remco Dijkman.
\newblock Diagnosing differences between business process models.
\newblock In {\em Business Process Management}, pages 261--277, Berlin,
  Heidelberg, 2008. Springer Berlin Heidelberg.

\bibitem[DMSW17]{DorschEA17}
Ulrich Dorsch, Stefan Milius, Lutz Schr{\"o}der, and Thorsten Wi{\ss}mann.
\newblock Efficient coalgebraic partition refinement.
\newblock In {\em Proc.~28th International Conference on Concurrency Theory
  (CONCUR 2017)}, LIPIcs. Schloss Dagstuhl - Leibniz-Zentrum f{\"u}r
  Informatik, 2017.
\newblock \href {https://doi.org/10.4230/LIPIcs.CONCUR.2017.32}
  {\path{doi:10.4230/LIPIcs.CONCUR.2017.32}}.

\bibitem[DMSW18]{DorschEA18}
Ulrich Dorsch, Stefan Milius, Lutz Schr{\"{o}}der, and Thorsten Wi{\ss}mann.
\newblock Predicate liftings and functor presentations in coalgebraic
  expression languages.
\newblock In {\em Coalgebraic Methods in Computer Science, {CMCS} 2018}, volume
  11202 of {\em LNCS}, pages 56--77. Springer, 2018.

\bibitem[DMSW19]{coparFM19}
Hans-Peter Deifel, Stefan Milius, Lutz Schr{\"o}der, and Thorsten Wi{\ss}mann.
\newblock Generic partition refinement and weighted tree automata.
\newblock In {\em Formal Methods -- The Next 30 Years, Proc.~3rd World Congress
  on Formal Methods (FM 2019)}, volume 11800 of {\em LNCS}, pages 280--297.
  Springer, 10 2019.

\bibitem[Dob09]{Doberkat09}
Ernst{-}Erich Doberkat.
\newblock {\em Stochastic Coalgebraic Logic}.
\newblock Springer, 2009.

\bibitem[FG10]{FigueiraG10}
Santiago Figueira and Daniel Gor{\'{\i}}n.
\newblock On the size of shortest modal descriptions.
\newblock In {\em Advances in Modal Logic 8, papers from the eighth conference
  on "Advances in Modal Logic," held in Moscow, Russia, 24-27 August 2010},
  pages 120--139. College Publications, 2010.

\bibitem[FvIK13]{FrenchEA13}
Tim French, Wiebe {van der Hoek}, Petar Iliev, and Barteld Kooi.
\newblock On the succinctness of some modal logics.
\newblock {\em Artificial Intelligence}, 197:56 -- 85, 2013.

\bibitem[GJ21]{GeuversJacobs21}
Herman Geuvers and Bart Jacobs.
\newblock Relating apartness and bisimulation.
\newblock {\em Logical Methods in Computer Science}, Volume 17, Issue 3, July
  2021.
\newblock \href {https://doi.org/10.46298/lmcs-17(3:15)2021}
  {\path{doi:10.46298/lmcs-17(3:15)2021}}.

\bibitem[Gri73]{Gries1973}
David Gries.
\newblock Describing an algorithm by {H}opcroft.
\newblock {\em Acta Informatica}, 2:97--109, 1973.

\bibitem[GS01]{GummS01}
H.~Peter Gumm and Tobias Schr{\"{o}}der.
\newblock Monoid-labeled transition systems.
\newblock In {\em Coalgebraic Methods in Computer Science, {CMCS} 2001}, volume
  44(1) of {\em ENTCS}, pages 185--204. Elsevier, 2001.

\bibitem[GS13]{GorinSchroeder13}
Daniel Gor{\'{\i}}n and Lutz Schr{\"{o}}der.
\newblock Simulations and bisimulations for coalgebraic modal logics.
\newblock In {\em Algebra and Coalgebra in Computer Science - 5th International
  Conference, {CALCO} 2013}, volume 8089 of {\em LNCS}, pages 253--266.
  Springer, 2013.

\bibitem[Gum05]{Gumm2005}
H.Peter Gumm.
\newblock From {$T$}-coalgebras to filter structures and transition systems.
\newblock In {\em Algebra and Coalgebra in Computer Science}, volume 3629 of
  {\em LNCS}, pages 194--212. Springer, 2005.

\bibitem[Hop71]{Hopcroft71}
John Hopcroft.
\newblock An $n \log n$ algorithm for minimizing states in a finite automaton.
\newblock In {\em Theory of Machines and Computations}, pages 189--196.
  Academic Press, 1971.

\bibitem[Kli05]{Klin05}
Bartek Klin.
\newblock The least fibred lifting and the expressivity of coalgebraic modal
  logic.
\newblock In {\em Algebra and Coalgebra in Computer Science, {CALCO} 2005},
  volume 3629 of {\em LNCS}, pages 247--262. Springer, 2005.

\bibitem[KMMS20]{KoenigEA20}
Barbara K{\"o}nig, Christina Mika-Michalski, and Lutz Schr{\"o}der.
\newblock Explaining non-bisimilarity in a coalgebraic approach: Games and
  distinguishing formulas.
\newblock In {\em Coalgebraic Methods in Computer Science}, pages 133--154.
  Springer, 2020.

\bibitem[Knu01]{Knuutila2001}
Timo Knuutila.
\newblock Re-describing an algorithm by {H}opcroft.
\newblock {\em Theor.\ Comput.\ Sci.}, 250:333 -- 363, 2001.

\bibitem[KS83]{KanellakisSmolka83}
Paris~C. Kanellakis and Scott~A. Smolka.
\newblock Ccs expressions, finite state processes, and three problems of
  equivalence.
\newblock In {\em Proceedings of the Second Annual ACM Symposium on Principles
  of Distributed Computing}, PODC '83, pages 228--240. ACM, 1983.

\bibitem[KS90]{KanellakisS90}
Paris~C. Kanellakis and Scott~A. Smolka.
\newblock {CCS} expressions, finite state processes, and three problems of
  equivalence.
\newblock {\em Inf. Comput.}, 86(1):43--68, 1990.

\bibitem[LAS91]{LarsenS91}
Kim~Guldstrand Larsen and Arne Arne~Skou.
\newblock Bisimulation through probabilistic testing.
\newblock {\em Inform.\ Comput.}, 94(1):1--28, 1991.

\bibitem[Mil89]{Milner89}
R.~Milner.
\newblock {\em Communication and Concurrency}.
\newblock International series in computer science. Prentice-Hall, 1989.

\bibitem[MV15]{MartiVenema15}
Johannes Marti and Yde Venema.
\newblock Lax extensions of coalgebra functors and their logic.
\newblock {\em J.\ Comput.\ Syst.\ Sci.}, 81(5):880--900, 2015.

\bibitem[Par81]{Park81}
D.~Park.
\newblock Concurrency and automata on infinite sequences.
\newblock In {\em Proceedings of 5th GI-Conference on Theoretical Computer
  Science}, volume 104 of {\em LNCS}, pages 167--183, 1981.

\bibitem[Pat03]{Pattinson03}
Dirk Pattinson.
\newblock Coalgebraic modal logic: soundness, completeness and decidability of
  local consequence.
\newblock {\em Theoretical Computer Science}, 309(1):177 -- 193, 2003.
\newblock \href {https://doi.org/10.1016/S0304-3975(03)00201-9}
  {\path{doi:10.1016/S0304-3975(03)00201-9}}.

\bibitem[Pat04]{Pattinson04}
Dirk Pattinson.
\newblock Expressive logics for coalgebras via terminal sequence induction.
\newblock {\em Notre Dame J.\ Formal Log.}, 45(1):19--33, 2004.

\bibitem[PT87]{PaigeTarjan87}
Robert Paige and Robert~E.\ Tarjan.
\newblock Three partition refinement algorithms.
\newblock {\em SIAM J.~Comput.}, 16(6):973--989, 1987.

\bibitem[RdV99]{RuttenDV99}
Jan Rutten and Erik de~Vink.
\newblock Bisimulation for probabilistic transition systems: a coalgebraic
  approach.
\newblock {\em Theoret.\ Comput.\ Sci.}, 221:271--293, 1999.

\bibitem[Rut00]{Rutten00}
Jan Rutten.
\newblock Universal coalgebra: a theory of systems.
\newblock {\em Theor.\ Comput.\ Sci.}, 249:3--80, 2000.

\bibitem[Sch01]{SchroederPhd01}
Tobias Schröder.
\newblock {\em Coalgebren und Funktoren}.
\newblock PhD thesis, Philipps-Universität Marburg, 2001.
\newblock \href {https://doi.org/10.17192/z2001.0205}
  {\path{doi:10.17192/z2001.0205}}.

\bibitem[Sch08]{Schroder08}
Lutz Schr{\"{o}}der.
\newblock Expressivity of coalgebraic modal logic: The limits and beyond.
\newblock {\em Theor.\ Comput.\ Sci.}, 390(2-3):230--247, 2008.
\newblock \href {https://doi.org/10.1016/j.tcs.2007.09.023}
  {\path{doi:10.1016/j.tcs.2007.09.023}}.

\bibitem[SP11]{SchroderPattinson11}
Lutz Schr{\"{o}}der and Dirk Pattinson.
\newblock Modular algorithms for heterogeneous modal logics via multi-sorted
  coalgebra.
\newblock {\em Math.\ Struct.\ Comput.\ Sci.}, 21(2):235--266, 2011.
\newblock \href {https://doi.org/10.1017/S0960129510000563}
  {\path{doi:10.1017/S0960129510000563}}.

\bibitem[Trn71]{trnkova71}
V\v{e}ra Trnkov\'a.
\newblock On a descriptive classification of set functors {I}.
\newblock {\em Commentationes Mathematicae Universitatis Carolinae},
  12(1):143--174, 1971.

\bibitem[VF10]{ValmariF10}
Antti Valmari and Giuliana Franceschinis.
\newblock Simple {$\CO(m\log n)$} time {M}arkov chain lumping.
\newblock In {\em Tools and Algorithms for the Construction and Analysis of
  Systems, TACAS 2010}, volume 6015 of {\em LNCS}, pages 38--52. Springer,
  2010.

\bibitem[VGRW22]{VGRW22}
Frits Vaandrager, Bharat Garhewal, Jurriaan Rot, and Thorsten Wi{\ss}mann.
\newblock A new approach for active automata learning based on apartness.
\newblock In {\em Tools and Algorithms for the Construction and Analysis of
  Systems - 28th International Conference, {TACAS} 2022}, Lecture Notes in
  Computer Science. Springer, 04 2022.

\bibitem[VL08]{ValmariLehtinen08}
Antti Valmari and Petri Lehtinen.
\newblock Efficient minimization of dfas with partial transition.
\newblock In {\em Theoretical Aspects of Computer Science, {STACS} 2008},
  volume~1 of {\em LIPIcs}, pages 645--656. Schloss Dagstuhl -- Leibniz-Zentrum
  für Informatik, Germany, 2008.

\bibitem[WDMS20]{concurSpecialIssue}
Thorsten Wißmann, Ulrich Dorsch, Stefan Milius, and Lutz Schröder.
\newblock {Efficient and Modular Coalgebraic Partition Refinement}.
\newblock {\em {Logical Methods in Computer Science}}, {Volume 16, Issue 1},
  January 2020.
\newblock \href {https://doi.org/10.23638/LMCS-16(1:8)2020}
  {\path{doi:10.23638/LMCS-16(1:8)2020}}.

\bibitem[WDMS21]{wdms21}
Thorsten Wi\ss\/mann, Hans-Peter Deifel, Stefan Milius, and Lutz Schr\"oder.
\newblock From generic partition refinement to weighted tree automata
  minimization.
\newblock {\em Form.~Asp.~Comput.}, 33:695--727, 2021.

\bibitem[WMS21]{WissmannEA21}
Thorsten Wi{\ss}mann, Stefan Milius, and Lutz Schr{\"{o}}der.
\newblock Explaining behavioural inequivalence generically in quasilinear time.
\newblock In Serge Haddad and Daniele Varacca, editors, {\em 32nd International
  Conference on Concurrency Theory, {CONCUR} 2021}, volume 203 of {\em LIPIcs},
  pages 32:1--32:18, 2021.
\newblock \href {https://doi.org/10.4230/LIPIcs.CONCUR.2021.32}
  {\path{doi:10.4230/LIPIcs.CONCUR.2021.32}}.

\end{thebibliography}

\clearpage
\appendix

\section{Worst Case Tree Size of Certificates for Transition Systems}\label{app:tree-size}

In the following, we provide a direct proof that the states of
\autoref{threetower} only admit distinguishing formulae of at least exponential
size.
As the worst case size of a certificate, we will derive the minimum size of
a certificate of $x_{i}$. The crucial argument here is that there is no
certificate of $x_i$ with only one top-level modality. To this end, we look at
all formulae with only one top-level modality:
\begin{equation}
  \Diamond \varphi \qquad\text{and}\qquad
  \neg\Diamond \varphi%
  \label{xCertShapes}
\end{equation}
for some formula $\varphi$ (note that these patterns also subsume
$\Box$-formulae). Given the semantics of $\varphi$ in the layer $L_i$,
the semantics of the formulae from~\eqref{xCertShapes} is clear:
\begin{table}[h!]
\begin{center}
  \def\arraystretch{1.2}
  \begin{tabular}[b]{@{}L@{\hspace{8mm}}L@{\hspace{8mm}}L@{}}
    \toprule
    \semantics{\varphi} \cap L_i
    & \semantics{\Diamond\varphi} \cap L_{i+1}
    & \semantics{\neg\Diamond\varphi} \cap L_{i+1}
    \\
      \midrule
    \emptyset
    & \emptyset & L_{i+1}
    \\
    \{x_i\} & \{x_{i+1},z_{i+1}\} & \{y_{i+1}\}
    \\
    \{y_i\} & \{x_{i+1},y_{i+1}\} & \{z_{i+1}\}
    \\
    \{z_i\} & L_{i+1} & \emptyset
    \\
    \{x_i,y_i\} & L_{i+1} & \emptyset
    \\
    \{x_i,z_i\} & L_{i+1} & \emptyset
    \\
    \{y_i,z_i\} & L_{i+1} & \emptyset
    \\
    \{x_i,y_i,z_i\} & L_{i+1} & \emptyset
    \\
    \bottomrule
  \end{tabular}
  \qquad
  \begin{tikzpicture}[coalgebra,x=15mm, y=18mm]
    \threetowerstates{}
    \threetowertransitions{up}{down}
  \end{tikzpicture}
\end{center}
\caption{Semantics of modalities in \autoref{threetower}}%
\label{tableThreeTower}
\end{table}

\noindent Of course, $\semantics{\neg\Diamond\varphi}$ is just the complement of
$\semantics{\Diamond\varphi}$. Since every state of $L_{i+1}$ has a
transition to $z_i$ we have that
$z_i \in \semantics{\varphi} \cap L_i$ implies
$\semantics{\Diamond\varphi} = L_{i+1}$.

\begin{lemma}\label{threetowerAllDifferent}
All states of $(C,c)$ in \autoref{threetower} have different behaviour.
\end{lemma}

\begin{proof}
  We show by induction on $i$ that all states of $L_0\cup\cdots \cup L_i$
  have different behaviour.
\begin{itemize}
\item All states of $L_0$ have different behaviour because
  \begin{itemize}
  \item $x_{0}\not\sim y_{0}$ because $x_0\to y_0$ but $y_0\not\to$.
  \item $z_{0}\not\sim y_{0}$ because $z_0\to x_0$ but $y_0\not\to$.
  \item $x_{0}\not\sim z_{0}$ because $z_0\to x_0$ and $z_0\to y_0$ but
    $x_0\not\sim y_0$.
  \end{itemize}

\item Assume that all elements of $L_0\cup\cdots \cup L_i$ have different
  behaviour. Every element of $L_{i+1}$ is behaviourally different from every state
  in $L_0\cup\cdots \cup L_i$ because every state in $L_{i+1}$ has a transition
  to $z_i$, but no state in $L_0\cup\cdots \cup L_i$ has
  and (by the induction hypothesis) there is no other state in
  $L_0\cup\cdots \cup L_i$ that is bisimilar to $z_i$.
  Since all state of $L_i$ have different behaviour,
  the same holds for $L_{i+1}$:
  \begin{itemize}
  \item $x_{i+1}\not\sim y_{i+1}$ because $x_{i+1}\to x_i$ but $y_{i+1}\not\to x_i$.
  \item $x_{i+1}\not\sim z_{i+1}$ because $x_{i+1}\to y_i$ but $z_{i+1}\not\to y_i$.
  \item $y_{i+1}\not\sim z_{i+1}$ because $y_{i+1}\to y_i$ but $z_{i+1}\not\to y_i$.
  \end{itemize}
  Hence, all states of $L_0\cup\cdots\cup L_i\cup L_{i+1}$ have different behaviour.
\qedhere
\end{itemize}
\end{proof}

\begin{lemma}\label{dnfSubformula} Let~$\varphi$ be a formula, and
  let $a,b\in C$ such that $a\in \semantics{\varphi}$ and
  $b\notin \semantics{\varphi}$. Then~$\varphi$ has a subformula~$\alpha$ such that
  \begin{enumerate}
  \item $\alpha$ has one of the forms $\Diamond\beta$ or
    $\neg\Diamond\beta$;
  \item $\alpha$ appears at the top level, i.e.\ outside the scope of
    any modality, in~$\varphi$;
  \item $a\in \semantics{\alpha}$ and  $b\notin \semantics{\alpha}$.
  \end{enumerate}
\end{lemma}
\begin{proof}
  We may assume without loss of generality that $\varphi$ is in
  disjunctive normal form on the top level, i.e.\ $\varphi$ is a
  disjunction of conjunctions (called \emph{conjunctive clauses}) of
  possibly negated $\Diamond$-formulae, since transforming~$\varphi$ into
  such a disjunctive normal form does not create new
  $\Diamond$-subformulae.

  Since $a\in \semantics{\varphi}$, the disjunctive normal
  form~$\varphi$ must contain a conjunctive clause~$\psi$ such that
  $a\in \semantics{\psi}$. Since $b\notin \semantics{\varphi}$, we
  necessarily have $b\notin\semantics{\psi}$. Since
  $b\notin \semantics{\psi}$, the conjunctive clause $\psi$ must
  contain a conjunct $\alpha$, of shape either $\Diamond\beta$ or
  $\neg\Diamond\beta$, such that $b\notin\semantics{\alpha}$. Since
  $a\in\semantics{\psi}$, we moreover necessarily have
  $a\in\semantics{\alpha}$; this proves the claim.
\end{proof}
\begin{proposition}\label{xCertShape} Let $\varphi$ be a formula such
  that $\semantics{\varphi}\cap L_{i+1} = \{x_{i+1}\}$. Then $\varphi$
  contains top-level subformulae (i.e.\ subformulae not in scope of a
  modality) \(\Diamond \varphi_{x_i} \) and
  \( \Diamond \varphi_{y_i} \) such that
  $\semantics{\varphi_{x_i}}\cap L_i = \{x_i\}$ and
  $\semantics{\varphi_{y_i}}\cap L_i = \{y_i\}$.
\end{proposition}
\begin{proof}
  By \autoref{dnfSubformula}, applied to $a = x_{i+1}$ and
  $b = y_{i+1}$, there must be a subformula $\alpha$ of $\varphi$ such
  that \(y_{i+1} ~\notin ~\semantics{\alpha} ~\ni~ x_{i+1}\)
  and~$\alpha$ has the form of either $\Diamond\psi$ or
  $\neg\Diamond\psi$. Since $x_{i+1},y_{i+1} \in L_{i+1}$, we also
  have
  \[
    y_{i+1} ~\notin ~\semantics{\alpha}\cap L_{i+1} ~\ni~ x_{i+1}.
  \]
  Looking at \autoref{tableThreeTower}, we find that the only choice
  for $\alpha$ is $\Diamond \psi$ for some formula $\psi$ satisfying
  $\semantics{\psi}\cap L_i = \{x_i\}$; so~$\psi$ may serve as the
  desired $\varphi_{x_i}$. Similarly, since
  $z_{i+1}\notin\semantics{\varphi}\owns x_{i+1}$, the formula $\varphi$ must have a
  subformula $\alpha'$ of the form either $\Diamond\chi$ or
  $\neg\Diamond\chi$ such that
  \[
    z_{i+1} ~\notin ~\semantics{\alpha'}\cap L_{i+1} ~\ni~ x_{i+1},
  \]
  which by \autoref{tableThreeTower} implies that
  $\semantics{\chi} \cap L_i = \{y_i\}$, so that $\chi$ may serve
  as~$\varphi_{y_i}$.
  \qedhere
\end{proof}

\begin{proposition}\label{yCertShape} Let $\varphi$ be a formula such
  that $\semantics{\varphi}\cap L_{i+1} = \{y_{i+1}\}$. Then $\varphi$
  has a top-level subformula of the form \(\Diamond \varphi_{x_i} \)
  such that $\semantics{\varphi_{x_i}}\cap L_i = \{x_i\}$.
\end{proposition}

\begin{proof}
  Similarly as in the proof of the previous proposition, since
  $x_{i+1}\not \in \semantics{\varphi}\owns y_{i+1}$, the formula
  $\varphi$ must, by \autoref{dnfSubformula}, have a subformula
  $\alpha$ of the shape either $\Diamond\psi$ or $\neg\Diamond\psi$
  such that
  \[
    x_{i+1} ~\notin ~\semantics{\alpha'}\cap L_{i+1} ~\ni~ y_{i+1}.
  \]
  Looking at \autoref{tableThreeTower} we find that this implies that
  $\alpha$ has the form $\neg\Diamond\psi$ for some formula $\psi$ such that
  $\semantics{\psi}\cap L_i = \{x_i\}$, so that $\psi$ may serve
  as~$\varphi_{x_i}$.
\end{proof}
\begin{proposition}\label{fibCertSize}
  Every certificate of $x_{n}$ has at least the size $\fib(n)$.
\end{proposition}
\begin{proof}
  We show more generally that if~$\varphi$ is a formula such that
  $\semantics{\varphi}\cap L_n = \{x_n\}$, then~$\varphi$ has size at
  least $\fib(n)$. We proceed by induction on~$n$, with trivial base
  cases $n\in\{0,1\}$.

  In the inductive step, we assume the statement for $n$ and $n+1$ and
  prove it for $n+2$.
  If $\varphi$ is a formula such that
  $\semantics{\varphi}\cap L_{n+2} = \{x_{n+2}\}$, then by
  \autoref{xCertShape},~$\varphi$ has subformulae
  $\Diamond\varphi_{x_{n+1}}$ and $\Diamond\varphi_{y_{n+1}}$ such that
  \[
    \semantics{\varphi_{x_{n+1}}}\cap L_{n+1} = \{x_{n+1}\}
    \qquad\text{and}\qquad
    \semantics{\varphi_{y_{n+1}}} \cap L_{n+1} = \{y_{n+1}\}.
  \]
  By their semantics, $\varphi_{x_{n+1}}$ and $\varphi_{y_{n+1}}$ are
  necessarily different. Hence,
  \begin{equation}\label{eqFib1}
    |\varphi_{x_{n+2}}|
    \ge
    |\varphi_{x_{n+1}}|
    + |\varphi_{y_{n+1}}|,
  \end{equation}
  where $|\mathord{-}|$ denotes formula size. Now
  \autoref{yCertShape} implies
  that $\varphi_{y_{n+1}}$ has a subformula of the form
  $\Diamond\varphi_{x_n}$ for some $\varphi_{x_n}$ such that
  $\semantics{\varphi_{x_n}}\cap L_n = \{x_n\}$. Hence,
  \begin{equation}\label{eqFib2}
    |\varphi_{y_{n+1}}|
    \ge
    |\varphi_{x_{n}}|.
  \end{equation}
  Thus, we derive
  \[
    |\varphi_{x_{n+2}}|
    \overset{\eqref{eqFib1}}\ge
    |\varphi_{x_{n+1}}| + |\varphi_{y_{n+1}}|
    \overset{\eqref{eqFib2}}\ge
    |\varphi_{x_{n+1}}| + |\varphi_{x_{n}}|
    \qquad\text{for all $n \geq 0$}.
  \]
  Thus, we have $|\varphi_{x_{n}}| \ge \fib(n)$ by induction.%
\end{proof}
\begin{proposition}\label{separatorBigAsCert} Every formula $\varphi$
  distinguishing $x_{i+1}$ from $y_{i+1}$ contains a
  subformula~$\varphi_{x_i}$ such that
  $\semantics{\varphi_{x_i}}\cap L_i = \{x_i\}$.  Hence, every such
  $\varphi$ has at least size $\fib(i)$.
\end{proposition}
\begin{proof}
  Let $\varphi$ distinguish $x_{i+1}$ from $y_{i+1}$, which means that
  $x_{i+1}\in
  \semantics{\varphi}$ and $y_{i+1}\notin
  \semantics{\varphi}$. By \autoref{dnfSubformula}, $\varphi$ has a subformula
  $\alpha$ such that $x_{i+1}\in \semantics{\alpha}$, $y_{i+1}\notin \semantics{\alpha}$, and
  $\alpha$ has one of the forms $\Diamond\psi$ or $\neg\Diamond\psi$ for some formula
  $\psi$.
  Looking at \autoref{tableThreeTower}, we see that the only choice
  for~$\alpha$ is $\Diamond\psi$ for some formula $\psi$
  satisfying $\semantics{\psi}\cap L_i = \{x_i\}$, so that $\psi$ may
  serve as $\varphi_{x_i}$.
\end{proof}

\end{document}